\documentclass[12pt]{extarticle}
\usepackage[T1]{fontenc}
\usepackage{lmodern}
\usepackage[multiple]{footmisc}
\usepackage{mathrsfs}
\usepackage{fullpage}
\usepackage{amssymb}
\usepackage[authoryear]{natbib}

\usepackage{amsmath}
\usepackage{amsthm}
\usepackage{amsfonts}
\usepackage{mathtools}
\usepackage{bm}
\usepackage{array}
\usepackage{xcolor}
\usepackage{latexsym}
\usepackage{tikz}
\usetikzlibrary{positioning,calc}
\usepackage{setspace}
\usepackage{xurl}
\usepackage[title]{appendix}
\usepackage{caption}
\usepackage{mathabx}
\usepackage{enumitem}
\usepackage{needspace}

\let\emptyset\relax
\let\varnothing\relax
\DeclareMathSymbol{\emptyset}{\mathord}{symbols}{"3B}
\DeclareMathSymbol{\varnothing}{\mathord}{AMSb}{"3F}
\usepackage[hidelinks]{hyperref}
\hypersetup{
  pdftitle={Serial Top Trading Cycles and the Exclusion Core in Mixed-Ownership Object Allocation Problems},
  pdfauthor={Bettina Klaus, Yang Chen, Anh Tri\^eu},
  pdfsubject={Object allocation with mixed ownership and exclusion core stability},
  pdfkeywords={exclusion core stability, serial top trading cycles, complex endowments, mixed-ownership object allocation problems, von Neumann--Morgenstern stability, path to exclusion core stability}
}

\newtheorem{theorem}{Theorem}
\newtheorem{prop}{Proposition}
\newtheorem{lemm}{Lemma}
\newtheorem{corollary}{Corollary}
\theoremstyle{definition}
\newtheorem{claim}{Claim}
\newtheorem{exam}{Example}
\newtheorem{remark}{Remark}
\newtheorem{defi}{Definition}

\makeatletter
\renewcommand{\paragraph}{\@startsection{paragraph}{4}{\z@}%
  {4.5\p@ \@plus1.5\p@ \@minus1\p@}{-1em}{\normalfont\normalsize\bfseries}}
\renewcommand{\subparagraph}{\@startsection{subparagraph}{5}{\z@}%
  {4.5\p@ \@plus1.5\p@ \@minus1\p@}{-1em}{\normalfont\normalsize\itshape}}
\newcommand{\proofcase}[1]{\par\addvspace{4.5\p@ \@plus1.5\p@ \@minus1\p@}\noindent\textup{#1}\ }
\makeatother

\newcommand{\SC}{\mathrm{SC}}
\newcommand{\WC}{\mathrm{WC}}
\newcommand{\EC}{\mathrm{EC}}
\newcommand{\STTC}{\mathrm{STTC}}

\title{Serial Top Trading Cycles and the Exclusion Core in Mixed-Ownership Object Allocation Problems\thanks{Y.\ Chen and B.\ Klaus  gratefully acknowledge support from the Swiss National Science Foundation (SNSF) through Project 100018-231777. We thank Jacob Coreno for helpful comments and discussions, for Example~\ref{exam:noexstable}, and for suggesting the characterization in Proposition~\ref{prop.acyclicchar}. We also thank participants at the 18th Meeting of the Society for Social Choice and Welfare (SSCW 2026), European Meeting on Game Theory 2026 (SING21), 8th International Workshop on Matching Under Preferences (MATCH-UP 2026), and seminars at Maastricht University for helpful comments. During the preparation of this work the authors used Claude (Anthropic) and ChatGPT (OpenAI) in order to check proofs and examples ex post. After using these tools, the authors carefully checked and edited all content themselves and take full responsibility for the content of this paper.}}

\author{Bettina Klaus\thanks{\textit{Corresponding author.} Faculty of Business and Economics, University of Lausanne, Internef 538, 1015 Lausanne, Switzerland; e-mail: \href{mailto:bettina.klaus@unil.ch}{\tt bettina.klaus@unil.ch}.} \and  Yang Chen\thanks{Faculty of Business and Economics, University of Lausanne, CH-1015 Lausanne, Switzerland; e-mail: \href{mailto:yang.chen@unil.ch}{\tt yang.chen@unil.ch}.} \and Anh Tri\^eu\thanks{Department of Quantitative Economics, Maastricht University, P.O. Box 616, 6200 MD Maastricht, The Netherlands; e-mail: \href{mailto:anh.trieu@maastrichtuniversity.nl}{\tt anh.trieu@maastrichtuniversity.nl}.}}

\date{November 30, 2026}

\begin{document}

\begingroup\linespread{1.185}\selectfont 
\maketitle

\begin{abstract}
\linespread{1.05}\selectfont
\noindent
We study \citeauthor{balbuzanov2019endowments}'s (\citeyear{balbuzanov2019endowments}) exclusion core for object allocation problems with complex endowments. We focus on mixed-ownership object allocation problems in which each object is either individually owned by one agent or publicly owned by all agents. Given a priority order over agents, an endowment structure, an initial allocation, and agents' preferences, we construct serial top trading cycles (STTC) graphs and define the associated STTC algorithm. The algorithm is the central tool for our analysis of exclusion-core-stable allocations.

First, we show that, among domains defined by restrictions on endowment structures, the mixed-ownership domain is the largest on which the STTC algorithms characterize the exclusion core. Second, we provide a simple test for exclusion core stability: an allocation is exclusion core stable if and only if it is a fixed point of the STTC algorithm initialized at that allocation, independently of the priority order. Third, we show that the exclusion core is a von Neumann--Morgenstern (\citeyear{vonneumann1944theory}) stable set with respect to exclusion-domination. Finally, we study a dynamic exclusion-blocking process and show that, from any allocation, there exists a self-contained exclusion-blocking path to the exclusion core. Von Neumann--Morgenstern stability and the existence of such paths extend to the acyclic endowment structures of \citeauthor{balbuzanov2019endowments}, but beyond them both can fail simultaneously.\par\smallskip

\par\smallskip

\enlargethispage{2\baselineskip}
\end{abstract}
\endgroup
\newpage

\begingroup\linespread{1.2400}\selectfont 
\section{Introduction}

Object allocation problems with indivisible goods arise in many economic settings, including housing markets, school choice, and the assignment of positions or resources when monetary transfers are infeasible or undesirable. Two benchmark models structure much of this literature. In \textbf{Shapley--Scarf housing markets} \citep{SHAPLEY197423}, each agent initially owns exactly one object, and ownership is purely private. In \textbf{house allocation problems} \citep{hylland1979}, all objects are publicly owned. \textbf{House allocation problems with existing tenants} \citep{abdulkadirouglu1999house} lie between these two models: some agents initially own exactly one object, the remaining objects are publicly owned, and other agents own none. These benchmark models are important, but many economically relevant settings involve more complex endowment structures.\par\smallskip

In such environments, an agent may own several objects or none, and several agents may jointly own an object. Despite this flexibility in endowments, each agent ultimately receives at most one object. \citet{balbuzanov2019endowments} call these problems ``\textbf{simple economies}''; we call them \textbf{object allocation problems with complex endowments}. Ownership may be private, public, or shared among agents.\par\smallskip

Standard core concepts are grounded in the interpretation that owners can consume or trade their endowed objects.
Beyond the benchmark models, the cores built on this interpretation lose their appeal:
the strong core may be empty and the weak core may contain wasteful allocations. \citeauthor{balbuzanov2019endowments} instead interpret ownership as a system of exclusion rights: owners can exclude other agents from an object, either directly or through chains of dependence. Their \textbf{exclusion core} restricts blocking to deviations that are feasible through these rights. It is always nonempty and contains only Pareto efficient allocations.\par\smallskip

Related recent work has further developed the analysis of object allocation with complex endowments. \citet{balbuzanov2024property} extend the exclusion core to economies with production and generalize the top trading cycles algorithm to that setting. \citet{SUN2026106135} introduce two extensions of the strong core, the effective core and a refinement of it called the rectified core, and show that the effective core neither contains nor is contained in the exclusion core. \citet{zhang2026consistent} applies the consistency principle to core concepts in this setting and introduces the refined exclusion core, a consistent refinement that admits unaffected agents into blocking coalitions under restrictions. \citet{ISHIDA2026287} study partitioned ownership and its implications for traditional core concepts and the exclusion core.\par\smallskip

We focus on \textbf{mixed-ownership object allocation problems}, a subclass of object allocation problems with complex endowments. In this subclass, each object is either individually owned by one agent or publicly owned by all agents. Proper co-ownership is ruled out, but an agent may own several objects. This class includes all three benchmark models described above. Importantly, mixed ownership does not eliminate the difficulties identified by \citet{balbuzanov2019endowments}: even without proper co-ownership, the strong core may still be empty and the weak core may still contain wasteful allocations (Example~\ref{exam:weak-strongcore}). The exclusion core therefore remains an appropriate stability concept.\par\smallskip

To analyze exclusion core stability in mixed-ownership object allocation problems, we introduce the \textbf{serial top trading cycles (STTC) algorithm}. Given a priority order, an endowment structure, an initial allocation, and agents' preferences, it constructs a sequence of directed graphs and executes top trading cycles sequentially. Unlike the generalized top trading cycles (GTTC) algorithm of \citet{balbuzanov2019endowments}, it allows for an arbitrary initial allocation and uses a single priority order over agents rather than object-specific priority structures. These features make the STTC algorithm the central tool for both generating exclusion-core-stable allocations and analyzing stability from an arbitrary allocation.\par\smallskip

Two preliminary results support the main analysis. First, every allocation produced by the STTC algorithm is exclusion core stable (Proposition~\ref{prop.sttc-exclusion}). Second, every exclusion-core-stable allocation can be obtained by running the STTC algorithm from the null allocation under a suitable priority order (Proposition~\ref{prop.corecharacterization}). Proposition~\ref{prop.corecharacterization} is closely related to \citet[][Theorem~4]{balbuzanov2019endowments} and \citet[][Proposition~5]{zhang2026consistent}.%
\par\smallskip

We then establish four main results. First, when a properly co-owned object is present, the STTC allocations may form a proper subset of the exclusion core (Theorem~\ref{thm.strictsuperset}). It follows that, among domains defined by restrictions on endowment structures, the mixed-ownership domain is the largest on which the STTC algorithms characterize the exclusion core (Corollary~\ref{cor:maxdomain}). This maximality result explains our focus on mixed ownership.
Second, we provide an exclusion-core-stability test (Theorem~\ref{theorem:coretest}). An allocation is exclusion core stable if and only if it is the outcome of the STTC algorithm initialized at that allocation, independently of the priority order.
Third, we show that the exclusion core is a von Neumann--Morgenstern (\citeyear{vonneumann1944theory}) stable set with respect to exclusion-domination (Theorem~\ref{theorem:vNMstability}).
Fourth, we capture the dynamic aspect of exclusion-blocking: from any allocation, there exists a self-contained exclusion-blocking path to the exclusion core (Theorem~\ref{theorem:pathtoexclusion}). As a consequence, a random process whose transition probabilities depend only on the current allocation and assign positive probability to every available post-blocking allocation reaches an exclusion-core-stable allocation with probability one (Corollary~\ref{coro:randompath}).\par\smallskip

We next clarify how our results relate to those of \citet{balbuzanov2019endowments}. Remark~\ref{remark1} explains how Proposition~\ref{prop.sttc-exclusion} extends when the common priority order is replaced by arbitrary edge choices. Remark~\ref{remark2} extends Proposition~\ref{prop.sttc-exclusion} to the full class of object allocation problems with complex endowments: every allocation generated by an owner-respecting flexible TTC algorithm, in which an object points to one of its owners whenever an owner remains, is exclusion core stable. For algorithms in which each object's edge is determined by a fixed priority order over agents, this also follows from Corollary~1(a) and Lemma~3(a) of \citet{balbuzanov2019endowments}; flexible TTC algorithms are more general, since an object may point to different remaining owners at different steps. Proposition~4 of \citet{balbuzanov2019endowments} establishes von Neumann--Morgenstern stability for a subclass of these problems. Our Theorem~\ref{theorem:vNMstability} is closely related, but we state and prove it separately because its constructive proof is intertwined with the path result of Section~\ref{subsec.paths}.\par\smallskip

For mixed-ownership problems, we prove von Neumann--Morgenstern stability jointly with the path-to-exclusion-core-stability result. The constructive proof links the myopic, step-by-step self-contained exclusion-blockings forming the path to the farsighted external-stability blocking by the final STTC allocation.\par\smallskip

Our results differ in how far they extend beyond mixed ownership. Exclusion core stability of STTC allocations is the most robust: it holds for the full class of object allocation problems with complex endowments (Remark~\ref{remark2}). The exclusion core characterization and the fixed-point test are the least robust: a single properly co-owned object is enough for both to fail.
The construction proving Theorem~\ref{thm.strictsuperset} yields an exclusion-core-stable allocation that is a fixed point of no owner-respecting flexible TTC algorithm.\par\smallskip

Internal stability holds by definition for every endowment structure; hence, for von Neumann--Morgenstern stability, only external stability is at issue. External stability of the exclusion core (Theorem~\ref{theorem:vNMstability}) and the existence of self-contained exclusion-blocking paths to the exclusion core (Theorem~\ref{theorem:pathtoexclusion}), however, reach beyond mixed ownership. Appendix~\ref{appendix:acyclic} shows that an endowment structure is acyclic in the sense of \citet{balbuzanov2019endowments} exactly when either no object is properly co-owned or all objects share a single set of owners (Proposition~\ref{prop.acyclicchar}). When all objects share a single set of owners, the exclusion core is again externally stable and such paths again exist (Theorems~\ref{thm.uniformvNM} and~\ref{thm.uniformpath}). Hence, both properties hold for every acyclic endowment structure (Corollaries~\ref{cor.acyclicvNM} and~\ref{cor.acyclicpath}). Beyond the acyclic class, external stability and the existence of such paths may both fail (Example~\ref{exam:noexstable}), although acyclicity is necessary for neither property (Example~\ref{exam:beyondacyclic}). Figure~\ref{fig:endowmentclasses} in Section~\ref{sec.conclude} summarizes how far each result extends. Since the only acyclic endowment structures beyond mixed ownership are those in which a single group co-owns every object, these robustness results, together with Corollary~\ref{cor:maxdomain}, support our focus on mixed ownership.\par\smallskip

The remainder of the paper is organized as follows. Section~\ref{sec.pre} introduces the model and the relevant core concepts. Section~\ref{sec.sttc} defines the STTC algorithm and establishes its basic relationship with the exclusion core. Section~\ref{sec.mainresult} presents our main results, and Section~\ref{sec.conclude} contains concluding remarks. Appendices~\ref{appendix:prelim}--\ref{appendix:main} contain the proofs of the results in Sections~\ref{sec.pre}--\ref{sec.mainresult}. Appendix~\ref{appendix:STTCexample} illustrates, in a detailed example, the STTC algorithm, Proposition~\ref{prop.corecharacterization}, and Theorems~\ref{theorem:coretest}--\ref{theorem:pathtoexclusion}. Appendix~\ref{appendix:acyclic} studies acyclic endowment structures and the two complex-endowment examples discussed above.

\endgroup
\newpage
\section{Preliminaries}\label{sec.pre}

\subsection{Model}\label{subsec.model}

We consider a finite \textbf{set of agents} $\bm{N = \{1, \ldots, n\}}$ and a finite \textbf{set of objects} $\bm{O = \{o_1, \ldots, o_m\}}$. The numbers of agents and objects may differ. \textbf{Ownership is mixed}: each object is either individually owned by exactly one agent or publicly owned by all agents. For each agent $i \in N$, we denote the set of objects he individually owns by $\omega_i$. The list $\bm{\omega=(\omega_1,\ldots,\omega_n)}$ then describes the \textbf{endowment structure}. We denote the \textbf{set of individually owned objects} by $\bm{O_I=\bigcup_{i\in N}\omega_i}$ and the set of \textbf{public objects} by $\bm{O_P}=O \setminus O_I$.\par\smallskip

A nonempty subset $S \subseteq N$ is called a \textbf{coalition}. For each coalition $S\subsetneq N$, the \textbf{set of objects collectively owned by $\bm{S}$} is denoted by $\bm{\omega(S)} = \bigcup_{i\in S}\omega_i \subseteq O_I$. For the grand coalition, we define $\bm{\omega(N)}:=O$. Thus, public ownership is represented collectively rather than through the individual sets $\omega_i$: public objects do not belong to any $\omega_i$, but they belong to the endowment of the grand coalition. To illustrate the endowment structure, consider the following mixed-ownership situation with three agents $N =\{1, 2, 3\}$ and four objects $O=\{a, b, c, d\}$. If agent~$1$ owns objects $a$ and $b$, agent~$2$ owns object $c$, and agent~$3$ owns no object, then $\omega_1 = \{a,b\}$, $\omega_2 = \{c\}$, $\omega_3 = \emptyset$, $\omega(\{1,2\}) = \{a,b,c\}$, $\omega(\{1,3\})=\{a,b\}$, $\omega(\{2,3\}) = \{c\}$, and $\omega(\{1,2,3\}) = \{a,b,c,d\}$.\par\smallskip

Each agent has an outside option, modeled as the \textbf{null object} $\bm{\varnothing}$,\footnote{Throughout, $\varnothing$ denotes the null object, whereas $\emptyset$ denotes the empty set.} which may literally represent ``receiving nothing.'' The null object has enough capacity for all agents. When we need to distinguish the elements of $O$ from the null object, we call them \textbf{real objects}.\par\smallskip

\citet{balbuzanov2019endowments} study a more general class of object allocation problems. They call these problems ``simple economies''; we call them \textbf{object allocation problems with complex endowments}. In their model, objects may be co-owned by multiple agents.\footnote{\citet{balbuzanov2019endowments} consider the following more general model. For each object $o \in O$, $S^o \subseteq N$ denotes the set of owners of object $o$. If an object $o\in O$ is individually owned by agent $i\in N$, then $S^{o} = \{i\}$, and if it is publicly owned, then $S^{o} = N$. \citet{balbuzanov2019endowments}, however, also allow proper co-ownership, that is, ownership of an object by a coalition $S^o$ with $1<|S^o|<n$, which is not allowed in our model. Because this broader model permits complex endowment structures, we use the term ``object allocation problems with complex endowments'' for the problems introduced by \citet{balbuzanov2019endowments}.\label{footnote:BKmodel1}} Our mixed-ownership model is the subclass in which each object is either individually owned by exactly one agent or publicly owned by all agents.\footnote{In one result, \citet[][Lemma~3(b)]{balbuzanov2019endowments} consider the same model as ours, referring to it as a ``simple economy where every house is either privately or publicly owned.''}\par\smallskip

We assume that each agent $i\in N$ has a preference relation $R_i$ over the set of objects and the null object $O \cup \{\varnothing\}$, which is \textit{complete}, \textit{transitive}, and \textit{antisymmetric}. For an agent $i \in N$ and two objects $a,b \in O \cup \{\varnothing\}$, we write $a \mathbin{R_i} b$ if agent $i$ weakly prefers object $a$ to object $b$, and $a \mathbin{P_i} b$ if agent $i$ strictly prefers object $a$ to object $b$. Thus, $P_i$ is the strict part of $R_i$, and, by antisymmetry, $a \mathbin{R_i} b$ if and only if $a \mathbin{P_i} b$ or $a=b$. Let $R:=(R_i)_{i\in N}$ denote a preference profile. A \textbf{mixed-ownership object allocation problem} (henceforth, a \textbf{problem}) is specified by a quadruple $(N,O,\omega,R)$. Throughout, the null object $\varnothing$ is implicitly available in every problem. Whenever $(N,O,\omega)$ is fixed, we simply denote a problem by $R$.\par\smallskip

We assume that each agent demands exactly one object in $O \cup \{\varnothing\}$. An \textbf{allocation} $\mu: N\to O\cup\{\varnothing\}$ assigns each agent $i\in N$ an \textbf{allotment} $\mu(i)\in O\cup\{\varnothing\}$, and assigns each object $o\in O$ to at most one agent. We also say that agent $i$ \textbf{occupies} his allotment $\mu(i)$ at $\mu$. Equivalently, we write an allocation as a list of its allotments, $\mu=(\mu_1,\ldots,\mu_n)$, where the $i$-th entry is agent~$i$'s allotment. We denote the set of allotments of the members of a coalition $S$ at allocation $\mu $ by $\mu(S) = \bigcup_{i\in S}\{\mu(i)\}$.\par\smallskip

Our mixed-ownership model includes several classical object allocation models as special cases:
 \setlist{nolistsep}
\begin{itemize}[noitemsep]
\item when each agent owns exactly one object and the number of agents equals the number of objects, our model reduces to the \textbf{Shapley--Scarf housing market model} \citep{SHAPLEY197423};
\item when all objects are public, our model reduces to the \textbf{house allocation model} \citep{hylland1979};
\item when some agents own exactly one object and all remaining objects are public, our model reduces to the \textbf{house allocation model with existing tenants} \citep{abdulkadirouglu1999house}.
\end{itemize}

\subsection{Core stability notions}\label{subsec.core}

We begin with the classical notions of weak blocking and strong blocking. They define the strong core and the weak core, respectively. This distinction and its consequences for markets with indivisible goods go back to \citet{ROTH1977131}. These definitions are grounded in the interpretation that owners can freely \textbf{consume} or \textbf{trade} their endowed objects. In line with this interpretation, we say that a coalition \textbf{$\bm{S}$ is trade-enforcing at allocation $\bm{\sigma}$} if $\sigma(S) \subseteq \omega(S) \cup \{\varnothing\}$.

\begin{defi}[\textbf{Strong core stability}]
Given a problem $R$, allocations $\mu$ and $\sigma$, and a coalition $S$, we say that \textbf{$\bm{S}$ weakly blocks $\bm{\mu}$ via $\bm{\sigma}$} if
    \begin{enumerate}
        \item coalition $S$ is trade-enforcing at $\sigma$; and
        \item each agent in $S$ is weakly better off at $\sigma$, and at least one agent in $S$ is strictly better off at $\sigma$; i.e., for each $i \in S$, $\sigma(i) \mathbin{R_i} \mu(i)$, and, for some $j \in S$, $\sigma(j) \mathbin{P_j} \mu(j)$.
    \end{enumerate}
The \textbf{strong core} for a problem $R$, denoted by $\SC(R)$, consists of all allocations that cannot be weakly blocked by any coalition. An allocation in the strong core is \textbf{strongly core stable}.
\end{defi}

\begin{defi}[\textbf{Weak core stability}]
   Given a problem $R$, allocations $\mu$ and $\sigma$, and a coalition $S$, we say that \textbf{$\bm{S}$ strongly blocks $\bm{\mu}$ via $\bm{\sigma}$} if
    \begin{enumerate}
        \item coalition $S$ is trade-enforcing at $\sigma$; and
        \item each agent in $S$ is strictly better off at $\sigma$, i.e., for each $i \in S$, $\sigma(i) \mathbin{P_i} \mu(i)$.
    \end{enumerate}
The \textbf{weak core} for a problem $R$, denoted by $\WC(R)$, consists of all allocations that cannot be strongly blocked by any coalition. An allocation in the weak core is \textbf{weakly core stable}.
\end{defi}

Strong core allocations also satisfy the following classical efficiency criterion.

\begin{defi}[\textbf{Pareto efficiency}]
Given a problem $R$ and allocations $\mu$ and $\sigma$, we say that \textbf{$\bm{\sigma}$ Pareto dominates $\bm{\mu}$} if each agent in $N$ is weakly better off at $\sigma$ and at least one agent in $N$ is strictly better off at $\sigma$; i.e., for each $i\in N$, $\sigma(i)\mathbin{R_i}\mu(i)$, and, for some $j\in N$, $\sigma(j)\mathbin{P_j}\mu(j)$.
An allocation $\mu$ is \textbf{Pareto efficient} if no allocation $\sigma$ Pareto dominates it.
\end{defi}

We now introduce exclusion-blocking, following \citet{balbuzanov2019endowments}. The central idea is to interpret ownership as a system of exclusion rights rather than solely as a right to consume or trade endowed objects. Ownership allows a coalition to evict other agents from objects it owns and therefore controls.\par\smallskip

For each coalition $S$, we call $\bm{\omega(S)}$ its \textbf{direct eviction set}.\footnote{In the mixed-ownership model, $\omega(S)$ is defined in Section~\ref{subsec.model}. For object allocation problems with complex endowments, Footnote~\ref{footnote:BKmodel1} defines $S^o$ as the set of owners of object $o$ and $\omega(S)=\{o\in O\mid S^o\subseteq S\}$.} It is the set of objects over which coalition $S$ has direct exclusion rights. By exercising these rights, coalition $S$ can directly evict the current occupants of these objects. The set $\omega(S)$ depends only on the endowment structure; the prevailing allocation determines which agents occupy these objects.\par\smallskip

By exercising its direct exclusion rights, a coalition may become able to exercise further exclusion rights. Suppose that coalition $S$ can evict agent $j$ from his allotment at allocation $\mu$. Agent $j$ then depends on continued accommodation by $S$. In the exclusion-core logic, coalition $S$ can therefore induce agent $j$ to exercise the exclusion rights associated with his endowment. This may allow the coalition to evict further agents and thereby gain access to their exclusion rights. We formalize this chain of dependence recursively.\par\smallskip

Fix an allocation $\mu$ and a coalition $S$. For each set of objects $A\subseteq O$, let $\mu^{-1}(A)$ denote the set of agents whose allotment at $\mu$ belongs to $A$, that is, the agents occupying the objects in $A$: \(\mu^{-1}(A):=\{i\in N\mid \mu(i)\in A\}.\)\par\smallskip

For each $k\geq 0$, the set $S^\mu_k$ will consist of the original coalition and all agents reached through at most $k$ rounds of eviction. Initially, \(S^\mu_0:=S.\)
If the agents in $S^\mu_k$ have been reached, the coalition controls the objects in $\omega(S^\mu_k)$ and can therefore reach the agents currently occupying those objects. Accordingly, define \(S^\mu_{k+1}:=S^\mu_k\cup\mu^{-1}\bigl(\omega(S^\mu_k)\bigr).\)\par\smallskip

The sequence \(S^\mu_0\subseteq S^\mu_1\subseteq S^\mu_2\subseteq\cdots\) is increasing. Since $N$ is finite, it reaches a fixed point. Let $\ell$ be the smallest integer such that \(S^\mu_{\ell+1}=S^\mu_\ell.\)
The \textbf{extended eviction set of coalition $\bm{S}$ at allocation $\bm{\mu}$} is \(\Omega(S\mid\mu):=\omega(S^\mu_\ell).\)
Thus, $\Omega(S\mid\mu)$ contains all objects from which coalition $S$ can evict agents either directly or through a finite chain of evictions.\par\smallskip

Endowments interpreted as exclusion rights give rise to the following counterpart of trade-enforcement. Given allocations $\mu$ and $\sigma$, coalition \textbf{$\bm{S}$ exclusion-enforces $\bm{\sigma}$ at $\bm{\mu}$} if every agent outside $S$ who is worse off at $\sigma$ occupies at $\mu$ an object in the extended eviction set of $S$; formally, for each $j\notin S$, $\mu(j)\mathbin{P_j}\sigma(j)$ implies $\mu(j)\in\Omega(S\mid\mu)$.
Every outsider harmed by the move to $\sigma$ must therefore be removable through exclusion rights available to the coalition, either directly or through the recursive chain of dependence.

\begin{defi}[\textbf{Exclusion core stability}]\label{defi.exclusioncore}
    Given a problem $R$, allocations $\mu$ and $\sigma$, and a coalition $S$, we say that \textbf{$\bm{S}$ exclusion-blocks $\bm{\mu}$ via $\bm{\sigma}$} if
    \begin{enumerate}
        \item coalition $S$ exclusion-enforces $\sigma$ at $\mu$; and
        \item each agent in $S$ is strictly better off at $\sigma$, i.e., for each $i \in S$, $\sigma(i) \mathbin{P_i} \mu(i)$.
    \end{enumerate}
The \textbf{exclusion core} for a problem $R$, denoted by $\EC(R)$, consists of all allocations that cannot be exclusion-blocked by any coalition. An allocation in the exclusion core is \textbf{exclusion core stable}.
\end{defi}

Four properties of extended eviction sets recur throughout the analysis: inclusion of the direct eviction set, coalitional monotonicity, closure under the endowments of agents reached through eviction, and propagation along chains of dependence.

\begin{lemm}[\textbf{Properties of extended eviction sets}]\label{lemm.evictionproperties}
Let $R$ be a problem, let $\mu$ be an allocation, and let $S$ and $S'$ be coalitions.
\begin{enumerate}
\item[\textup{(a)}] $\omega(S)\subseteq\Omega(S\mid\mu)$.
\item[\textup{(b)}] \textbf{\textup{Monotonicity.}} If $S\subseteq S'$, then $\Omega(S\mid\mu)\subseteq\Omega(S'\mid\mu)$.
\item[\textup{(c)}] \textbf{\textup{Closure.}} For an agent $i\in N$, if $\mu(i)\in\Omega(S\mid\mu)$, then $\omega_i\subseteq\Omega(S\mid\mu)$.
\item[\textup{(d)}] \textbf{\textup{Chains of dependence.}} Let $q\geq1$, and consider agents $i^1,\ldots,i^q$ and real objects $o^1,\ldots,o^q$ that form a \textbf{chain of dependence at $\bm{\mu}$}: for each $r\in\{1,\ldots,q\}$, agent $i^r$ owns object $o^r$ (that is, $o^r\in\omega_{i^r}$), and, for each $r\in\{1,\ldots,q-1\}$, agent $i^r$ occupies object $o^{r+1}$ (that is, $\mu(i^r)=o^{r+1}$). Thus, ownership and occupancy alternate along the chain. We depict such a chain by\nopagebreak\vspace{-0.2cm}
\[ o^1\longrightarrow i^1\longrightarrow o^2\longrightarrow i^2 \longrightarrow\cdots\longrightarrow o^q\longrightarrow i^q,\vspace{-0.2cm} \]
where an arrow from an object to an agent records that the agent owns the object, and an arrow from an agent to an object records that the object is his allotment at $\mu$. If coalition $S$ reaches agent $i^q$, in the sense that $i^q\in S$ or $\omega_{i^q}\subseteq\Omega(S\mid\mu)$, then $o^1\in\Omega(S\mid\mu)$.
\end{enumerate}
\end{lemm}

Appendix~\ref{appendix:prelim} contains the proof of Lemma~\ref{lemm.evictionproperties}.\par\smallskip

The following example illustrates extended eviction sets and exclusion-blocking.\par\smallskip

In the preference tables and displayed allocations in the examples below, boldface marks each agent's allotment at the reference allocation, while underlining marks his allotment at the alternative allocation. An allotment shown in both formats is unchanged.

\begin{exam}[\textbf{Exclusion-blocking}]\label{exam:extendedendowment}
    Consider a Shapley--Scarf housing market with three agents $N =  \{1, 2, 3\}$ and three objects $O = \{a, b, c\}$. Agents $1$, $2$, and $3$ own objects $a$, $b$, and $c$, respectively. Preferences are as follows.\nopagebreak\vspace{-0.2cm}
      \[\begin{array}{ccccc}
          R_{1}: & \underline{a} & \bm{c} & b & \varnothing\\
          R_{2}: & \bm{a} & \underline{b} & c & \varnothing\\
          R_{3}: & \bm{b} & a & \underline{c} & \varnothing
      \end{array}\vspace{-0.2cm}\]

Consider allocation \(\bm{\mu = (c,a,b)}\) and the singleton coalition $S=\{1\}$. Initially, $S^\mu_0=S=\{1\}$ and $\omega(S^\mu_0)=\{a\}$.
Since agent~$2$ receives object $a$ at $\mu$, \(S^\mu_1 = S^\mu_0\cup\mu^{-1}\bigl(\omega(S^\mu_0)\bigr) = \{1,2\}.\)
Together, agents~$1$ and~$2$ own objects $a$ and $b$, and therefore \(\omega(S^\mu_1)=\{a,b\}.\)
Since agent~$3$ receives object $b$ at $\mu$, \(S^\mu_2 = S^\mu_1\cup\mu^{-1}\bigl(\omega(S^\mu_1)\bigr) = \{1,2,3\}=N.\)
Hence, $\omega(S^\mu_2)=O=\{a,b,c\}$ and $S^\mu_3=S^\mu_2$. It follows that  \(\Omega(\{1\}\mid\mu)=\{a,b,c\}.\)\par\smallskip

What the recursion traces here is a chain of dependence in the sense of Lemma~\ref{lemm.evictionproperties}(d): agent~$1$ owns object~$a$, which agent~$2$ occupies; agent~$2$ owns object~$b$, which agent~$3$ occupies; and agent~$3$ owns object~$c$. In the notation of that lemma,
\(c\longrightarrow 3\longrightarrow b\longrightarrow 2\longrightarrow a\longrightarrow 1,\)
and coalition $\{1\}$ reaches object~$c$ by way of agents~$2$ and~$3$. As a result, agent~$1$ can exclusion-block $\mu$ via allocation \(\underline{\sigma} = (\underline{a},\underline{b},\underline{c}).\) Agent~$1$ is strictly better off at $\sigma$, and agents~$2$ and~$3$, who are worse off, occupy at $\mu$ the objects $a$ and $b$, which belong to $\Omega(\{1\}\mid\mu)$.\hfill~$\diamond$
\end{exam}

Exclusion-blocking includes Pareto improvements as a special case, because a move that makes no agent worse off requires no exclusion rights.

\begin{lemm}[\textbf{Pareto efficiency and exclusion-blocking}] Let $R$ be a problem. If $\sigma$ Pareto dominates $\mu$, then the coalition of agents who are strictly better off at $\sigma$ than at $\mu$ exclusion-blocks $\mu$ via $\sigma$. Hence, if an allocation is exclusion core stable, then it is Pareto efficient.\label{lemm.paretoexclusion}
\end{lemm}

Appendix~\ref{appendix:prelim} contains the proof of Lemma~\ref{lemm.paretoexclusion}.

\subsection{Relationships between core concepts}\label{subsec.relationcores}

We now describe the relationships among the traditional core concepts and the exclusion core for mixed-ownership object allocation problems. The strong core is a refinement of the weak core, but it may be empty; the weak core is always nonempty, but it may contain wasteful allocations (see Example~\ref{exam:weak-strongcore}). The exclusion core is always nonempty and consists only of Pareto efficient allocations. Moreover, the exclusion core is a refinement of the weak core \citep[Lemma~1]{balbuzanov2019endowments}: \(\EC(R)\subseteq \WC(R).\)
However, neither the strong core nor the exclusion core contains the other. Example~\ref{exam:weak-strongcore} shows that the exclusion core need not be contained in the strong core, whereas Example~\ref{exam:logicalindependent} shows that the strong core need not be contained in the exclusion core.\par\smallskip

The two sets may nevertheless coincide: in Example~\ref{exam:extendedendowment}, a Shapley--Scarf housing market, the strong core and the exclusion core both consist of the single allocation assigning object~$a$ to agent~$1$, object~$b$ to agent~$2$, and object~$c$ to agent~$3$.\footnote{In a housing market with strict preferences, the strong core consists of the unique top trading cycles allocation \citep{ROTH1977131}, which is the allocation just described. That the exclusion core is the same singleton follows from Proposition~\ref{prop.corecharacterization}.} These relationships are depicted in Figure~\ref{fig:coresexclusioncore}.\vspace{-0.2cm}

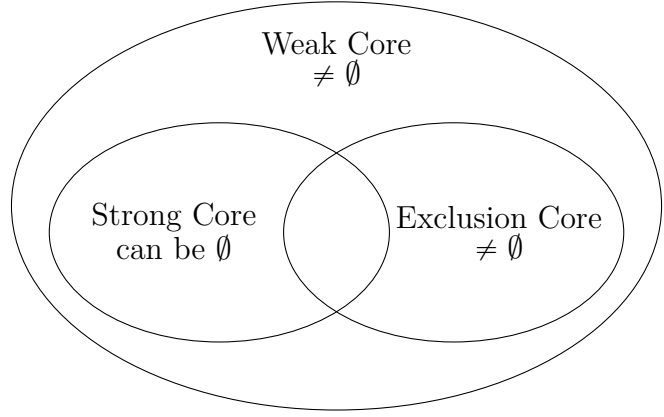
\begin{figure}[!htbp]
    \centering
    \begin{tikzpicture}
        \draw (0,0) ellipse (4.3 and 2.7);
        \node at (0,2.15) {Weak Core};
        \node at (0,1.75) {$\neq\emptyset$};

        \draw (-1.55,-0.35) ellipse (2.25 and 1.45);
        \node at (-2.15,-0.15) {Strong Core};
        \node at (-2.15,-0.55) {can be $\emptyset$};

        \draw (1.55,-0.35) ellipse (2.25 and 1.45);
        \node at (2.15,-0.15) {Exclusion Core};
        \node at (2.15,-0.55) {$\neq\emptyset$};
    \end{tikzpicture}
    \caption{The relationships among the weak core, the strong core, and the exclusion core for mixed-ownership object allocation problems.}
    \label{fig:coresexclusioncore}\vspace{-0.2cm}
\end{figure}

The following example, from \cite{balbuzanov2019endowments}, illustrates that the strong core can be empty and that the exclusion core can be a strict subset of the weak core.

\begin{exam}[\textbf{An empty strong core and an exclusion core strictly contained in the weak core}] \label{exam:weak-strongcore}
    Consider a problem with three agents $N=\{1,2,3\}$ and two objects $O=\{a,b\}$. Assume that agent~$1$ owns both objects $a$ and $b$. Preferences are as follows.\nopagebreak\vspace{-0.2cm}
        \[\begin{array}{cccc}
          R_{1}: & \bm{a} & b & \varnothing\\
          R_{2}: & a & b & \bm{\varnothing} \\
          R_{3}: & a & b & \bm{\varnothing}
      \end{array}\vspace{-0.2cm}\]

Consider allocation \(\bm{\mu = (a, \varnothing, \varnothing )}.\)
Any coalition that trade-enforces a different allocation must include agent~$1$. Since agent~$1$ receives his most-preferred object at $\mu$, no coalition can strongly block $\mu$. Thus, $\mu$ is a weak core allocation. However, since object $b$ is not assigned, $\mu$ is wasteful.\par\smallskip

The strong core in this market is empty. At any strong core allocation, agent~$1$ must receive object $a$; otherwise, the singleton coalition $\{1\}$ weakly blocks by assigning object $a$ to agent~$1$. Pareto efficiency then requires object $b$ to be assigned to some agent $i\in\{2,3\}$. Let $j$ be the other agent in $\{2,3\}$. Coalition $\{1,j\}$ can trade-enforce an allocation at which agent~$1$ keeps object $a$ and agent~$j$ receives object $b$. Agent~$1$ is indifferent and agent~$j$ is strictly better off, so this coalition weakly blocks the original allocation.\par\smallskip

The exclusion core in this market consists of two allocations $\mu^1$ and $\mu^2$ where agent~$1$ receives object $a$, and object $b$ is assigned to either $2$ or $3$: \(\mu^1=(a,b,\varnothing)\) and \(\mu^2=(a,\varnothing,b)\).
Both $\mu^1$ and $\mu^2$ are weak core allocations, while the allocation $\mu=(a,\varnothing,\varnothing)$ considered above is a weak core allocation that does not belong to the exclusion core. Hence, the exclusion core is a strict subset of the weak core.\hfill~$\diamond$
\end{exam}

The following example, also from \cite{balbuzanov2019endowments}, shows that the strong core need not be contained in the exclusion core.

\begin{exam}[\textbf{A strong core allocation that is not exclusion core stable}] \label{exam:logicalindependent}
Consider a problem with four agents $N=\{1,2,3,4\}$ and four objects $O=\{a,b,c,d\}$. Agents $1$, $2$, and $3$ own objects $a$, $b$, and $c$, respectively; $d$ is publicly owned. Preferences are as follows.\nopagebreak\vspace{-0.2cm}
        \[\begin{array}{cccccc}
          R_{1}: & \underline{b} & \bm{a} & \varnothing & \ldots\\
          R_{2}: & \underline{d} & \bm{c} & b & \varnothing & a\\
          R_{3}: & \bm{b}& \underline{c} & \varnothing & \ldots\\
          R_{4}: & \underline{a} & \bm{d} & c & \varnothing & b.
      \end{array}\vspace{-0.2cm}\]

Consider allocation \(\bm{\mu}=(\bm{a},\bm{c},\bm{b},\bm{d})\).
Define \(\underline{\sigma}:=(\underline{b},\underline{d},\underline{c},\underline{a})\).
The allocation $\mu$ belongs to the strong core and hence to the weak core, but not to the exclusion core: coalition $S=\{1,2,4\}$ exclusion-enforces $\sigma$ at $\mu$ and therefore exclusion-blocks $\mu$ via $\sigma$.\footnote{The only worse-off agent is agent~$3$ with $\mu(3)=b\in \Omega(\{1,2,4\} \mid \mu)=\{a,b,c,d\}$: the recursion first reaches agents $1$ and $3$, who occupy the objects in $\omega(\{1,2,4\})=\{a,b\}$, so its terminal agent set is the grand coalition $N$, and $\omega(N)=O$.}\hfill~$\diamond$
    \end{exam}

\section{Serial top trading cycles}\label{sec.sttc}

\subsection{STTC graphs and the STTC algorithm}\label{subsec.algorithm}

In this subsection, we introduce the STTC graphs and reduced STTC graphs used to define the STTC algorithm. For a given problem, priority order, and arbitrary initial allocation, the STTC algorithm serves as the central tool in our analysis of exclusion core stability.\par\smallskip

Let $\bm{\rhd}$ denote a \textbf{priority order} (i.e., a strict linear order) over the set of agents $N$, and let $\bm{\Delta}$ denote the set of priority orders over $N$. Write $i \rhd j$ when agent~$i$ has higher priority than agent~$j$.\par\smallskip

Given a problem $(N,O,\omega,R)$, an allocation $\mu$, and a priority order $\rhd$ over the set of agents $N$, we refer to $\bm{(N,O,\omega,R,\mu,\rhd)}$ as an \textbf{instance}.

\begin{defi}[\textbf{Serial top trading cycles graphs}]\label{defi.graph}
    Given an instance $(N,O,\omega,R,\mu,\rhd)$, let $\bm{V}=(N\cup O\cup \{\varnothing\})$ be the \textbf{set of vertices} and $\bm{E}$ the \textbf{set of edges}. We denote a directed edge from vertex $x$ to vertex $y$ by $(x,y)$ and say that $x$ points to $y$. The set of edges $E$ is defined as follows:
\begin{itemize}
    \item Each agent $i \in N$ points to his most-preferred object $o$ in $O \cup \{\varnothing\}$ according to $R_i$, yielding an edge $(i,o)$;
    \item each individually owned object $o \in O_I$ points to its owner $i$, yielding an edge $(o,i)$;
    \item each public object $o\in O_P$ that is assigned at $\mu$ points to its assigned agent $i$ at $\mu$ (i.e., $\mu(i)=o$), yielding an edge $(o,i)$;
    \item each public object $o\in O_P$ that is unassigned at $\mu$ points to the agent $i$ with the highest priority according to $\rhd$, yielding an edge $(o,i)$;
    \item the null object $\varnothing$ points to each agent $i\in N$ who points to $\varnothing$, yielding an edge $(\varnothing,i)$ for each such agent.
\end{itemize}
We call the resulting directed graph the \textbf{serial top trading cycles (STTC) graph} $\bm{\mathcal{G}(N,O,\omega,R,\mu,\rhd)}:=(V,E)$ (or $\mathcal{G}(R,\mu,\rhd)$ for short).
\end{defi}

For an individually owned object, the edge $(o,i)$ records exactly the ownership relation appearing in the chains of dependence of Lemma~\ref{lemm.evictionproperties}(d), which is what allows Lemma~\ref{lemm.cyclechain} in Appendix~\ref{appendix:sttc} to read such a chain directly off a cycle of the algorithm.\par\smallskip

A \textbf{cycle} $\bm{\mathcal{C}}$ in an STTC graph, such as $\mathcal{G}(R,\mu,\rhd)$, is an alternating sequence of $k\geq1$ distinct agents and $k$ distinct objects connected by directed edges: $o^1\rightarrow i^1\rightarrow o^2\rightarrow\cdots\rightarrow o^k\rightarrow i^k\rightarrow o^1$. Denote the set of all agents who participate in such a cycle $\mathcal{C}$ by $S(\mathcal{C})$. A cycle containing a single agent, $i\rightarrow o\rightarrow i$, is a \textbf{self-cycle}; its object $o$ may be a real object or the null object. Independently of this, we call a cycle \textbf{ordinary} if it does not involve the null object, that is, if all of its objects are real.\par\smallskip

Since the null object points to \emph{every} agent who points to it, a sequence through $\varnothing$ could formally close through several agents; we therefore adopt the convention that each agent pointing to $\varnothing$ forms the self-cycle $i\rightarrow\varnothing\rightarrow i$, and that this is the only cycle through the null object containing agent~$i$. Every vertex of an ordinary cycle has out-degree one. Hence, distinct cycles are always disjoint in agents and in real objects; the only vertex that distinct cycles can share is the null object, which has enough capacity for all agents.\par\smallskip

The STTC algorithm proceeds in steps. At each step, all agents belonging to cycles leave with the objects they receive, and the same graph construction is then applied to the remaining agents and objects. The following definition formalizes the resulting reduced instance. Its components restrict the original data to the surviving agents and objects, with one substantive exception: an agent whose allotment has already been taken is recorded at the reduced allocation as holding the null object.

\begin{defi}[\textbf{Reduced instances and STTC graphs}]\label{defi.residualgraph}
Consider an instance $(N,O,\omega,R,\mu,\rhd)$. Suppose that the agents in a set $\widebar{N}$ leave the problem with their allotments under an allocation $\sigma$, which may differ from $\mu$. The remaining sets of agents and objects are $\bm{N'}=N\setminus\widebar{N}$ and $\bm{O'}=O\setminus \sigma(\widebar{N})$; as always, the null object remains available.
\begin{itemize}
  \item If $N'=\{i\}$, then define $\omega'_i=O'$; otherwise, for each $i\in N'$, let $\omega'_i:=\omega_i\cap O'$.\footnote{When only one agent remains, individual and public ownership induce the same edges in the reduced STTC graph: every remaining object points to the remaining agent, either as its owner or, if it is public, as its assignee or as the highest-priority remaining agent. Thus, setting $\omega'_i=O'$ is without loss and simplifies the definition of the reduced instance. Note, however, that $\omega'_i\subseteq\omega_i$ holds whenever $|N'|\geq2$, but may fail when $N'=\{i\}$.} Then, the \textbf{reduced endowment structure} is $\bm{\omega'}=(\omega'_i)_{i\in N'}$. Note that $\bigcup_{i\in N'} \omega^{'}_{i} \subseteq O'$. In particular, if $|N'|\geq 2$, a remaining object whose owner has left belongs to no $\omega'_i$ and is therefore public in the reduced instance.
  \item For each agent $i\in N'$, let $R'_i$ be the restriction of $R_i$ to $O'\cup\{\varnothing\}$; that is, for all $a,b\in O'\cup\{\varnothing\}$, $a\mathbin{R'_i}b$ if and only if $a\mathbin{R_i}b$. Then, the \textbf{reduced preference profile} is $R':=(R'_i)_{i\in N'}$.
  \item For each $i\in N'$, let\nopagebreak\vspace{-0.2cm}
      \[\mu'(i) :=
    \begin{cases}
     \mu(i) & \text{ if } \mu(i)\in O'\cup\{\varnothing\}\\
      \varnothing       & \text{ otherwise}.
    \end{cases}\vspace{-0.2cm}\]
    The resulting \textbf{reduced allocation} is $\bm{\mu'}=(\mu'(i))_{i\in N'}$.
  \item For each pair $i,j\in N'$, let $i\rhd' j$ if and only if $i\rhd j$. Then, $\bm{\rhd'}$ denotes the \textbf{reduced priority order}  over the set of agents $N'$.
\end{itemize}
The resulting \textbf{reduced instance} is $\bm{(N',O',\omega',R',\mu',\rhd')}$, and its \textbf{reduced STTC graph} is $\bm{\mathcal{G}(N',O',\omega',R',\mu',\rhd')}$.
\end{defi}

We now use STTC graphs and their reductions to define the STTC algorithm.

\begin{defi}[\textbf{The serial top trading cycles algorithm}]\hfill\\
\noindent \rule{1\columnwidth}{1pt}

\begin{description}
    \item [Input:] An instance $(N,O,\omega,R,\mu,\rhd)$.

    \item \textbf{Step~1}. Let $N^1 := N$ and $O^1 := O$, and construct the STTC graph $\mathcal{G}^1=\mathcal{G}(N^1,O^1,\omega,R,\mu,\rhd)$ according to Definition~\ref{defi.graph}.

    \item[]Identify all top trading cycles in $\mathcal{G}^1$ and execute each cycle:\footnote{Note that each agent points to exactly one object, his most-preferred object, which is why we refer to the cycles that form as \textbf{top trading cycles}. At least one top trading cycle exists at every step: any agent pointing to $\varnothing$ forms the self-cycle $i\to\varnothing\to i$; if no agent points to $\varnothing$, then $\varnothing$ is isolated and every remaining vertex other than the isolated null object has out-degree one, so following the outgoing edges from any agent within the finite vertex set eventually revisits a vertex and thereby traces a cycle.}
    \begin{itemize}
        \item For each cycle $\mathcal{C}$, let $S(\mathcal{C})$ denote its set of agents. Assign each $i \in S(\mathcal{C})$ the object to which he points in $\mathcal{C}$, and denote this allotment by $\STTC_i(R,\mu,\rhd)$.
        \item Let $\widebar{N}^1$ and $\widebar{O}^1$ denote the sets of agents and objects, respectively, belonging to cycles. Remove these agents and objects, except for the null object $\varnothing$, from the current instance, and set $N^2:=N^1\setminus \widebar{N}^1$ and $O^2:=O^1\setminus \widebar{O}^1$.
        \item If $N^2=\emptyset$, stop. If $N^2\neq \emptyset$ and $O^2=\emptyset$, assign the null object to every agent in $N^2$ and stop. Otherwise, continue to Step~2.
    \end{itemize}

    \item \textbf{Step}~$\bm{t \ge 2}$. Step~$t-1$ ended with the definitions of $N^t:=N^{t-1}\setminus \widebar{N}^{t-1}\neq \emptyset$ and $O^t:=O^{t-1}\setminus \widebar{O}^{t-1}\neq \emptyset$ after the removal of top trading cycles in the (reduced) STTC graph $\mathcal{G}^{t-1}$.

        Use $N'=N^t$ and $O'=O^t$ to construct the  corresponding reduced instance (Definition~\ref{defi.residualgraph}) and STTC graph $\mathcal{G}^t=\mathcal{G}(N',O',\omega',R',\mu',\rhd')$.

    \item[] As in Step~1, identify all top trading cycles in $\mathcal{G}^t$ and execute each cycle.

    \begin{itemize}
        \item \raggedright Let $\widebar{N}^t$ and $\widebar{O}^t$ denote the sets of agents and objects, respectively, belonging to cycles in $\mathcal{G}^t$. Assign each cycle agent his corresponding allotment $\STTC_i(R,\mu,\rhd)$, remove all agents in $\widebar{N}^t$ and all objects in $\widebar{O}^t$, except for the null object, and set $N^{t+1}:=N^t\setminus \widebar{N}^t$ and $O^{t+1}:=O^t\setminus \widebar{O}^t$.
        \item If $N^{t+1}=\emptyset$, stop. If $N^{t+1}\neq \emptyset$ and $O^{t+1}=\emptyset$, assign the null object to every agent in $N^{t+1}$ and stop. Otherwise, continue to Step~$t+1$.
    \end{itemize}

    \item \textbf{Termination.} Since at least one agent is removed at each step, the procedure terminates after at most $n$ steps, when each agent has received an allotment.

    \item [Output:] Denote the resulting allocation by $\STTC(R,\mu,\rhd)$.
\end{description}
\noindent \rule{1\columnwidth}{1pt}

An instance $(N,O,\omega,R,\mu,\rhd)$ consists of a problem $(N,O,\omega,R)$, an initial allocation $\mu$, and a priority order $\rhd$. We refer to the STTC algorithm with priority order $\rhd$ and initial allocation $\mu$ as the \textbf{$\bm{(\rhd,\mu)}$-STTC algorithm}.
\end{defi}

Next, we introduce $\mu$-cycles and non-$\mu$-cycles. Along a $\mu$-cycle, each agent points to the object he occupies at $\mu$ and receives it back, so no object changes hands. Thus, only non-$\mu$-cycles can move objects between agents. The presence or absence of such cycles underlies the exclusion-core-stability test in Theorem~\ref{theorem:coretest}.

\begin{defi}[\textbf{$\bm{\mu}$- and non-$\bm{\mu}$-cycles}]
    Let $\mathcal{C}$ be a cycle in $\mathcal{G}^t$ at Step~$t$ of the $(\rhd,\mu)$-STTC algorithm, and let $S(\mathcal{C})\subseteq N^t$ denote its set of agents.
    The cycle $\mathcal{C}$ is a \textbf{$\bm{\mu}$-cycle} if every $i\in S(\mathcal{C})$ points to $\mu(i)$; otherwise, it is a \textbf{non-$\bm{\mu}$-cycle}.

    If a non-$\mu$-cycle forms, let $t^\star$ be the first step at which one forms. Any non-$\mu$-cycle in $\mathcal{G}^{t^\star}$ is called an \textbf{earliest non-$\bm{\mu}$-cycle}.
\end{defi}

Example~\ref{example:STTC} in Appendix~\ref{appendix:STTCexample} illustrates the above notions.

\subsection{STTC allocations and the exclusion core}\label{subsec.sttccore}

The following proposition establishes that allocations produced by the $(\rhd,\mu)$-STTC algorithm are exclusion core stable. The intuition is as follows. Each agent leaves the algorithm with his most-preferred object among those still available. Hence, an agent can improve only by obtaining an object that has already been removed. However, agents who are still present have no exclusion rights over removed objects: an individually owned object cannot leave before its owner, and a public object is controlled only by the grand coalition.

\begin{prop}[\textbf{STTC allocations are exclusion core stable}]\label{prop.sttc-exclusion}
For each problem $R$, each priority order $\rhd$, and each allocation $\mu$, $\STTC(R,\mu,\rhd)$ is exclusion core stable.
\end{prop}

Suppose, contrary to the proposition, that the outcome of the $(\rhd,\mu)$-STTC algorithm is exclusion-blocked. An induction over the steps shows that no agent assigned at a given step can belong to the blocking coalition and that no object removed at that step can enter the blocking coalition's extended eviction set. Consequently, no coalition member can obtain through exclusion-blocking an object he prefers to his STTC allotment, a contradiction. Appendix~\ref{appendix:sttc} contains the proof of Proposition~\ref{prop.sttc-exclusion}.\par\smallskip

When the initial allocation is the \textbf{null allocation} \(\bm{\mu^{\varnothing} := (\varnothing, \varnothing, \ldots, \varnothing)}\), the $(\rhd,\mu^\varnothing)$-STTC algorithm specializes to several well-known algorithms:

\begin{itemize}
    \item For \textbf{Shapley--Scarf housing markets}, each $(\rhd,\mu^{\varnothing})$-STTC algorithm coincides with the classical \textbf{top trading cycles (TTC) algorithm} \citep{SHAPLEY197423}.
    \item For \textbf{house allocation problems}, the $(\rhd,\mu^{\varnothing})$-STTC algorithm coincides with the $\rhd$-\textbf{serial dictatorship} algorithm \citep{svensson1999}.
    \item For \textbf{house allocation problems with existing tenants}, the $(\rhd,\mu^{\varnothing})$-STTC algorithm coincides with the $\rhd$-\textbf{you-request-my-house--I-get-your-turn ($\rhd$-YRMH-IGYT) algorithm} \citep{abdulkadirouglu1999house}.
  \end{itemize}\par\smallskip

Next, we show that the exclusion core can be recovered from the outcomes of the $(\rhd,\mu^{\varnothing})$-STTC algorithms. This result is closely related to Theorem~4 of \citet{balbuzanov2019endowments}; their Lemma~3(b) adapts that theorem to our model.\footnote{\citet{balbuzanov2019endowments} allow object-specific priority orders, whereas our characterization uses a common priority order for all objects. This restriction strengthens the characterization but leaves less flexibility in the proof.}

\begin{prop}[\textbf{An exclusion core characterization}]\label{prop.corecharacterization}
For each problem $R$, the exclusion core is exactly the set of STTC allocations generated from the null allocation as the priority order varies over $\Delta$:
\( \EC(R) = \bigl\{\STTC(R,\mu^{\varnothing},\rhd)\mid \rhd\in\Delta\bigr\}.\)
\end{prop}

By Proposition~\ref{prop.sttc-exclusion}, every STTC allocation generated from the null allocation belongs to the exclusion core. It remains to prove the reverse inclusion. Starting from an exclusion-core-stable allocation $\mu$ and an arbitrary priority order $\rhd^*$, we use the successive steps of the $(\rhd^*,\mu)$-STTC algorithm to construct a priority order $\rhd$ for which \(\STTC(R,\mu^\varnothing,\rhd)=\mu\). Appendix~\ref{appendix:sttc} contains the proof of Proposition~\ref{prop.corecharacterization}.\footnote{\citet[][Proposition~5]{zhang2026consistent} also characterizes the exclusion core on this domain, using a generalization of the YRMH-IGYT algorithm. Our maximal domain result (Corollary~\ref{cor:maxdomain}) and the results of Section~\ref{sec.mainresult}, which rely on the STTC algorithm being initialized at an arbitrary allocation, have no counterpart there.}\par\smallskip

Propositions~\ref{prop.sttc-exclusion} and~\ref{prop.corecharacterization} establish basic properties of the $(\rhd,\mu)$-STTC algorithm that are used to derive our main results in Section~\ref{sec.mainresult}.\par\smallskip

The next example illustrates Proposition~\ref{prop.corecharacterization}. A more complex illustration appears in Example~\ref{example:STTC} in Appendix~\ref{appendix:STTCexample}.

\begin{exam}[\textbf{Example~\ref{exam:logicalindependent} continued}]

Consider the problem from Example~\ref{exam:logicalindependent}: $N=\{1,2,3,4\}$, $O=\{a,b,c,d\}$; agents $1$, $2$, and $3$ own objects $a$, $b$, and $c$, respectively, and $d$ is publicly owned. Preferences are as follows.\nopagebreak\vspace{-0.2cm}
        \[\begin{array}{cccccc}
          R_{1}: & b & a & \varnothing & \ldots\\
          R_{2}: & d & c & b & \varnothing & a\\
          R_{3}: & b& c & \varnothing & \ldots\\
          R_{4}: & a & d & c & \varnothing & b.
      \end{array}\vspace{-0.2cm}\]

As $\rhd$ varies over $\Delta$, the outcomes of the $(\rhd,\mu^\varnothing)$-STTC algorithm are exactly the two allocations \(\mu^1=(a,d,b,c)\) and \(\mu^2=(b,d,c,a)\).
Moreover, $\STTC(R,\mu^\varnothing,\rhd)=\mu^1$ if and only if agent~$3$ has the highest priority under $\rhd$, or agent~$2$ has the highest priority and agent~$3$ the second highest; otherwise $\STTC(R,\mu^\varnothing,\rhd)=\mu^2$.
Hence, by Proposition~\ref{prop.corecharacterization}, the exclusion core is $\{\mu^1,\mu^2\}$.\hfill~$\diamond$
\end{exam}

Before turning to our main results, we relate our STTC algorithms to the GTTC algorithms of \citet{balbuzanov2019endowments} and record two generalizations of Proposition~\ref{prop.sttc-exclusion}.\par\smallskip

The \textbf{generalized top trading cycles (GTTC) algorithm} of \citet[][Algorithm~2]{balbuzanov2019endowments} includes our $(\rhd,\mu^{\varnothing})$-STTC algorithm as a special case. Their algorithm uses an object-specific priority structure, so different objects may have different priority orders. When $\mu\neq\mu^{\varnothing}$, our $(\rhd,\mu)$-STTC algorithms differ in form from their $\rhd$-GTTC algorithm. In the proof of their Proposition~4, however, \citet{balbuzanov2019endowments} use a modification of their GTTC algorithm that corresponds to our STTC algorithm. The correspondence between the two algorithms yields the following common generalization.

\begin{remark}[\textbf{Flexible top trading cycles graphs, algorithms, and allocations}]
In an STTC graph, the priority order $\rhd$ is used only to determine the outgoing edge of each unassigned public object. These edges may be chosen arbitrarily rather than through a common priority order, for instance depending on the cycles executed at earlier steps.

We refer to the resulting constructions as \textbf{flexible top trading cycles (flexible TTC) graphs, algorithms}, and \textbf{allocations}. For mixed-ownership problems, the class of flexible TTC algorithms includes both our STTC algorithms and the corresponding GTTC algorithms of \citet{balbuzanov2019endowments}. The proof of Proposition~\ref{prop.sttc-exclusion} uses no property of the edge choices for unassigned public objects beyond the fact that each such object points to a remaining agent. Hence, every allocation generated by a flexible TTC algorithm is exclusion core stable.

This observation concerns only Proposition~\ref{prop.sttc-exclusion}. Proposition~\ref{prop.corecharacterization} and Theorem~\ref{theorem:coretest}, by contrast, draw their strength from the single priority order of the STTC family: the former asserts the existence of a priority order under which a given allocation is generated, the latter characterizes exclusion core stability through the fixed points of the $(\rhd,\mu)$-STTC algorithm.\label{remark1}\hfill~$\diamond$
\end{remark}

\begin{remark}[\textbf{Flexible TTC allocations for object allocation problems with complex endowments}]
Recall that in object allocation problems with complex endowments, an object may be owned by any coalition of agents, not only by a singleton or by the grand coalition. We can extend the flexible TTC construction to such problems by requiring an object to point to one of its remaining owners whenever at least one owner remains; once no owner remains, its outgoing edge can be chosen arbitrarily among the remaining agents. A flexible TTC algorithm satisfying this requirement is \textbf{owner-respecting}: an object points to an owner whenever one remains, so owners' claims are not bypassed while they are present.

The argument of Proposition~\ref{prop.sttc-exclusion} extends to this class. Indeed, if a set $T\subsetneq N$ consists only of agents who remain at some step and an object belongs to the direct eviction set of $T$, then all owners of that object belong to $T$ and hence remain; in particular, the object is not public, so owner-respectingness applies to it. Since all of its owners remain, they were present at every earlier step. If the object had been removed in an earlier cycle, owner-respectingness would have made it point to one of its owners, who would then have been removed with it, a contradiction. Thus, a coalition of agents who all remain at a given step has no exclusion rights over objects that have already been assigned, exactly as in the proof of Proposition~\ref{prop.sttc-exclusion}, and the rest of that proof is unchanged. Hence, every allocation generated by an \textbf{owner-respecting flexible TTC algorithm} is exclusion core stable. For the special case in which each object's edge is determined by a fixed priority order over agents, this also follows from Corollary~1(a) and Lemma~3(a) of \citet{balbuzanov2019endowments}. Owner-respecting flexible TTC algorithms are more general: an object may point to different owners at different steps, as long as each is a remaining owner, and such edge choices need not be generated by any fixed order.\label{remark2}\hfill~$\diamond$
\end{remark}

\section{Main results}\label{sec.mainresult}

\subsection{A maximal domain result}\label{subsec.maximaldomain}

Proposition~\ref{prop.corecharacterization} raises the question of whether the STTC characterization of the exclusion core extends beyond mixed ownership. The next result shows that it does not when the characterization is required to hold for every preference profile. Every complex endowment structure containing a properly co-owned object admits a preference profile for which some exclusion-core-stable allocation cannot be obtained by the extension of the STTC algorithm defined below. In fact, that allocation cannot be obtained by any owner-respecting flexible TTC algorithm described in Remark~\ref{remark2}.\par\smallskip

Throughout this subsection, we allow proper co-ownership and write $(\omega,R)$ for a problem $(N,O,\omega,R)$. Accordingly, $\omega_i$ now denotes the set of objects that agent~$i$ owns, alone or as a proper co-owner. As before, public objects do not belong to any individual set $\omega_i$. An object $a\in O$ is \textbf{properly co-owned} if at least two but not all agents own it; that is, there exist distinct agents $i,j\in N$ with $a\in\omega_i\cap\omega_j$ and an agent $k\in N$ with $a\notin\omega_k$. Under mixed ownership, by contrast, every object is owned either by exactly one agent or by all agents. For each object $o\in O$, let $S^o\subseteq N$ denote its \textbf{set of owners}, defined by\nopagebreak\vspace{-0.2cm}
\[S^o :=
\begin{cases}
\{i\in N \mid o\in\omega_i\}, & \text{if } \{i\in N \mid o\in\omega_i\}\neq\emptyset,\\
N, & \text{otherwise}.
\end{cases}\vspace{-0.2cm}\]
The direct eviction set of a coalition $S$ is \(\omega(S):=\{o\in O\mid S^o\subseteq S\}.\)
We call the members of $S^o$ the owners of $o$, and use \emph{co-owner} only when $o$ is properly co-owned and the presence of several owners matters. Thus, a coalition can directly evict others from an object exactly when it contains all of that object's owners.\footnote{ Whenever no object is properly co-owned, this coincides with the direct eviction set of Section~\ref{subsec.model}: for $S\subsetneq N$, an individually owned object requires its single owner and a public object requires $N$, so $\omega(S)=\bigcup_{i\in S}\omega_i$; and $\omega(N)=O$.} The extended eviction set is obtained from this direct eviction set by the recursion introduced in Section~\ref{subsec.core}.\par\smallskip

In an extended $(\rhd,\mu)$-STTC algorithm, let $\mu'$ denote the current reduced allocation, with $\mu'=\mu$ at Step~1. A properly co-owned object points as follows. If it is assigned at $\mu'$ to a remaining co-owner, it points to that agent. Otherwise, if any co-owner remains, it points to the highest-priority remaining co-owner; if none remains, it points to its assigned agent at $\mu'$ if it is assigned, and otherwise to the highest-priority remaining agent. The outgoing-edge rules for individually and publicly owned objects remain unchanged.\par\smallskip

For a problem $(\omega,R)$, let $\STTC(\omega,R,\mu,\rhd)$ denote the outcome of the extended $(\rhd,\mu)$-STTC algorithm. The set of \textbf{STTC allocations generated from the null allocation} is\nopagebreak\vspace{-0.2cm}
\[ \STTC(\omega,R):=\{\STTC(\omega,R,\mu^\varnothing,\rhd)\mid\rhd\in\Delta\}.\vspace{-0.2cm}\]
We denote the exclusion core of the problem $(\omega,R)$ by $\EC(\omega,R)$.

\begin{theorem}[\textbf{Strict inclusion of the STTC allocations in the exclusion core under proper co-ownership}]\label{thm.strictsuperset}
Let $\omega$ be an endowment structure with a properly co-owned object. Then there exists a preference profile $R$ such that \( \STTC(\omega,R) \subsetneq \EC(\omega,R).\)
\end{theorem}

The proof fixes a properly co-owned object $a$ and chooses preferences under which every agent ranks $a$ above $\varnothing$ and every other object below $\varnothing$. Let $\mu^a$ denote the allocation that assigns $a$ to a non-co-owner and $\varnothing$ to every other agent. By Remark~\ref{remark2}, for every priority order $\rhd$, $\STTC(\omega,R,\mu^\varnothing,\rhd)\in\EC(\omega,R)$. Hence, \(\STTC(\omega,R)\subseteq\EC(\omega,R)\). The allocation $\mu^a$ is exclusion core stable but is not an STTC outcome. Therefore, \( \STTC(\omega,R)\subsetneq\EC(\omega,R).\) Appendix~\ref{appendix:main} contains the proof of Theorem~\ref{thm.strictsuperset}.\par\smallskip

In fact, $\mu^a$ also lies outside the range of every owner-respecting flexible TTC algorithm described in Remark~\ref{remark2}. The underlying tension is that these trading-cycles algorithms respect the co-owners' claims, whereas exclusion core stability may sustain an outsider's consumption of a properly co-owned object because a blocking coalition must make all its members strictly better off.\par\smallskip

Combining Proposition~\ref{prop.corecharacterization} and Theorem~\ref{thm.strictsuperset} yields a maximal domain result.

\begin{corollary}[\textbf{A maximal domain result}]\label{cor:maxdomain}
For each endowment structure $\omega$, the following statements are equivalent:
\begin{enumerate}
    \item[\textup{(i)}] for every preference profile $R$, $\STTC(\omega,R)=\EC(\omega,R)$;
    \item[\textup{(ii)}] $\omega$ is a mixed-ownership structure, i.e., each object is either individually or publicly owned.
\end{enumerate}
\end{corollary}

Appendix~\ref{appendix:main} contains the proof of Corollary~\ref{cor:maxdomain}.
Corollary~\ref{cor:maxdomain} shows that, among domains defined by restrictions on endowment structures, the mixed-ownership domain is the largest on which the STTC algorithms characterize the exclusion core.

\subsection{An exclusion-core-stability test}

The STTC algorithm also provides a simple test for exclusion core stability. An allocation is exclusion core stable exactly when it is a fixed point of the $(\rhd,\mu)$-STTC algorithm---for one, and hence every, priority order $\rhd$. Otherwise, an earliest non-$\mu$-cycle identifies an exclusion-blocking coalition.

\begin{theorem}[\textbf{An exclusion-core-stability test}]\label{theorem:coretest}
Let $R$ be a problem and $\mu$ an allocation. Then the following statements are equivalent:
\begin{enumerate}
\item[\textup{(i)}] $\mu$ is exclusion core stable;
\item[\textup{(ii)}] for some priority order $\rhd$, $\mu=\STTC(R,\mu,\rhd)$;
\item[\textup{(iii)}] for every priority order $\rhd$, $\mu=\STTC(R,\mu,\rhd)$.
\end{enumerate}
\end{theorem}

That \textup{(ii)} implies \textup{(i)} is Proposition~\ref{prop.sttc-exclusion}, and that \textup{(iii)} implies \textup{(ii)} is immediate. For the remaining implication, suppose that $\mu$ is exclusion core stable and that $\STTC(R,\mu,\rhd)\neq\mu$ for some priority order~$\rhd$. Then an earliest non-$\mu$-cycle exists, and the agents in that cycle who strictly improve upon their $\mu$-allotments form an exclusion-blocking coalition, contradicting the exclusion core stability of $\mu$. Appendix~\ref{appendix:main} contains the proof of Theorem~\ref{theorem:coretest}.
Example~\ref{example:STTC} in Appendix~\ref{appendix:STTCexample} illustrates the earliest-non-$\mu$-cycle argument used in the proof of Theorem~\ref{theorem:coretest}.\par\smallskip

The construction behind Theorem~\ref{thm.strictsuperset} also shows that this fixed-point characterization does not extend to object allocation problems with proper co-ownership. Recall the allocation $\mu^a$ defined after Theorem~\ref{thm.strictsuperset}, which is exclusion core stable. However, when an owner-respecting flexible TTC algorithm is initialized at $\mu^a$, the properly co-owned object is reassigned from a non-co-owner to a co-owner at the first step. Hence, the implication from exclusion core stability to the fixed-point condition fails.

\subsection{Von Neumann--Morgenstern stability}\label{subsec.vNM}

We first introduce the standard definition of a von Neumann--Morgenstern (\citeyear{vonneumann1944theory}) stable set with respect to an arbitrary domination relation between allocations.

\begin{defi}[\textbf{Von Neumann--Morgenstern stable sets}]
A set of allocations $\mathcal{V}$ is \textbf{von Neumann--Morgenstern stable} with respect to a domination relation if it is
\begin{itemize}
    \item \textbf{internally stable}: for all distinct $\mu,\nu\in\mathcal{V}$, $\nu$ does not dominate $\mu$; and
    \item \textbf{externally stable}: for each $\mu \notin \mathcal{V}$, there exists $\nu \in \mathcal{V}$ that dominates $\mu$.
\end{itemize}
\end{defi}

The domination relation we use is the one induced by exclusion-blocking.

\begin{defi}[\textbf{Exclusion-domination}]\label{defi.exclusiondomination}
Given a problem $R$ and allocations $\mu$ and $\nu$, we say that $\nu$ \textbf{exclusion-dominates} $\mu$ if some coalition $S$ exclusion-blocks $\mu$ via $\nu$.
\end{defi}

We next show that the exclusion core is always a von Neumann--Morgenstern stable set with respect to exclusion-domination, in a strong sense: a dominating exclusion-core-stable allocation can be computed with the STTC algorithm.

\begin{theorem}[\textbf{Von Neumann--Morgenstern stability of the exclusion core}]\label{theorem:vNMstability} Let $R$ be a problem. Then,  $\EC(R)$ is a von Neumann--Morgenstern stable set with respect to exclusion-domination. More precisely, it is
\begin{itemize}
    \item \textbf{internally stable}: for all distinct $\mu,\nu\in\EC(R)$, $\nu$ does not exclusion-dominate $\mu$; and
    \item \textbf{externally stable} in the following strong sense: for each $\mu \notin \EC(R)$ and each priority order $\rhd$, $\STTC(R,\mu,\rhd)$ belongs to $\EC(R)$ and exclusion-dominates $\mu$.
\end{itemize}\par
\end{theorem}

We verify the two stability conditions separately.

\paragraph{Internal stability.} Internal stability follows directly from the definition of the exclusion core: no allocation exclusion-dominates an exclusion-core-stable allocation.

\paragraph{External stability.} External stability is the substantive condition. Theorem~\ref{theorem:vNMstability} asserts more than the existence of some dominating exclusion-core-stable allocation: for each allocation $\mu$ outside the exclusion core, every STTC algorithm initialized at $\mu$ produces one. The coalition that carries out this exclusion-blocking is assembled from the cycles of the $(\rhd,\mu)$-STTC algorithm, and the same cycles generate the blocking path of Theorem~\ref{theorem:pathtoexclusion}. We therefore prove the two theorems together: the joint argument shows how the same STTC cycles generate both local blocking transitions and a global exclusion-dominating allocation. Section~\ref{subsec.paths} outlines the common construction, Appendix~\ref{appendix:main} contains the joint proof, and Example~\ref{example:STTC} in Appendix~\ref{appendix:STTCexample} illustrates both conclusions on a concrete problem.\par\smallskip

\citet[][Proposition~4]{balbuzanov2019endowments} establish a related stability result for what they call relational economies with acyclic priority structures. Their proposition concerns the strong and weak exclusion cores of such an economy, which are derived from conditional endowments that vary with the prevailing allocation; by their Lemma~3(b), these solutions coincide with the exclusion core of an object allocation problem exactly when every object is individually or publicly owned. How far Theorem~\ref{theorem:vNMstability} itself extends is taken up in Appendix~\ref{appendix:acyclic}. An endowment structure is acyclic in the sense of \citet{balbuzanov2019endowments} exactly when either no object is properly co-owned or all objects share a single set of owners (Proposition~\ref{prop.acyclicchar} in Appendix~\ref{appendix:acyclic}). Theorem~\ref{thm.uniformvNM} (Appendix~\ref{appendix:acyclic}) establishes von Neumann--Morgenstern stability on the second class, so the conclusion of Theorem~\ref{theorem:vNMstability} holds throughout the acyclic class (Corollary~\ref{cor.acyclicvNM} in Appendix~\ref{appendix:acyclic}). Beyond the acyclic class, external stability may fail. Example~\ref{exam:noexstable} in Appendix~\ref{appendix:acyclic}, due to Jacob Coreno, has five agents and three objects, one individually owned and two properly co-owned. In it, four allocations outside the exclusion core form a cycle of exclusion-domination: each of them is exclusion-dominated only by the next one in the cycle, so none of them is exclusion-dominated by an exclusion-core-stable allocation. By contrast, Example~\ref{exam:beyondacyclic} in Appendix~\ref{appendix:acyclic} exhibits a non-acyclic endowment structure, with three agents, one object co-owned by two of them, and one public object, whose exclusion core is externally stable at every preference profile. Hence, acyclicity is sufficient, but not necessary, for external stability.\par\medskip

\subsection{Paths to exclusion core stability}\label{subsec.paths}

The core concepts introduced in Section~\ref{subsec.core} are static: they are defined by the absence of blocking. Blocking, however, describes a coalitional deviation and is inherently dynamic. This raises a natural question: does an explicit process of blocking lead to core allocations? \citet{serrano2008mistakes} answer this question affirmatively for classical core concepts.\par\smallskip

Paths to stability have also been established for marriage markets by \citet{roth1990random}, for roommate markets admitting stable matchings by \citet{DIAMANTOUDI200418}, and for matching markets with couples under weak responsiveness by \citet{KLAUS2007154}. We now tackle this question for exclusion-blocking and the exclusion core.\par\medskip

When a coalition $S$ exclusion-blocks $\mu$ via $\sigma$, every member of $S$ is strictly better off at $\sigma$. Implementing $\sigma$, however, may require agents outside $S$ who also benefit from it. Thus, $S$ need not have sufficient exclusion rights to implement the blocking on its own. The following simple example illustrates this point.\par\smallskip

Suppose that $N=\{1,2\}$, agent~1 owns object~$a$, agent~2 owns object~$b$, and preferences satisfy
\( b\mathbin{P_1}a\mathbin{P_1}\varnothing\) and \(a\mathbin{P_2}b\mathbin{P_2}\varnothing.\)
Consider \(\mu=(a,b)\) and \(\sigma=(b,a)\).
Allocation $\sigma$ Pareto dominates $\mu$, and coalition $\{1\}$ exclusion-blocks $\mu$ via $\sigma$: agent~1 is strictly better off, while agent~2 is also better off and therefore need not be excluded. However, if only agent~1 participates and receives object~$b$, the resulting allocation is $(b,\varnothing)$. Since \( b\notin\Omega(\{1\}\mid\mu)=\{a\},\) coalition $\{1\}$ cannot exclusion-block $\mu$ via this allocation. Thus, although $\{1\}$ exclusion-blocks $\mu$ via $\sigma$, it cannot implement this blocking on its own.\par\medskip

To model exclusion-blocking dynamically, we distinguish the blocking allocation $\sigma$ from the allocation obtained when coalition members receive their $\sigma$-allotments and outsiders' allotments are disrupted as little as possible. This distinction leads to the notions of post-blocking allocations, self-contained exclusion-blockings, and the paths they generate.

\begin{defi}[\textbf{Post-blocking allocations and self-contained exclusion-blockings}]\label{defi.postblocking}
Given a problem $R$, allocations $\mu$ and $\sigma$, and a coalition $S$ that exclusion-blocks $\mu$ via $\sigma$, the \textbf{post-blocking allocation induced by $\bm{(S,\sigma)}$ at $\bm{\mu}$}, denoted by $\widehat\sigma^{S,\mu}$, is defined by\nopagebreak\vspace{-0.2cm}
\begin{equation}
\widehat\sigma^{S,\mu}(i)=
\begin{cases}
\sigma(i) &\text{ if } i\in S,\\
\varnothing &\text{ if } i\notin S \text{ and } \mu(i)\in\sigma(S)\setminus\{\varnothing\},\\
\mu(i) &\text{ otherwise.}\footnotemark
\end{cases}\vspace{-0.2cm}
\end{equation}
\footnotetext{Since the null object has enough capacity for all agents, an outsider whose $\mu$-allotment is $\varnothing$ is never displaced; hence, the eviction clause refers to $\sigma(S)\setminus\{\varnothing\}$. The resulting map is an allocation: no real object is assigned to more than one coalition member; an outsider retains a real $\mu$-allotment only if no coalition member receives it; and all displaced outsiders receive the null object, which has enough capacity for all agents.} We also say that $\widehat\sigma^{S,\mu}$ is obtained by satisfying the exclusion-blocking of $\mu$ by $S$ via $\sigma$. This exclusion-blocking is \textbf{self-contained} if $S$ also exclusion-blocks $\mu$ via $\widehat\sigma^{S,\mu}$.
\end{defi}

We use self-contained exclusion-blockings to define the corresponding dynamic paths.

\begin{defi}[\textbf{Self-contained exclusion-blocking paths}]\label{defi.path}
For $K\geq0$, a sequence of allocations $\mu^0\mapsto\mu^1\mapsto\ldots\mapsto\mu^K$ is \textbf{a self-contained exclusion-blocking path of length $\bm{K}$ from $\bm{\mu^0}$ to $\bm{\mu^K}$} if, for every $\ell\in\{0,1,\ldots,K-1\}$, $\mu^{\ell+1}$ is the post-blocking allocation induced by a self-contained exclusion-blocking of $\mu^\ell$.
\end{defi}

The next lemma shows that restricting attention to self-contained exclusion-blockings does not eliminate any blocking opportunity.

\begin{lemm}[\textbf{Self-contained completion of an exclusion-blocking coalition}]\label{lemm.selfcontained}
Suppose that $S$ exclusion-blocks $\mu$ via $\sigma$, and let
\( S^+(\mu,\sigma):=\{i\in N\mid \sigma(i)\mathbin{P_i}\mu(i)\}\)
be the set of agents who are strictly better off at $\sigma$ than at $\mu$. Then $S\subseteq S^+(\mu,\sigma)$, $S^+(\mu,\sigma)$ exclusion-blocks $\mu$ via $\sigma$, and this exclusion-blocking is self-contained.
\end{lemm}

Appendix~\ref{appendix:main} contains the proof of Lemma~\ref{lemm.selfcontained}.\par\smallskip

In the preceding illustration, $S^+(\mu,\sigma)=\{1,2\}$, and its induced post-blocking allocation is $\sigma$. The completion therefore adds the improving outsider whose participation is needed to implement the blocking. The set $S^+(\mu,\sigma)$ contains every agent who improves, so it is the largest coalition that can carry out this blocking; in general, it need not be inclusion-minimal among the self-contained blocking coalitions.\par\medskip

Because self-contained exclusion-blockings capture every blocking opportunity and each of them induces a well-defined successor allocation, they can be applied repeatedly. Whether doing so leads to the exclusion core is the central question of the dynamic approach: does every allocation admit a self-contained exclusion-blocking path (Definition~\ref{defi.path}) to the exclusion core? We show that the answer is affirmative for mixed-ownership problems and, more generally, for acyclic endowment structures.\par\medskip

Before presenting our convergence result, we show that arbitrary sequences of self-contained exclusion-blockings need not converge. Small blocking coalitions may be particularly likely to form.
An agent can identify a profitable unilateral deviation without coordinating with others, and two agents may discover a joint deviation through bilateral communication. By contrast, forming and coordinating a larger coalition may be more difficult. Motivated by this possibility, in the following example we repeatedly execute the smallest self-contained exclusion-blocking coalitions, which have size one or two.

\begin{exam}[\textbf{The self-contained exclusion-blocking dynamics can cycle}]
Let $N=\{1,2,3\}$, $O=\{a,b,c\}$, and endowments $\omega_1=\{a\}$, $\omega_2=\{b\}$, and $\omega_3=\{c\}$. Preferences are given by\nopagebreak\vspace{-0.2cm}
\[\begin{array}{ccccc}
R_1: & b & c & a & \varnothing\\
R_2: & c & a & b & \varnothing\\
R_3: & a & b & c & \varnothing.
\end{array}\vspace{-0.2cm}\]

Consider the following sequence of allocations, obtained by successively satisfying self-contained exclusion-blockings of minimum size. Starting from $\mu^1=(b,a,c)$, the successive blockings are summarized in the following table.

\begin{center}
\begingroup
\small
\setlength{\tabcolsep}{3pt}
\renewcommand{\arraystretch}{1.25}
\begin{tabular}{>{\centering\arraybackslash}p{0.9cm} >{\raggedright\arraybackslash}p{6.0cm} >{\raggedright\arraybackslash}p{4.0cm} >{\raggedright\arraybackslash}p{4.6cm}}
\hline
\textbf{Step} & \textbf{Blocking coalition and gains} & \textbf{Resulting allocation} & \textbf{Displaced outsider}\\
\hline
1 & $\{2,3\}$: $c\mathbin{P_2}a$, $b\mathbin{P_3}c$ & $\mu^2=(\varnothing,c,b)$ & Agent~$1$, from $b\in\omega_2$\\
2 & $\{1\}$: $a\mathbin{P_1}\varnothing$ & $\mu^3=(a,c,b)$ & None\\
3 & $\{1,3\}$: $c\mathbin{P_1}a$, $a\mathbin{P_3}b$ & $\mu^4=(c,\varnothing,a)$ & Agent~$2$, from $c\in\omega_3$\\
4 & $\{2\}$: $b\mathbin{P_2}\varnothing$ & $\mu^5=(c,b,a)$ & None\\
5 & $\{1,2\}$: $b\mathbin{P_1}c$, $a\mathbin{P_2}b$ & $\mu^6=(b,a,\varnothing)$ & Agent~$3$, from $a\in\omega_1$\\
6 & $\{3\}$: $c\mathbin{P_3}\varnothing$ & $\mu^1=(b,a,c)$ & None\\
\hline
\end{tabular}
\endgroup
\end{center}

At each step, every member of the listed coalition is strictly better off. At Steps~1, 3, and~5, the only displaced outsider occupies an object owned by a member of the blocking coalition; at Steps~2, 4, and~6, no outsider is displaced. Hence, every transition is generated by a self-contained exclusion-blocking, and Step~6 returns to $\mu^1$, so the sequence cycles.\par\smallskip

At every step, the listed coalition has minimum size among the available self-contained exclusion-blocking coalitions. At Steps~2, 4, and~6, they are singletons.

At Step~1, agent~1 already receives his most-preferred object, while
\( c\notin\Omega(\{2\}\mid\mu^1)=\{a,b\}\) and \(\{a,b\}\cap\Omega(\{3\}\mid\mu^1)=\emptyset,\)
so neither agent~2 nor agent~3 can block alone.
Similarly, at Step~3,
\( \{b,c\}\cap\Omega(\{1\}\mid\mu^3)=\emptyset\) and \( a\notin\Omega(\{3\}\mid\mu^3)=\{b,c\},\)
and at Step~5,
\( b\notin\Omega(\{1\}\mid\mu^5)=\{a,c\}\) and \(\{a,c\}\cap\Omega(\{2\}\mid\mu^5)=\emptyset.\)
Thus, no singleton exclusion-blocks at Steps~1, 3, or~5.\par\smallskip

Allocation \(\mu=(b,c,a)\) assigns every agent his most-preferred object. It therefore Pareto dominates every other allocation and is the unique Pareto efficient allocation. By Lemma~\ref{lemm.paretoexclusion}, for each $k\in\{1,\ldots,6\}$, the agents who are strictly better off at $\mu$ than at $\mu^k$ exclusion-block $\mu^k$ via~$\mu$.\hfill~$\diamond$
\end{exam}

Although the exclusion-blocking dynamics can cycle, the next theorem shows that every allocation admits some self-contained exclusion-blocking path to the exclusion core.

\begin{theorem}[\textbf{Paths to exclusion core stability}]\label{theorem:pathtoexclusion}
Let $R$ be a problem and $\mu$ an allocation. Then there exist $K\geq0$ and allocations $\mu^0,\mu^1,\ldots,\mu^K$ such that $\mu^0=\mu$, $\mu^K$ is exclusion core stable, and
\( \mu^0\mapsto\mu^1\mapsto\cdots\mapsto\mu^K\)
is a self-contained exclusion-blocking path.
\end{theorem}

Appendix~\ref{appendix:main} contains the joint proof of Theorems~\ref{theorem:vNMstability} and~\ref{theorem:pathtoexclusion}. We outline the common construction here.\par\smallskip

Let $\mu$ be an allocation that is not exclusion core stable, fix a priority order $\rhd$, and set
\( \mu^\rhd:=\STTC(R,\mu,\rhd).\)
By Proposition~\ref{prop.sttc-exclusion}, $\mu^\rhd$ is exclusion core stable, and by Theorem~\ref{theorem:coretest}, $\mu^\rhd\neq\mu$.\par\smallskip

Follow the steps of the $(\rhd,\mu)$-STTC algorithm at which non-$\mu$-cycles form. At each such step, assign the cycle agents their final STTC allotments and assign the null object to every outsider whose current allotment is taken. Let the cycle agents who strictly improve relative to the current allocation form the local coalition. Whenever this coalition is nonempty, a backward argument along a chain of dependence shows that it exclusion-blocks the current allocation and that the resulting allocation is precisely its post-blocking allocation. If the coalition is empty, the constructed allocation is unchanged. Omitting these repetitions yields a self-contained exclusion-blocking path from $\mu$ to $\mu^\rhd$.\par\smallskip

The same induction establishes external stability. The agents encountered along the construction who prefer their final STTC allotments to their original allotments form a cumulative coalition that exclusion-blocks $\mu$ via $\mu^\rhd$.\par\smallskip

Example~\ref{example:STTC} in Appendix~\ref{appendix:STTCexample} illustrates the construction.\par\medskip
How far Theorem~\ref{theorem:pathtoexclusion} extends beyond mixed ownership is taken up in Appendix~\ref{appendix:acyclic}. There, Proposition~\ref{prop.acyclicchar} shows that an endowment structure is acyclic exactly when either no object is properly co-owned or all objects share a single set of owners. The first case is covered by Theorem~\ref{theorem:pathtoexclusion} itself and the second by Theorem~\ref{thm.uniformpath} (Appendix~\ref{appendix:acyclic}), so the conclusion of Theorem~\ref{theorem:pathtoexclusion} holds throughout the acyclic class (Corollary~\ref{cor.acyclicpath} in Appendix~\ref{appendix:acyclic}). Beyond the acyclic class, the conclusion may fail: in Example~\ref{exam:noexstable} in Appendix~\ref{appendix:acyclic}, no self-contained exclusion-blocking path that starts at one of the four allocations described in Section~\ref{subsec.vNM} reaches the exclusion core. Acyclicity is not necessary for the conclusion either (Example~\ref{exam:beyondacyclic}).\par\medskip

Theorem~\ref{theorem:pathtoexclusion} exhibits one path to the exclusion core, constructed from the cycles of a single STTC algorithm. A decentralized process need not follow that path: at an allocation outside the exclusion core, several self-contained exclusion-blockings may be available, and which of them is carried out is not determined by the model. We therefore ask what happens when the blocking that occurs at each step is left to chance. The following process makes this precise.\par\smallskip

Fix a problem $R$ and an allocation $\mu$, and set $\mu^1:=\mu$. Given $\mu^\ell$, if $\mu^\ell$ is exclusion core stable, set $\mu^{\ell+1}:=\mu^\ell$. Otherwise, choose at random a post-blocking allocation induced by a self-contained exclusion-blocking of $\mu^\ell$ and denote it by $\mu^{\ell+1}$.\par\smallskip

The set of available post-blocking allocations is nonempty whenever $\mu^\ell$ is not exclusion core stable. Indeed, some coalition then exclusion-blocks $\mu^\ell$, and Lemma~\ref{lemm.selfcontained} yields a self-contained completion and hence an available post-blocking allocation. Thus, the process is well defined.\par\smallskip

Theorem~\ref{theorem:pathtoexclusion} yields the following corollary.

\begin{corollary}[\textbf{Random paths to exclusion core stability}]\label{coro:randompath}
Let $R$ be a problem and $\mu$ an allocation. Suppose that, in the random process defined above, every available post-blocking allocation is selected with positive probability and the selection probabilities depend only on the current allocation. Then, with probability one, the sequence $(\mu^\ell)_{\ell\geq 1}$ reaches an exclusion-core-stable allocation after finitely many steps and remains there.
\end{corollary}

Appendix~\ref{appendix:main} contains the proof of Corollary~\ref{coro:randompath}, which uses only the existence of self-contained exclusion-blocking paths to the exclusion core; hence, its conclusion extends to the acyclic endowment structures studied in Appendix~\ref{appendix:acyclic}, for which such paths exist as well.\footnote{The process is also well defined for these endowment structures: Lemma~\ref{lemm.selfcontained} and its proof use only the monotonicity of extended eviction sets (Lemma~\ref{lemm.evictionproperties}(b)), which holds for every endowment structure.} Beyond the acyclic class, however, it may fail: in Example~\ref{exam:noexstable}, the random process started at any of the four allocations described in Section~\ref{subsec.vNM} never reaches the exclusion core.\par\smallskip

The next example applies the path result to Example~\ref{exam:extendedendowment} and clarifies the distinction between external stability and dynamic reachability. External stability is a static domination property, whereas a self-contained exclusion-blocking path describes how an exclusion-core-stable allocation can be reached through successive deviations.

\begin{exam}[\textbf{Example~\ref{exam:extendedendowment} continued}]
Consider again the problem from Example~\ref{exam:extendedendowment}, with $N=\{1,2,3\}$ and $O=\{a,b,c\}$. Agents $1$, $2$, and $3$ own objects $a$, $b$, and $c$, respectively. Preferences are as follows.\nopagebreak\vspace{-0.2cm}
      \[ \begin{array}{ccccc}
          R_{1}: & \underline{a} & \bm{c} & b & \varnothing\\
          R_{2}: & \bm{a} & \underline{b} & c & \varnothing\\
          R_{3}: & \bm{b} & a & \underline{c} & \varnothing
      \end{array}\vspace{-0.2cm}      \]
Recall that coalition $S=\{1\}$ exclusion-blocks allocation \(\bm{\mu}=(\bm{c},\bm{a},\bm{b})\) via the exclusion-core-stable allocation \(\underline{\sigma}=(\underline{a},\underline{b},\underline{c})\).
At $\mu$, $\{1\}$ is the only coalition that exclusion-blocks $\mu$. We construct a self-contained exclusion-blocking path of length three from $\mu$ to $\underline{\sigma}$ and then show that no shorter path reaches the exclusion core. The three transitions are summarized in the following table.

\begin{center}
\begingroup
\small
\setlength{\tabcolsep}{3pt}
\renewcommand{\arraystretch}{1.25}
\begin{tabular}{>{\centering\arraybackslash}p{0.9cm} >{\raggedright\arraybackslash}p{6.0cm} >{\raggedright\arraybackslash}p{4.0cm} >{\raggedright\arraybackslash}p{4.6cm}}
\hline
\textbf{Step} & \textbf{Blocking coalition and gains} & \textbf{Resulting allocation} & \textbf{Displaced outsider}\\
\hline
1 & $\{1\}$: $a\mathbin{P_1}c$ & $\mu^1=(a,\varnothing,b)$ & Agent~$2$, from $a\in\omega_1$\\
2 & $\{2\}$: $b\mathbin{P_2}\varnothing$ & $\mu^2=(a,b,\varnothing)$ & Agent~$3$, from $b\in\omega_2$\\
3 & $\{3\}$: $c\mathbin{P_3}\varnothing$ & $\underline{\sigma}=(a,b,c)$ & None\\
\hline
\end{tabular}
\endgroup
\end{center}

At each step, the listed coalition exclusion-blocks the preceding allocation via the resulting post-blocking allocation. Hence,
\( \mu\mapsto\mu^1\mapsto\mu^2\mapsto\underline{\sigma}\)
is a self-contained exclusion-blocking path to the exclusion core.

\paragraph{Minimality.} No self-contained exclusion-blocking path from $\mu$ to the exclusion core has length less than three. At $\mu$, agents~$2$ and~$3$ receive their most-preferred objects, so only agent~$1$ can be made strictly better off; hence $S=\{1\}$, and agent~$1$ can improve only by receiving object~$a$. The induced post-blocking allocation is therefore $\mu^1$, the unique successor of $\mu$. At $\mu^1$, agents~$1$ and~$3$ receive their most-preferred objects, so any exclusion-blocking coalition must be $\{2\}$. Since $\Omega(\{2\}\mid\mu^1)=\{b,c\}$, agent~$2$ can obtain only $b$ or $c$. These choices yield, respectively, $\mu^2$ and $(a,c,b)$. Neither equals $\underline{\sigma}$, and $\EC(R)=\{\underline{\sigma}\}$ by Proposition~\ref{prop.corecharacterization}. Hence, no self-contained exclusion-blocking path of length at most two reaches the exclusion core.\par\smallskip

Although $\{1\}$ exclusion-blocks $\mu$ via $\sigma$, satisfying this blocking under the minimally disruptive rule yields $\mu^1$, not $\sigma$. The blocking is self-contained because $\{1\}$ also exclusion-blocks $\mu$ via $\mu^1$, but reaching $\sigma$ requires the three-step path above. Thus, external stability does not itself provide a dynamic route to the exclusion core.\hfill~$\diamond$
\end{exam}

\section{Concluding remarks}\label{sec.conclude}

We study the exclusion core in mixed-ownership object allocation problems and introduce the serial top trading cycles algorithm as a tool for analyzing it. For every initial allocation and every priority order, the STTC algorithm produces an exclusion-core-stable allocation. When the STTC algorithm is initialized at the null allocation, the set of its outcomes across priority orders coincides with the exclusion core. A single properly co-owned object is enough for this coincidence to fail. Among domains defined by restrictions on endowment structures, mixed ownership is therefore the largest domain on which the STTC algorithms generate exactly the exclusion core. The algorithm also provides a fixed-point test for exclusion core stability.\par\smallskip

The sequence of cycles generated by the STTC algorithm provides a common constructive basis for two further stability results. Locally, the agents who improve at successive non-$\mu$-cycle steps generate a self-contained exclusion-blocking path to the exclusion core. Globally, the agents who participate in non-$\mu$-cycles and prefer their final STTC allotments to their initial allotments form a cumulative coalition that exclusion-blocks the initial allocation. The path result also has a dynamic consequence. When the blocking carried out at each step is left to chance---with selection probabilities that depend only on the current allocation and give every available post-blocking allocation positive probability---the process reaches the exclusion core in finitely many steps with probability one.\par\smallskip

Figure~\ref{fig:endowmentclasses} summarizes how far our results extend beyond mixed ownership. Proper co-ownership marks a sharp boundary for the characterization and the fixed-point test. In the presence of a single properly co-owned object, an exclusion-core-stable allocation may lie outside the range of every owner-respecting flexible TTC algorithm, and the fixed-point test may fail. The two stability results reach further. Appendix~\ref{appendix:acyclic} characterizes the acyclic endowment structures, those in which either no object is properly co-owned or all objects share a single set of owners, and shows that both results continue to hold on this larger class. Beyond acyclicity, both results may fail, and they may fail together (Example~\ref{exam:noexstable}). Acyclicity is nevertheless necessary for neither result: Example~\ref{exam:beyondacyclic} gives a non-acyclic endowment structure for which, at every preference profile, the exclusion core is externally stable and reachable from every allocation by a self-contained exclusion-blocking path. In sum, acyclicity is sufficient for the two stability results and cannot simply be dropped. Since the only acyclic endowment structures beyond mixed ownership are those under which a single group co-owns every object, and since mixed ownership is the maximal domain for the STTC characterization and the fixed-point test, these findings support our focus on mixed-ownership problems.

\begin{figure}[!htbp]
  \centering
  \begin{tikzpicture}[font=\footnotesize]
    \draw (0,0) ellipse (7.4 and 3.9);
    \node[fill=white,inner sep=2pt] at (0,3.9) {All endowment structures};
    \draw[dashed] (-0.5,-0.5) ellipse (5.6 and 2.9);
    \node[fill=white,inner sep=2pt,font=\scriptsize\itshape] at (-0.5,2.4) {vNM stability and paths at every profile};
    \draw (-1.8,-0.5) ellipse (3.8 and 2.2);
    \node[fill=white,inner sep=2pt] at (-1.8,1.7) {Acyclic};
    \draw[thick] (-2.6,-0.7) ellipse (2.2 and 1.3);
    \node[fill=white,inner sep=2pt] at (-2.6,0.6) {Mixed ownership};
    \node[align=center,font=\scriptsize] at (0.85,-0.55) {Uniform\\co-ownership};
    \fill (3.4,-0.5) circle (2pt) node[below=2pt,font=\scriptsize] {Ex.~\ref{exam:beyondacyclic}};
    \node[font=\normalsize] at (6.0,0.2) {$\times$};
    \node[below=4pt,font=\scriptsize] at (6.0,0.2) {Ex.~\ref{exam:noexstable}};
  \end{tikzpicture}
  \caption{Endowment structures and the reach of the main results. By Proposition~\ref{prop.acyclicchar}, an endowment structure is acyclic exactly when it has mixed ownership or uniform co-ownership; the acyclic ellipse therefore consists of the mixed-ownership ellipse and the ring of uniformly co-owned structures around it. Under mixed ownership, the exclusion core is the set of STTC allocations generated from the null allocation (Proposition~\ref{prop.corecharacterization}), the fixed-point test applies (Theorem~\ref{theorem:coretest}), and the exclusion core is von Neumann--Morgenstern stable and reachable by self-contained exclusion-blocking paths (Theorems~\ref{theorem:vNMstability} and~\ref{theorem:pathtoexclusion}). The last two properties extend to uniform co-ownership (Theorems~\ref{thm.uniformvNM} and~\ref{thm.uniformpath}) and hence to all acyclic endowment structures (Corollaries~\ref{cor.acyclicvNM} and~\ref{cor.acyclicpath}). The dashed region contains the endowment structures for which both properties hold at every preference profile. Its boundary is not characterized. We only know that the dashed region contains every acyclic endowment structure and the structure of Example~\ref{exam:beyondacyclic}, and that it excludes the structure of Example~\ref{exam:noexstable}. For every endowment structure, owner-respecting flexible TTC allocations are exclusion core stable (Remark~\ref{remark2}).}
  \label{fig:endowmentclasses}
\end{figure}

\newpage
\appendix

\section*{Appendices}

\section{Proofs of the results in Section~\ref{sec.pre}}\label{appendix:prelim}

\begin{proof}[\textbf{Proof of Lemma~\ref{lemm.evictionproperties}}]
Fix a problem $R$, an allocation $\mu$, and coalitions $S$ and $S'$. Recall from Section~\ref{subsec.core} that $S^\mu_0:=S$ and that, for each $k\geq0$, $S^\mu_{k+1}:=S^\mu_k\cup\mu^{-1}\bigl(\omega(S^\mu_k)\bigr)$. This sequence is increasing, and since $N$ is finite it is constant from some index on. Let $\ell$ be the smallest such index, so that $\Omega(S\mid\mu)=\omega(S^\mu_\ell)$. The four parts assert, respectively, that $\Omega(S\mid\mu)$ contains the direct eviction set, is monotone in the coalition, is closed under the endowments of agents whose allotments it contains, and propagates along chains of dependence.

We first record that the endowment function $\omega$ is \textbf{monotone}: for all coalitions $T\subseteq T'$, $\omega(T)\subseteq\omega(T')$. (This is distinct from the monotonicity of $\Omega(\cdot\mid\mu)$ asserted in part~(b).) If $T'=N$, then $\omega(T')=O\supseteq\omega(T)$. If $T'\subsetneq N$, then $\omega(T)=\bigcup_{i\in T}\omega_i \subseteq \bigcup_{i\in T'}\omega_i=\omega(T')$.

\proofcase{(a)}Since $S=S^\mu_0\subseteq S^\mu_\ell$, by monotonicity of $\omega$, $\omega(S)\subseteq\omega(S^\mu_\ell)=\Omega(S\mid\mu)$.

\proofcase{(b)}Suppose that $S\subseteq S'$, and let $\ell'$ be the smallest index from which $(S'^\mu_k)_{k\geq0}$ is constant. We show by induction that, for every $k\geq0$, $S^\mu_k\subseteq S'^\mu_k$. For $k=0$, $S^\mu_0=S\subseteq S'=S'^\mu_0$. For the induction step, let $k\geq0$ and suppose that $S^\mu_k\subseteq S'^\mu_k$. Then, by monotonicity of $\omega$, $\omega(S^\mu_k)\subseteq\omega(S'^\mu_k)$. Hence, $\mu^{-1}(\omega(S^\mu_k))\subseteq\mu^{-1}(\omega(S'^\mu_k))$. Thus, $S^\mu_{k+1}\subseteq S'^\mu_{k+1}$. Setting $\bar k:=\max\{\ell,\ell'\}$, we obtain $\Omega(S\mid\mu)=\omega(S^\mu_{\bar k}) \subseteq\omega(S'^\mu_{\bar k}) =\Omega(S'\mid\mu)$.

\proofcase{(c)}Let $i\in N$ and suppose that $\mu(i)\in\Omega(S\mid\mu)=\omega(S^\mu_\ell)$. Then $i\in\mu^{-1}(\omega(S^\mu_\ell))\subseteq S^\mu_{\ell+1}=S^\mu_\ell$, and therefore $\omega_i\subseteq\omega(S^\mu_\ell)=\Omega(S\mid\mu)$.

\proofcase{(d)}Let $q\geq1$, let agents $i^1,\ldots,i^q$ and real objects $o^1,\ldots,o^q$ form a chain of dependence, and suppose that $S$ reaches $i^q$; that is, $i^q\in S$ or $\omega_{i^q}\subseteq\Omega(S\mid\mu)$. We show by downward induction on $r$ that, for each $r\in\{1,\ldots,q\}$, $o^r\in\Omega(S\mid\mu)$. Let $r=q$. If $i^q\in S$, then by~(a), $o^q\in\omega_{i^q}\subseteq\omega(S)\subseteq\Omega(S\mid\mu)$. If $\omega_{i^q}\subseteq\Omega(S\mid\mu)$, then $o^q\in\omega_{i^q}\subseteq\Omega(S\mid\mu)$. Now let $r<q$ and suppose that $o^{r+1}\in\Omega(S\mid\mu)$. Since $\mu(i^r)=o^{r+1}$, by part~(c), $\omega_{i^r}\subseteq\Omega(S\mid\mu)$. Since $o^r\in\omega_{i^r}$, it follows that $o^r\in\Omega(S\mid\mu)$. In particular, $o^1\in\Omega(S\mid\mu)$.
\end{proof}

\begin{proof}[\textbf{Proof of Lemma~\ref{lemm.paretoexclusion}}]
Let $R$ be a problem and let $\mu$ and $\sigma$ be allocations such that $\sigma$ Pareto dominates $\mu$. Define $\widehat S:=\{i\in N\mid \sigma(i)\mathbin{P_i}\mu(i)\}$.
Since $\sigma$ Pareto dominates $\mu$, $\widehat S\neq\emptyset$. Moreover, every agent in $\widehat S$ is strictly better off at $\sigma$. Since no agent outside $\widehat S$ is worse off at $\sigma$, the exclusion-enforcement condition is vacuously satisfied. Therefore, $\widehat S$ exclusion-blocks $\mu$ via $\sigma$.
\end{proof}

\section{Proofs of the results in Section~\ref{sec.sttc}}\label{appendix:sttc}

\begin{proof}[\textbf{Proof of Proposition~\ref{prop.sttc-exclusion}}]
Fix a problem $R$, a priority order $\rhd$, and an initial allocation $\mu$, and set $\mu^{\rhd}:=\STTC(R,\mu,\rhd)$.
Suppose, contrary to the proposition, that a coalition $S$ exclusion-blocks $\mu^{\rhd}$ via an allocation $\sigma$. Then:
\begin{enumerate}
\item each agent in $S$ is strictly better off at $\sigma$, i.e., for each $i \in S$, $\sigma(i) \mathbin{P_i} \mu^{\rhd}(i)$; and
\item if, for some $j \not\in S$, $\mu^{\rhd}(j) \mathbin{P_j} \sigma(j)$, then $\mu^{\rhd}(j) \in \Omega(S \mid \mu^{\rhd})$.
\end{enumerate}

We first show that a coalition consisting only of agents who are still present at Step~$t$ has no exclusion rights over any object that has already been assigned. A present agent's endowment cannot have been removed while he remains, so the recursion defining the extended eviction set never leaves the set of agents present at Step~$t$. Formally, for every Step~$t$ at which the remaining sets $N^t$ and $O^t$ are defined, let $O^{<t}$ denote the set of real objects removed before Step~$t$; that is, $O^{<t}:=O\setminus O^t$. We show that, for every set of agents $T\subseteq N^t$,\nopagebreak\vspace{-0.2cm}
\begin{equation}
\label{eq:remaining-eviction-invariant}
\Omega(T\mid\mu^\rhd)\cap O^{<t}=\emptyset.\vspace{-0.2cm}
\end{equation}

Since $O^{<1}=\emptyset$, the claim is immediate for $t=1$. Consider $t\geq2$. Write $T^{\mu^\rhd}_0:=T$ and $T^{\mu^\rhd}_{k+1}:=T^{\mu^\rhd}_k\cup(\mu^\rhd)^{-1}\bigl(\omega(T^{\mu^\rhd}_k)\bigr)$ for the recursion of Section~\ref{subsec.core} applied to $T$ at $\mu^\rhd$. We show by induction on $k$ that every agent in this recursion remains present at Step~$t$: for every $k\geq0$, $T^{\mu^\rhd}_k\subseteq N^t$.
Since $T^{\mu^\rhd}_0=T\subseteq N^t$, this is true for $k=0$.

Suppose that $T^{\mu^\rhd}_k\subseteq N^t$. Since $N^t$ is a proper subset of $N$, the direct eviction set $\omega(T^{\mu^\rhd}_k)$ consists only of objects individually owned by agents in $T^{\mu^\rhd}_k$. Every such object remains available at Step~$t$. Indeed, while its owner remains, the object points to that owner, so it could not have been removed in an earlier cycle without its owner also leaving in that cycle.

Consequently, any agent who receives an object in $\omega(T^{\mu^\rhd}_k)$ at the final allocation $\mu^\rhd$ also leaves no earlier than Step~$t$. Hence, $(\mu^\rhd)^{-1}\bigl(\omega(T^{\mu^\rhd}_k)\bigr)\subseteq N^t$, and therefore $T^{\mu^\rhd}_{k+1}\subseteq N^t$. This completes the induction.

The sequence $(T^{\mu^\rhd}_k)_{k\geq0}$ is increasing, and since $N$ is finite it is constant from some index on. Let $\ell_T$ be the smallest such index. All agents in the terminal set $T^{\mu^\rhd}_{\ell_T}$ are therefore still present at Step~$t$. Every object in $\Omega(T\mid\mu^\rhd)=\omega(T^{\mu^\rhd}_{\ell_T})$ is individually owned by one of these agents and hence also remains available at Step~$t$. This proves \eqref{eq:remaining-eviction-invariant}.

By going through the steps of the $(\rhd,\mu)$-STTC algorithm, we show that agents receiving their STTC allotments are not members of the exclusion-blocking coalition $S$.

\paragraph{$\bm{(\rhd,\mu)}$-STTC algorithm:}

\paragraph{Step~1.} At Step~1, every assigned agent receives his most-preferred object in $O \cup \{\varnothing\}$. Hence, no such agent is strictly better off at $\sigma$ than at $\mu^{\rhd}$, and none belongs to $S$. Hence, $S\subseteq N^2$.
Applying \eqref{eq:remaining-eviction-invariant} with $T=S$ and $t=2$ shows that no object removed at Step~1 belongs to $\Omega(S\mid\mu^\rhd)$. In particular, if agent $j$ is assigned his $\mu^\rhd$-allotment at Step~1, then $\mu^\rhd(j)\notin\Omega(S\mid\mu^\rhd)$.
We proceed inductively.

\paragraph{Step~$\bm{t \ge 2}$.}
Recall that no agent receiving his $\mu^{\rhd}$-allotment at Steps~$1,\ldots,t-1$ can belong to the exclusion-blocking coalition $S$. Hence, coalition $S$ is a subset of agents remaining after Step~$t-1$, i.e., $S\subseteq N^t$. Applying \eqref{eq:remaining-eviction-invariant} with $T=S$ shows that no object removed before Step~$t$ belongs to $\Omega(S\mid\mu^\rhd)$. Thus, if $j$ is assigned his $\mu^\rhd$-allotment at one of Steps~$1,\ldots,t-1$, then
$\mu^\rhd(j)\notin\Omega(S\mid\mu^\rhd)$.

At Step~$t$, every assigned agent receives his most-preferred object in $O^t \cup \{\varnothing\}$. We show that no such agent is strictly better off at $\sigma$.
Suppose, to the contrary, that there exists an agent $i\in N^t$ who is assigned object $\mu^\rhd(i)$ at Step~$t$ and satisfies $\sigma(i)\mathbin{P_i}\mu^\rhd(i)$.

Set $i^0:=i$ and $\tau_0:=t$. We construct recursively agents $i^0,i^1,\ldots$ and steps $\tau_0>\tau_1>\cdots$ at which they are assigned. Suppose that $i^k$ and $\tau_k$ have been constructed. If $\mu^\rhd(i^k)\mathbin{R_{i^k}}\sigma(i^k)$, the construction stops. Otherwise, $\sigma(i^k)\mathbin{P_{i^k}}\mu^\rhd(i^k)$; since $i^k$ receives his most-preferred object in $O^{\tau_k}\cup\{\varnothing\}$, the object $\sigma(i^k)$ was not available at Step~$\tau_k$, so it is a real object assigned at some earlier Step~$\tau_{k+1}<\tau_k$, say to agent $i^{k+1}$, and $\mu^\rhd(i^{k+1})=\sigma(i^k)$. Since the steps are positive integers and strictly decreasing, the construction stops after finitely many rounds, say at index $m$, with $\mu^\rhd(i^m)\mathbin{R_{i^m}}\sigma(i^m)$. By assumption $\sigma(i^0)\mathbin{P_{i^0}}\mu^\rhd(i^0)$, so the construction does not stop at $k=0$; hence $m\geq1$.

By construction, $\mu^\rhd(i^m)=\sigma(i^{m-1})$, which is a real object.
Since $\sigma$ is an allocation, this object cannot also be assigned to $i^m$ at $\sigma$. Hence,
$\sigma(i^m)\neq\mu^\rhd(i^m)$. Since $P_{i^m}$ is the strict part of $R_{i^m}$, it follows that $\mu^\rhd(i^m)\mathbin{P_{i^m}}\sigma(i^m)$.

Agent $i^m$ received his STTC allotment at a step strictly before Step~$t$. By the induction hypothesis, $i^m\notin S$ and $\mu^\rhd(i^m)\notin\Omega(S\mid\mu^\rhd)$.
However, since coalition $S$ exclusion-blocks $\mu^\rhd$ via $\sigma$ and $\mu^\rhd(i^m)\mathbin{P_{i^m}}\sigma(i^m)$, we must have $\mu^\rhd(i^m)\in\Omega(S\mid\mu^\rhd)$, a contradiction.

We conclude that no Step~$t$ agent can be strictly better off at allocation $\sigma$, and therefore no agent receiving his $\mu^{\rhd}$-allotment at Step~$t$ can belong to the exclusion-blocking coalition $S$. Hence, $S\subseteq N^{t+1}$.
Applying \eqref{eq:remaining-eviction-invariant} with $T=S$ and the remaining sets $N^{t+1}$ and $O^{t+1}$ shows that no object removed during Steps~$1,\ldots,t$ belongs to $\Omega(S\mid\mu^\rhd)$. In particular, if agent $j$ is assigned his $\mu^\rhd$-allotment during one of these steps, then $\mu^\rhd(j)\notin\Omega(S\mid\mu^\rhd)$.

\paragraph{Agents remaining at termination.} Suppose the algorithm stops after Step~$t$. If $N^{t+1}=\emptyset$, then $S\subseteq N^{t+1}=\emptyset$, contradicting that $S$ is a coalition.

Suppose instead that $N^{t+1}\neq\emptyset$. Then $O^{t+1}=\emptyset$, and each remaining agent $i\in N^{t+1}$ is assigned $\mu^\rhd(i)=\varnothing$, his most-preferred object in $O^{t+1}\cup\{\varnothing\}=\{\varnothing\}$. If some remaining agent $i$ belonged to $S$, then $\sigma(i)\mathbin{P_i}\varnothing$, so $\sigma(i)$ would be a real object. Since $O^{t+1}=\emptyset$, this object was removed at one of Steps~$1,\ldots,t$ and assigned to an agent at that step. The chain argument from Step~$t$, applied with $i^0:=i$ and $\tau_0:=t+1$, again yields a contradiction.

Every agent is either assigned at one of the steps of the $(\rhd,\mu)$-STTC algorithm or remains when the algorithm terminates. The preceding arguments show that neither type of agent can belong to $S$. Hence, $S=\emptyset$, contradicting that $S$ is a coalition.
\end{proof}

The following auxiliary lemma isolates an argument used in the proof of Lemma~\ref{lemm.fixedpoint} and twice in the joint proof of Theorems~\ref{theorem:vNMstability} and~\ref{theorem:pathtoexclusion}. It converts an STTC cycle into a chain of dependence of the type required by Lemma~\ref{lemm.evictionproperties}(d).

The idea is to walk around the cycle in the direction of its arrows, starting from the object that a displaced outsider loses. The agent who takes this object did not hold it before. He receives his most-preferred object among those still available, and his current allotment is also still available. Since preferences are strict, he therefore strictly improves. Hence, at least one agent of the cycle strictly improves; consider the first such agent on the walk. Every agent met before him receives exactly what he already holds. Thus, each object on the walk up to this agent is held by the agent just before it: by the displaced outsider for the first object, and by the preceding cycle agent otherwise. Such an object cannot be public, because an assigned public object points to its holder, whereas this object points to the next agent on the walk. Hence, it is individually owned by the agent it points to. This is the one place where the edge rule for assigned public objects in Definition~\ref{defi.graph} is essential. The walk therefore stops at the first strictly improving agent and delivers the required chain of dependence.

\begin{lemm}[\textbf{Cycle chains}]\label{lemm.cyclechain}
Consider Step~$t$ of the $(\rhd,\mu)$-STTC algorithm, and let $\mathcal{C}$ be an ordinary cycle forming at that step. Write $\mu^\rhd:=\STTC(R,\mu,\rhd)$. Let $\mu'$ be the reduced allocation at Step~$t$, which assigns each agent $i\in N^t$ his allotment $\mu(i)$ if $\mu(i)\in O^t\cup\{\varnothing\}$, and $\varnothing$ otherwise. Let $\nu$ be an allocation that agrees with $\mu'$ on $N^t$. Define the \textbf{improvers of $\bm{\mathcal{C}}$ relative to $\bm{\nu}$} by $T(\mathcal{C},\nu):=\{h\in S(\mathcal{C})\mid \mu^\rhd(h)\mathbin{P_h}\nu(h)\}$.
Suppose that there is an agent $j\in N^t\setminus S(\mathcal{C})$ such that $\nu(j)$ is an object of $\mathcal{C}$. Let $m$ be the number of agents of $\mathcal{C}$ and set $o^1:=\nu(j)$. Enumerate $\mathcal{C}$ from that object as $o^1\rightarrow i^1\rightarrow o^2\rightarrow i^2\rightarrow\cdots\rightarrow o^m\rightarrow i^m\rightarrow o^1$.
Then $T(\mathcal{C},\nu)\neq\emptyset$ and, letting $q$ be the smallest index with $i^q\in T(\mathcal{C},\nu)$,
\begin{enumerate}
\item[\textup{(a)}] for each $r\in\{1,\ldots,q\}$, $o^r\in\omega_{i^r}$; and
\item[\textup{(b)}] for each $r\in\{1,\ldots,q-1\}$, $\nu(i^r)=\mu^\rhd(i^r)=o^{r+1}$.
\end{enumerate}
\end{lemm}

\begin{proof}[\textbf{Proof of Lemma~\ref{lemm.cyclechain}}]
Fix Step~$t$ of the $(\rhd,\mu)$-STTC algorithm, an ordinary cycle $\mathcal{C}$ forming at that step, an allocation $\nu$ that agrees with the reduced allocation $\mu'$ on $N^t$, and an agent $j\in N^t\setminus S(\mathcal{C})$ such that $\nu(j)$ is an object of $\mathcal{C}$. We adopt the enumeration $o^1\rightarrow i^1\rightarrow\cdots\rightarrow o^m\rightarrow i^m\rightarrow o^1$ of the statement.

Each $h\in S(\mathcal{C})$ points to his most-preferred object in $O^t\cup\{\varnothing\}$ and receives it. Thus, $\mu^\rhd(h)$ is this object. Because $\nu(h)$ is his reduced allotment, it also belongs to $O^t\cup\{\varnothing\}$. Hence, for each $h\in S(\mathcal{C})$,\nopagebreak\vspace{-0.2cm}
\begin{equation}\label{eq:cyclechain-weak}
\mu^\rhd(h)\mathbin{R_h}\nu(h).\vspace{-0.2cm}
\end{equation}
By antisymmetry, every $h\in S(\mathcal{C})\setminus T(\mathcal{C},\nu)$ therefore satisfies $\mu^\rhd(h)=\nu(h)$.

Agent $i^m$ points to $o^1$ in $\mathcal{C}$, so $\mu^\rhd(i^m)=o^1=\nu(j)$. Since $\nu$ is an allocation, $j\neq i^m$, and $o^1$ is a real object, we have $\nu(i^m)\neq o^1=\mu^\rhd(i^m)$. Together with \eqref{eq:cyclechain-weak} and antisymmetry, $\mu^\rhd(i^m)\mathbin{P_{i^m}}\nu(i^m)$, so $i^m\in T(\mathcal{C},\nu)$. Hence, $T(\mathcal{C},\nu)\neq\emptyset$ and $q$ is well defined with $q\leq m$.

We prove part (b) first because it is used in the proof of part (a).

\proofcase{(b)}Let $r<q$. Since $q$ is the smallest index with $i^q\in T(\mathcal{C},\nu)$, we have $i^r\notin T(\mathcal{C},\nu)$, and therefore $\mu^\rhd(i^r)=\nu(i^r)$. Since $i^r$ points to $o^{r+1}$ in $\mathcal{C}$, we also have $\mu^\rhd(i^r)=o^{r+1}$.

\proofcase{(a)}Agent $j$ and the agents of $\mathcal{C}$ are all present at Step~$t$, so $|N^t|\geq2$ and, for each $h\in N^t$, by Definition~\ref{defi.residualgraph}, $\omega'_h=\omega_h\cap O^t\subseteq\omega_h$. Fix $r\in\{1,\ldots,q\}$. If $r=1$, the reduced allocation assigns $o^1=\nu(j)$ to agent $j$; if $r\geq2$, then by~(b), it assigns $o^r=\nu(i^{r-1})$ to agent $i^{r-1}$. In both cases $o^r$ is assigned in the reduced allocation to an agent different from $i^r$. Were $o^r$ public in the reduced instance, it would point to that agent rather than to $i^r$. Hence, $o^r$ is individually owned in the reduced instance by the agent it points to, i.e., $o^r\in\omega'_{i^r}\subseteq\omega_{i^r}$.
\end{proof}

The next lemma establishes a fixed-point property of exclusion-core-stable
allocations. It is used in the proof of Proposition~\ref{prop.corecharacterization} to recover each exclusion-core-stable allocation as an STTC outcome from the null allocation, and in the proof of
Theorem~\ref{theorem:coretest} to obtain the fixed-point characterization.

\begin{lemm}[\textbf{Fixed-point property of exclusion-core-stable allocations}]\label{lemm.fixedpoint}
Let $R$ be a problem and let $\mu$ be an exclusion-core-stable allocation. Then, for every priority order $\rhd$, \(\STTC(R,\mu,\rhd)=\mu\).
\end{lemm}

\begin{proof}[\textbf{Proof of Lemma~\ref{lemm.fixedpoint}}]
Let $R$ be a problem and let $\mu$ be an exclusion-core-stable allocation. Fix a priority order $\rhd$ and set $\mu^\rhd:=\STTC(R,\mu,\rhd)$. Suppose, to the contrary, that $\mu^\rhd\neq\mu$.

We first show that the algorithm executes at least one non-$\mu$-cycle. Suppose that every executed cycle is a $\mu$-cycle. Then each agent who leaves in a cycle receives his $\mu$-allotment, and a real object is removed only by its $\mu$-holder. An agent $i$ who never leaves in a cycle therefore has $\mu(i)=\varnothing$ and receives $\varnothing$. So $\mu^\rhd=\mu$, contrary to assumption.

Let $\mathcal{C}$ be an earliest non-$\mu$-cycle formed at Step~$t^\star$ in the STTC graph $\mathcal{G}^{t^\star}$. Every cycle executed before Step~$t^\star$ is a $\mu$-cycle, so as above a real object removed before Step~$t^\star$ was taken by its $\mu$-holder. Hence, for every $i\in N^{t^\star}$, $\mu(i)\in O^{t^\star}\cup\{\varnothing\}$.

We now construct a coalition $\widehat{S}$ and an allocation $\widehat{\sigma}$ and show that $\widehat{S}$ exclusion-blocks $\mu$ via $\widehat{\sigma}$. First, define the coalition $\widehat{S} := \{i \in S(\mathcal{C}) \mid \mu^{\rhd}(i) \mathbin{P_i} \mu(i)\}$ and show that $\widehat{S} \neq \emptyset$.

By the construction of $\mu^{\rhd}$ via the STTC algorithm, we know that at Step~$t^\star$, each agent in $S(\mathcal{C})$ points to, and receives, his most-preferred object in $O^{t^\star}\cup\{\varnothing\}$. Since $\mathcal{C}$ is a non-$\mu$-cycle, there exists a cycle-agent $i\in S(\mathcal{C})$ with $\mu^{\rhd}(i) \neq \mu(i)$. By the preceding observation, $\mu(i)\in O^{t^\star}\cup\{\varnothing\}$. Hence, agent $i$ pointing at $\mu^{\rhd}(i) \neq \mu(i)$ in the $(\rhd,\mu)$-STTC algorithm implies that $\mu^{\rhd}(i) \mathbin{P_i} \mu(i)$. Thus, $\widehat{S} \neq \emptyset$.

Second, define the allocation $\widehat{\sigma}$ by\nopagebreak\vspace{-0.2cm}
    \[\widehat{\sigma}(i) :=
    \begin{cases}
        \mu^{\rhd}(i) &\text{ if } i \in \widehat{S}, \\
        \varnothing &\text{ if } i \not\in \widehat{S} \text{ and there exists a $k \in \widehat{S}$ such that } \mu(i) = \mu^{\rhd}(k),\\
        \mu(i) &\text{ otherwise.}
    \end{cases} \vspace{-0.2cm}\]

We now show that coalition $\widehat{S}$ exclusion-blocks $\mu$ via $\widehat{\sigma}$. By the definition of $\widehat{S}$, for each $i \in \widehat{S}$, $\widehat{\sigma}(i) \mathbin{P_i} \mu(i)$.  By the preceding observation, for each cycle-agent $i\in S(\mathcal{C})$, $\mu(i)\in O^{t^\star}\cup\{\varnothing\}$. Hence, for each cycle-agent $i\in S(\mathcal{C})$, $\mu^{\rhd}(i) \mathbin{R_i} \mu(i)$.

If $\widehat{\sigma}$ Pareto dominates $\mu$, no agent outside $\widehat{S}$ is worse off at $\widehat{\sigma}$. Therefore, the exclusion-enforcement condition is vacuously satisfied. Since every member of $\widehat{S}$ is strictly better off, $\widehat{S}$ exclusion-blocks $\mu$ via $\widehat{\sigma}$, establishing the desired contradiction.

If $\widehat{\sigma}$ does not Pareto dominate $\mu$, let $j\notin\widehat{S}$ be any agent such that $\mu(j)\mathbin{P_j}\widehat{\sigma}(j)$. Since every member of $\widehat{S}$ is strictly better off at $\widehat{\sigma}$, such an agent exists.
By the definition of $\widehat{\sigma}$, $\widehat{\sigma}(j)=\varnothing$, and there exists an agent $i\in\widehat{S}$ such that $\mu(j)=\mu^\rhd(i)$.
In particular, $\mu(j)$ is a real object.

We first show that $j\notin S(\mathcal{C})$. Recall that every cycle agent $h\in S(\mathcal{C})$ satisfies $\mu^\rhd(h)\mathbin{R_h}\mu(h)$.
If $j\in S(\mathcal{C})\setminus\widehat{S}$, then strictness of preferences implies $\mu^\rhd(j)=\mu(j)$.
But $\mu^\rhd(i)=\mu(j)$ for the distinct agent $i\in\widehat{S}$, contradicting feasibility of $\mu^\rhd$.
Hence, $j\notin S(\mathcal{C})$.

We next show that agent $j$ is still present at Step~$t^\star$.
Suppose that $j$ left at an earlier step. Since $\mathcal{C}$ is an earliest non-$\mu$-cycle, agent $j$ then received his $\mu$-allotment, so $\mu^\rhd(j)=\mu(j)$.
Consequently, object $\mu(j)$ was removed before Step~$t^\star$.
But there exists a cycle agent $i\in\widehat{S}$ such that $\mu(j)=\mu^\rhd(i)$. Since agent $i$ receives this object in $\mathcal{C}$ at Step~$t^\star$, the object must still be available at that step, a contradiction.
Therefore, $j\in N^{t^\star}$.

It remains to show that $\mu(j)\in\Omega(\widehat{S}\mid \mu)$.

Let $\mu'$ denote the reduced allocation at Step~$t^\star$, which assigns each agent $h\in N^{t^\star}$ his allotment $\mu(h)$ if $\mu(h)\in O^{t^\star}\cup\{\varnothing\}$, and $\varnothing$ otherwise. Since $\mu(h)\in O^{t^\star}\cup\{\varnothing\}$ for every $h\in N^{t^\star}$, we have $\mu'(h)=\mu(h)$ for every such $h$. Hence, $\mu$ agrees with $\mu'$ on $N^{t^\star}$.

Moreover, there exists a cycle agent $i\in\widehat{S}$ such that $\mu(j)=\mu^\rhd(i)$. Hence, $\mu(j)$ is a real object in $\mathcal{C}$; by the convention of Section~\ref{subsec.algorithm}, a cycle through the null object contains no real object, so $\mathcal{C}$ is ordinary. We have also shown that $j\in N^{t^\star}\setminus S(\mathcal{C})$ and that $T(\mathcal{C},\mu)=\{h\in S(\mathcal{C})\mid\mu^\rhd(h)\mathbin{P_h}\mu(h)\}=\widehat{S}$.
Therefore, Lemma~\ref{lemm.cyclechain} applies with $\nu:=\mu$ and yields an enumeration $o^1\rightarrow i^1\rightarrow\cdots\rightarrow o^m\rightarrow i^m\rightarrow o^1$, where $o^1=\mu(j)$, and an index $q\in\{1,\ldots,m\}$ such that $i^q\in\widehat{S}$ and, for each $r\leq q$, $o^r\in\omega_{i^r}$, and, for each $r<q$, $\mu(i^r)=o^{r+1}$. Therefore, by Lemma~\ref{lemm.evictionproperties}(d), applied with $S:=\widehat{S}$, $\mu(j)=o^1\in\Omega(\widehat{S}\mid\mu)$.

Since $j$ was arbitrary among the agents outside $\widehat{S}$ who are worse off at $\widehat{\sigma}$, coalition $\widehat{S}$ exclusion-enforces $\widehat{\sigma}$ at $\mu$. Since every member of $\widehat{S}$ is strictly better off, $\widehat{S}$ exclusion-blocks $\mu$ via $\widehat{\sigma}$, contradicting exclusion core stability of $\mu$. Therefore, $\mu^\rhd=\mu$.
\end{proof}

\begin{proof}[\textbf{Proof of Proposition~\ref{prop.corecharacterization}}]
Fix a problem $R$.
By Proposition~\ref{prop.sttc-exclusion}, the set of allocations $\{\STTC(R,\mu^{\varnothing},\rhd)\mid \rhd\in\Delta\}$ is a subset of the exclusion core.

We now prove that the exclusion core is a subset of the set of allocations $\{\STTC(R,\mu^{\varnothing},\rhd)\mid \rhd\in\Delta\}$. That is, we show that if an allocation $\mu$ belongs to the exclusion core, then there exists a priority order $\rhd$ such that $\mu$ is the outcome of the $(\rhd,\mu^\varnothing)$-STTC algorithm.

Let $\mu$ be an allocation in the exclusion core. Consider an arbitrary priority order $\rhd^*$. Without loss of generality, let $1\rhd^* 2\rhd^*\cdots\rhd^* n$. By Lemma~\ref{lemm.fixedpoint}, $\STTC(R,\mu,\rhd^*)=\mu$.
Consequently, every cycle executed by the $(\rhd^*,\mu)$-STTC algorithm is a $\mu$-cycle. We use the steps of this algorithm to construct a priority order $\rhd$ such that the $(\rhd,\mu^\varnothing)$-STTC algorithm also produces $\mu$.

\paragraph{$\bm{(\rhd^{*},\mu)}$-STTC algorithm:}

\paragraph{Step~1.} Starting with allocation $\mu$, consider Step~1 of the $(\rhd^*,\mu)$-STTC algorithm. Let $\mathscr{C}^1$ denote the set of all cycles that form, and define $S(\mathscr{C}^1):=\bigcup_{\mathcal C\in\mathscr{C}^1}S(\mathcal C)$. Adjust the given priority order  $\rhd^*$ by moving the agents in $S(\mathscr{C}^1)$ above the remaining agents in $N^2=N\setminus S(\mathscr{C}^1)$, preserving the relative order of agents in $S(\mathscr{C}^1)$ and in $N\setminus S(\mathscr{C}^1)$ according to $\rhd^{*}$. Denote the adjusted priority order by $\rhd^1$. We proceed inductively.

\paragraph{Step~$\bm{t \ge 2}$.} Given the adjusted order $\rhd^{t-1}$ at the end of Step~$t-1$ of the $(\rhd^*,\mu)$-STTC algorithm, consider Step~$t$ of the $(\rhd^*,\mu)$-STTC algorithm. Let $\mathscr{C}^t$ denote the set of all cycles that form, and define $S(\mathscr{C}^t):=\bigcup_{\mathcal C\in\mathscr{C}^t}S(\mathcal C)$.
Adjust the current priority order $\rhd^{t-1}$ by moving the agents in $S(\mathscr{C}^t)$ right below the agents in $S(\mathscr{C}^{t-1})$, and above the remaining agents $N^{t+1}=N\setminus \bigcup_{i=1}^{t}S(\mathscr{C}^{i})$, preserving the relative order of agents  in $S(\mathscr{C}^{t})$ and in $N\setminus S(\mathscr{C}^{t})$  according to $\rhd^{t-1}$. Denote the adjusted priority order by $\rhd^t$.

Finally, at the end of the $(\rhd^*,\mu)$-STTC algorithm, the resulting allocation is $\mu$, and we obtain a final adjusted priority order $\rhd$. By construction of $\rhd$, if $t<s$, then every agent in $S(\mathscr{C}^t)$ has higher priority than every agent in $S(\mathscr{C}^s)$. Any agents who remain when the $(\rhd^*,\mu)$-STTC algorithm exhausts the set of real objects are ranked below all the sets $S(\mathscr{C}^t)$. Within each set, the relative order induced by $\rhd^*$ is preserved.\par\smallskip

To complete the proof, starting with the null allocation $\mu^\varnothing$, we apply the $(\rhd,\mu^\varnothing)$-STTC algorithm and show that $\STTC(R,\mu^\varnothing,\rhd)=\mu$.
Let $T$ be the last step at which the $(\rhd^*,\mu)$-STTC algorithm executes a cycle. For each agent $i$, define $\tau(i)$ as follows. If there exists $t\in\{1,\ldots,T\}$ such that $i\in S(\mathscr{C}^t)$, let $\tau(i):=t$. If $i$ receives the null object when the algorithm stops because no real object remains, let $\tau(i):=T+1$. Thus, $\tau(i)<\tau(j)$ implies $i\rhd j$.

We prove by induction over the steps of the $(\rhd,\mu^\varnothing)$-STTC algorithm that every departing agent receives his $\mu$-allotment. Consider Step~$t$, and suppose that every agent who departed at an earlier step received his $\mu$-allotment. We first show that the $\mu$-allotment of every agent remaining at Step~$t$ is still available. Indeed, by the induction hypothesis a real object removed earlier was taken by its $\mu$-holder, who has departed; so the $\mu$-allotment of a remaining agent cannot have been removed. The null object is always available. Therefore, for every agent $i$ remaining at Step~$t$, $\mu(i)\in O^t\cup\{\varnothing\}$. Consider a cycle $\mathcal{C}$ in the current STTC graph.

Suppose first that exactly one agent $i$ remains. Then, by Definition~\ref{defi.residualgraph}, $\omega'_i=O'$ and every remaining object is individually owned by agent~$i$ in the reduced instance, and no remaining object is public there. Reduced ownership therefore carries no information about original ownership, and we argue directly. Suppose that agent~$i$ points to an object $o\neq\mu(i)$. Since $\mu(i)$ is available and agent~$i$ points to his most-preferred object in $O^t\cup\{\varnothing\}$, $o\mathbin{P_i}\mu(i)$.
In the $(\rhd^*,\mu)$-STTC algorithm, agent~$i$ received $\mu(i)$ while the null object was available, so $\mu(i)\mathbin{R_i}\varnothing$ and hence $o$ is a real object. Since $o\mathbin{P_i}\mu(i)$, object $o$ cannot have been available when agent~$i$ received $\mu(i)$; it was therefore removed at an earlier step of that algorithm. As every cycle executed in that algorithm is a $\mu$-cycle, there exists an agent $j$ with $\tau(j)<\tau(i)$ such that $o=\mu(j)$. Since agent~$i$ is the only remaining agent, agent $j$ has already departed, and by the induction hypothesis he departed with $\mu(j)=o$. Hence, $o$ is no longer available, a contradiction. Therefore, agent~$i$ points to $\mu(i)$ and receives it.

Assume henceforth that at least two agents remain, so that, for each remaining agent $k$, $\omega'_k=\omega_k\cap O'\subseteq\omega_k$, and reduced individual ownership implies original ownership.

First, suppose that $\mathcal{C}$ is not ordinary. By the convention of Section~\ref{subsec.algorithm}, $\mathcal{C}$ is then the self-cycle $i\rightarrow\varnothing\rightarrow i$ of a single agent~$i$, who points to $\varnothing$. Since $\mu(i)$ is available and the null object was also available when agent $i$ was assigned in the $(\rhd^*,\mu)$-STTC algorithm, strict preferences imply $\mu(i)=\varnothing$.
Hence, agent~$i$ receives his $\mu$-allotment.

Second, suppose that $\mathcal{C}$ is ordinary and contains an object that is public in the current reduced instance. Since the initial allocation is $\mu^\varnothing$, every reduced allocation assigns the null object to every remaining agent; hence no public object is assigned, and every public object points to the highest-priority remaining agent. Denote this agent by $h$. Then $\tau(h)$ is minimal among the $\tau$-values of the remaining agents. We claim that $h$ points to $\mu(h)$. Otherwise, $h$ points to some object $o\neq\mu(h)$ with $o\mathbin{P_h}\mu(h)$; since $\mathcal{C}$ is ordinary, $o$ is a real object.
Since $h$ received $\mu(h)$ in the $(\rhd^*,\mu)$-STTC algorithm, object $o$ was removed before Step~$\tau(h)$ of that algorithm. Hence, there exists an agent $j$ with $\tau(j)<\tau(h)$ such that $o=\mu(j)$. By the induction hypothesis and the minimality of $\tau(h)$ among the remaining agents, agent $j$ has already departed with object $o$, contradicting that $o$ is still available. Therefore, $h$ points to $\mu(h)$.

We now follow the cycle from $h$. If $\mu(h)$ is public in the current reduced problem, it points back to $h$. If $\mu(h)$ is individually owned by a remaining agent $j$, then $\mu(h)$ points to $j$. At Step~$\tau(h)$ of the $(\rhd^*,\mu)$-STTC algorithm, either $j$ had already departed or $j$ belonged to the same $\mu$-cycle as $h$. The first possibility would imply $\tau(j)<\tau(h)$, contradicting the minimality of $\tau(h)$ among the remaining agents. Hence, $\tau(j)=\tau(h)$.
The preceding argument then shows that $j$ points to $\mu(j)$. Continuing around the cycle proves that every agent in $\mathcal{C}$ receives his $\mu$-allotment.

Finally, suppose that $\mathcal{C}$ is ordinary and contains no object that is public in the current reduced instance. Thus, every object in $\mathcal{C}$ is individually owned by the agent to whom it points. Choose an agent $h\in S(\mathcal{C})$ with minimal $\tau(h)$ among the cycle agents. We again claim that $h$ points to $\mu(h)$. Suppose instead that $h$ points to an object $o\neq\mu(h)$. Since $\mu(h)$ is available, $o\mathbin{P_h}\mu(h)$.
Thus, object $o$ was removed before Step~$\tau(h)$ of the $(\rhd^*,\mu)$-STTC algorithm. Let $j$ be the owner to whom $o$ points in the current cycle. When $o$ was removed in the $(\rhd^*,\mu)$-STTC algorithm, agent $j$ had either departed at an earlier step or belonged to the cycle in which $o$ was removed. In the first case, $\tau(j)$ is smaller than the step at which $o$ was removed; in the second case, it equals that step. Since $o$ was removed before Step~$\tau(h)$, both cases imply $\tau(j)<\tau(h)$, contradicting the choice of $h$. Hence, $h$ points to $\mu(h)$.

Since $\mu(h)$ belongs to $\mathcal{C}$ and no object in $\mathcal{C}$ is public in the current reduced instance, there exists a remaining agent $j$ who individually owns $\mu(h)$. Hence, $\mu(h)$ points to $j$, and $j$ belongs to $\mathcal{C}$. At Step~$\tau(h)$ of the $(\rhd^*,\mu)$-STTC algorithm, agent $j$ had either already departed or belonged to the same $\mu$-cycle as $h$. If $j$ had already departed, then $\tau(j)<\tau(h)$, and since $j\in S(\mathcal{C})$, this contradicts the choice of $h$. Therefore, $j$ belonged to the same $\mu$-cycle as $h$, and $\tau(j)=\tau(h)$.
Repeating the argument around the cycle shows that every cycle agent points to, and receives, his $\mu$-allotment.

We have shown that at every step of the $(\rhd,\mu^\varnothing)$-STTC algorithm, all cycle agents receive their $\mu$-allotments. The induction therefore continues until the algorithm terminates. Suppose that no real object remains while some agents are still present. Suppose that a remaining agent $i$ had a real $\mu$-allotment. Since every removed object was taken by its $\mu$-holder, object $\mu(i)$ could only have been taken by $i$ himself, contradicting that $i$ remains. Hence, every remaining agent has $\mu(i)=\varnothing$ and is assigned the null object. Consequently, $\STTC(R,\mu^\varnothing,\rhd)=\mu$.
\end{proof}

\section{Proofs of the results in Section~\ref{sec.mainresult}}\label{appendix:main}

\begin{proof}[\textbf{Proof of Theorem~\ref{thm.strictsuperset}}]
Let $\omega$ be an endowment structure with a properly co-owned object $a\in O$. Thus, $a$ has at least two co-owners, and there exists a non-co-owner $k\in N$. Let $R$ be a preference profile in which each agent $i\in N$ satisfies $a\mathbin{P_i}\varnothing\mathbin{P_i}o$ for every object $o\in O\setminus\{a\}$; the relative ranking of the objects below $\varnothing$ is arbitrary. Let $\mu^a$ be the allocation at which object $a$ is assigned to a non-co-owner $k$, while all other agents receive $\varnothing$; that is, $\mu^a(k)=a$ and, for each $i\in N\setminus\{k\}$, $\mu^a(i)=\varnothing$.
The following two claims show that $\mu^a \in \EC(\omega,R)$ and $\mu^a \notin \STTC(\omega,R)$.

\begin{claim}\label{claim.muinEC}
    $\mu^a \in \EC(\omega,R)$.
\end{claim}
\begin{proof}[\textbf{Proof}] Suppose, to the contrary, that a coalition $S\subseteq N$ exclusion-blocks $\mu^a$ via an allocation $\sigma$. Since each agent ranks only $a$ above $\varnothing$, each member of $S$ must receive $a$ under $\sigma$ in order to be strictly better off. In particular, the non-co-owner $k$, who already receives $a$ at $\mu^a$, cannot belong to $S$. Since some member of $S$ receives $a$ at $\sigma$, agent $k$ is worse off at $\sigma$. Hence, exclusion-enforcement requires $a=\mu^a(k)\in\Omega(S\mid\mu^a)$.
We show that this implies that every co-owner of $a$ belongs to $S$. At $\mu^a$, every co-owner of $a$ receives $\varnothing$. The recursion defining $\Omega(S\mid\mu^a)$ adds an agent only when his $\mu^a$-allotment belongs to the current direct eviction set, which is a subset of $O$ and therefore never contains $\varnothing$. Consequently, a co-owner of $a$ who does not belong to $S$ cannot enter any of the recursively generated agent sets $S^{\mu^a}_k$. Since $a$ belongs to the direct eviction set of an agent set only if that set contains all co-owners of $a$, the inclusion $a\in\Omega(S\mid\mu^a)$ therefore implies that every co-owner of $a$ already belongs to $S$.

Since $a$ is indivisible, at most one agent can receive $a$ under $\sigma$. As $a$ has at least two co-owners, not all members of $S$ can be made strictly better off, contradicting that $S$ exclusion-blocks $\mu^a$. Therefore, $\mu^a\in\EC(\omega,R)$.
\end{proof}

\begin{claim} \label{claim.munotinSTCC}
    $\mu^a \notin \STTC(\omega,R)$.
\end{claim}
\begin{proof}[\textbf{Proof}] Fix an arbitrary priority order $\rhd$ and consider the first step of the $(\rhd,\mu^\varnothing)$-STTC algorithm. Since each agent ranks $a$ above $\varnothing$, all agents initially point to $a$. Because $a$ has at least two co-owners, at the first step it points to its highest-priority co-owner under $\rhd$. Thus, $a$ and its highest-priority co-owner form a self-cycle, which assigns $a$ to that co-owner. Once this self-cycle is executed, $a$ is removed, so it cannot later be assigned to the non-co-owner $k$. Because $\rhd$ was arbitrary, no priority order assigns $a$ to $k$. Since $\mu^a(k)=a$, it follows that $\mu^a \notin \STTC(\omega,R)$.
\end{proof}

The same argument applies to every owner-respecting flexible TTC algorithm described in Remark~\ref{remark2}. Initially, all agents point to $a$, while $a$ points to one of its co-owners. Thus, the first cycle containing $a$ assigns it to a co-owner and removes it from the problem.

The same construction also shows that the fixed-point characterization in Theorem~\ref{theorem:coretest} fails under proper co-ownership. Initialize an owner-respecting flexible TTC algorithm at $\mu^a$. All agents point to $a$. Because $a$ is assigned at $\mu^a$ to the non-co-owner $k$, the owner-respecting rule makes $a$ point to one of its co-owners. Hence, $a$ and that co-owner form a self-cycle at the first step, so $a$ is reassigned to the co-owner and removed. The resulting allocation differs from $\mu^a$, although Claim~\ref{claim.muinEC} establishes that $\mu^a$ is exclusion core stable.

By Remark~\ref{remark2}, for every priority order $\rhd$, $\STTC(\omega,R,\mu^\varnothing,\rhd)\in\EC(\omega,R)$. Hence, $\STTC(\omega,R)\subseteq\EC(\omega,R)$. Thus, by Claims~\ref{claim.muinEC} and~\ref{claim.munotinSTCC}, $\STTC(\omega,R)\subsetneq\EC(\omega,R)$.
\end{proof}

\begin{proof}[\textbf{Proof of Corollary~\ref{cor:maxdomain}}]
If the endowment structure is a mixed-ownership structure, then each associated problem $R$ is a mixed-ownership problem and the coincidence follows from Proposition~\ref{prop.corecharacterization}. Conversely, if the endowment structure is not a mixed-ownership structure, then some object $a \in O$ is properly co-owned; that is, at least two but not all agents own $a$. By Theorem~\ref{thm.strictsuperset}, there then exists a preference profile $R$ with $\STTC(\omega,R) \neq \EC(\omega,R)$.
\end{proof}

\begin{proof}[\textbf{Proof of Theorem~\ref{theorem:coretest}}]
Let $R$ be a problem and let $\mu$ be an allocation. We prove $\textup{(i)}\Rightarrow\textup{(iii)}\Rightarrow\textup{(ii)}\Rightarrow\textup{(i)}$.

Assume \textup{(i)}. Then, for every priority order $\rhd$, Lemma~\ref{lemm.fixedpoint} in Appendix~\ref{appendix:sttc} implies that $\mu=\STTC(R,\mu,\rhd)$, which is \textup{(iii)}. That \textup{(iii)} implies \textup{(ii)} is immediate, since at least one priority order exists. Finally, assume \textup{(ii)} and let $\rhd$ be a priority order with $\mu=\STTC(R,\mu,\rhd)$. Then Proposition~\ref{prop.sttc-exclusion} implies that $\mu$ is exclusion core stable, which is \textup{(i)}.
\end{proof}

\begin{proof}[\textbf{Proof of Lemma~\ref{lemm.selfcontained}}]
Let $R$ be a problem, let $\mu$ and $\sigma$ be allocations, and let $S$ be a coalition that exclusion-blocks $\mu$ via $\sigma$. Since every member of $S$ is strictly better off at $\sigma$, $S\subseteq S^+(\mu,\sigma)$. Consider an agent $j\notin S^+(\mu,\sigma)$ such that $\mu(j)\mathbin{P_j}\sigma(j)$. Then, $j\notin S$. Since $S$ exclusion-blocks $\mu$ via $\sigma$, $\mu(j)\in\Omega(S\mid\mu)\subseteq\Omega(S^+(\mu,\sigma)\mid\mu)$, where the inclusion follows from Lemma~\ref{lemm.evictionproperties}(b). Hence, $S^+(\mu,\sigma)$ exclusion-blocks $\mu$ via $\sigma$.

It remains to show that this blocking is self-contained. For ease of notation, let $S^+:=S^+(\mu,\sigma)$ and $\widehat\sigma:=\widehat\sigma^{S^+,\mu}$. Each agent $i\in S^+$ receives $\widehat\sigma(i)=\sigma(i)$, which he strictly prefers to $\mu(i)$. Consider an agent $j\notin S^+$ such that $\mu(j)\mathbin{P_j}\widehat\sigma(j)$. By the definition of the post-blocking allocation, $\widehat\sigma(j)=\varnothing$ and $\mu(j)\in\sigma(S^+)\setminus\{\varnothing\}$. Hence, there exists an agent $i\in S^+$ such that $\sigma(i)=\mu(j)$. Since $j\notin S^+$, we have $i\neq j$; feasibility of $\sigma$ therefore implies that $\sigma(j)\neq\mu(j)$. Since preferences are strict and $j\notin S^+$, it follows that $\mu(j)\mathbin{P_j}\sigma(j)$. Because $S^+$ exclusion-blocks $\mu$ via $\sigma$, we obtain $\mu(j)\in\Omega(S^+\mid\mu)$. Therefore, $S^+$ also exclusion-blocks $\mu$ via $\widehat\sigma$, and the blocking is self-contained.
\end{proof}

\begin{proof}[\textbf{Proof of Theorems~\ref{theorem:vNMstability} and~\ref{theorem:pathtoexclusion}}]
Let $R$ be a problem. Internal stability in Theorem~\ref{theorem:vNMstability} follows directly from the definition of the exclusion core. For the remaining claims, let $\mu$ be an allocation. If $\mu\in\EC(R)$, the constant sequence consisting only of $\mu$ is a self-contained exclusion-blocking path of length zero. Thus, for the path result, it remains to consider $\mu\notin\EC(R)$; this is also precisely the case relevant for external stability.

Fix a priority order $\rhd$ and apply the $(\rhd,\mu)$-STTC algorithm. Set $\mu^\rhd:=\STTC(R,\mu,\rhd)$.
By Proposition~\ref{prop.sttc-exclusion}, $\mu^\rhd$ is exclusion core stable, and by Theorem~\ref{theorem:coretest}, $\mu^\rhd\neq\mu$. If every executed cycle were a $\mu$-cycle, every agent leaving in a cycle would receive his $\mu$-allotment. Moreover, any agent remaining when no real object is left must have $\mu(i)=\varnothing$; otherwise, his real $\mu$-allotment would have been removed in an earlier $\mu$-cycle, in which he himself would have participated. Hence, the algorithm would return $\mu$, contrary to $\mu^\rhd\neq\mu$. Thus, at least one non-$\mu$-cycle forms. Let $K$ be the number of steps at which at least one non-$\mu$-cycle forms, and enumerate these steps as $t^1<\cdots<t^K$. For each $\ell\in\{1,\ldots,K\}$, let $\mathscr{C}^\ell$ be the set of all non-$\mu$-cycles at Step~$t^\ell$, and define $S(\mathscr{C}^\ell):=\bigcup_{\mathcal C\in\mathscr{C}^\ell}S(\mathcal C)$.

Set $\mu^0:=\mu$. For each $\ell\in\{1,\ldots,K\}$, define $\mu^\ell$ recursively by\nopagebreak\vspace{-0.2cm}
\[ \mu^\ell(i):=
\begin{cases}
\mu^\rhd(i),
&\text{if }i\in S(\mathscr{C}^\ell),\\
\varnothing,
&\text{if }i\notin S(\mathscr{C}^\ell)
 \text{ and there exists }h\in S(\mathscr{C}^\ell)
 \text{ such that }\mu^{\ell-1}(i)=\mu^\rhd(h),\\
\mu^{\ell-1}(i),
&\text{otherwise.}
\end{cases}\vspace{-0.2cm} \]

Call an agent outside $S(\mathscr{C}^\ell)$ an \textbf{outsider} at Step~$t^\ell$, and say that an outsider is \textbf{displaced} if some cycle agent receives his $\mu^{\ell-1}$-allotment. Since distinct cycles are disjoint in agents and in real objects, and every displaced outsider receives the null object, $\mu^\ell$ is a well-defined allocation.

We use the same STTC steps to construct two kinds of coalitions. Define $\widehat S^0:=\emptyset$. For each $\ell\in\{1,\ldots,K\}$, define
\begin{itemize}
    \item the \emph{path coalition}\nopagebreak\vspace{-0.2cm}
    \[ \widehat{S}^\ell_{\mathrm{path}}
    :=
    \{i\in S(\mathscr{C}^\ell)\mid
       \mu^\rhd(i)\mathbin{P_i}\mu^{\ell-1}(i)\},\vspace{-0.2cm} \]
    which compares the final STTC allotment with the current path allotment; and
    \item the \emph{external-stability coalition}\nopagebreak\vspace{-0.2cm}
    \[\widehat{S}^\ell
    :=
    \left\{i\in\bigcup_{k=1}^{\ell}S(\mathscr{C}^k)
       \mathrel{\Big|}
       \mu^\rhd(i)\mathbin{P_i}\mu(i)\right\},\vspace{-0.2cm} \]
    which compares the final STTC allotment with the original allotment.
\end{itemize}
Thus, the path coalition is local and myopic, whereas the external-stability coalition is cumulative and farsighted. By construction, for every $\ell\in\{1,\ldots,K\}$, $\widehat{S}^{\ell-1}\subseteq\widehat{S}^{\ell}$.

\paragraph{History and location property.}
At each of the allocations $\mu^0,\ldots,\mu^K$, every agent holds one of three allotments: his original allotment, his final STTC allotment, or the null object. More precisely, we record the following \emph{history and location property}: for every $\ell\in\{1,\ldots,K\}$ and every agent $i$,
\begin{enumerate}
\item[\textup{(i)}] if $i$ is present at Step~$t^\ell$, then $\mu^{\ell-1}(i)\in\{\mu(i),\varnothing\}$, and $\mu^{\ell-1}(i)$ is exactly his allotment in the reduced allocation at Step~$t^\ell$;
\item[\textup{(ii)}] if $i$ departed before Step~$t^\ell$, then $\mu^{\ell-1}(i)$ is his final STTC allotment, which is an already-removed object or the null object; and
\item[\textup{(iii)}] if $i\in S(\mathscr{C}^\ell)$, then $\mu^\rhd(i)\mathbin{R_i}\mu^{\ell-1}(i)$.
\end{enumerate}
To see this, note that once an agent participates in an STTC cycle he receives $\mu^\rhd(i)$ and keeps it at every later $\mu^\ell$, whereas an agent who has not yet participated either keeps $\mu(i)$ or has been displaced and receives the null object. Part~\textup{(iii)} follows because a cycle agent points to his most-preferred object in $O^{t^\ell}\cup\{\varnothing\}$, and by part~\textup{(i)} his current allotment lies in that set.

Two consequences follow. By part~\textup{(ii)}, any agent whose $\mu^{\ell-1}$-allotment is a real object still present at Step~$t^\ell$ is himself present at Step~$t^\ell$. By part~\textup{(iii)} and antisymmetry, every cycle agent outside $\widehat{S}^\ell_{\mathrm{path}}$ keeps his current allotment and so displaces no outsider; hence the outsiders displaced in the definition of $\mu^\ell$ are exactly those whose $\mu^{\ell-1}$-allotment is claimed by a member of $\widehat{S}^\ell_{\mathrm{path}}$. Hence, whenever $\widehat{S}^\ell_{\mathrm{path}}\neq\emptyset$, the post-blocking formula in Definition~\ref{defi.postblocking}, applied to $\widehat{S}^\ell_{\mathrm{path}}$ at $\mu^{\ell-1}$ with blocking allocation $\mu^\ell$, reproduces $\mu^\ell$. If $\widehat{S}^\ell_{\mathrm{path}}=\emptyset$, then $\mu^\ell=\mu^{\ell-1}$.

We also use Lemma~\ref{lemm.evictionproperties}(c): once an agent's original allotment has been reached, the exclusion rights associated with his endowment become available to the coalition.

\paragraph{A cyclewise argument based on chains of dependence.}
Fix $\ell\in\{1,\ldots,K\}$ and an ordinary cycle $\mathcal{C}\in\mathscr{C}^\ell$, and define its local improvers by\nopagebreak\vspace{-0.2cm}
\[ \widehat S^\ell_{\mathrm{path}}(\mathcal{C}) := \widehat{S}^\ell_{\mathrm{path}}\cap S(\mathcal{C})
= T(\mathcal{C},\mu^{\ell-1}),\vspace{-0.2cm} \]
the last equality holding by the definitions of $\widehat{S}^\ell_{\mathrm{path}}$ and of $T(\mathcal{C},\nu)$ in Lemma~\ref{lemm.cyclechain}. Suppose that there exists an agent $j\notin S(\mathcal{C})$ such that $\mu^{\ell-1}(j)$ is a real object of $\mathcal{C}$. Thus, $j$ is displaced by this cycle. By parts~\textup{(i)} and~\textup{(ii)} of the history and location property, $j\in N^{t^\ell}\setminus S(\mathcal{C})$ and $\mu^{\ell-1}$ agrees with the reduced allocation on $N^{t^\ell}$, so Lemma~\ref{lemm.cyclechain} in Appendix~\ref{appendix:sttc} applies with $\nu:=\mu^{\ell-1}$. By that lemma, $\widehat S^\ell_{\mathrm{path}}(\mathcal{C})\neq\emptyset$. Let $m$ be the number of agents of $\mathcal{C}$ and enumerate $\mathcal{C}$ from $o^1:=\mu^{\ell-1}(j)$ as $o^1\rightarrow i^1\rightarrow\cdots\rightarrow o^m\rightarrow i^m\rightarrow o^1$, and let $q$ be the smallest index with $i^q\in\widehat S^\ell_{\mathrm{path}}(\mathcal{C})$. Then, for each $r\leq q$, $o^r\in\omega_{i^r}$, and, for each $r<q$, $\mu^{\ell-1}(i^r)=o^{r+1}$.\par\smallskip
These are the conditions of a chain of dependence at $\mu^{\ell-1}$, and $\widehat S^\ell_{\mathrm{path}}(\mathcal{C})$ reaches $i^q$ because $i^q$ belongs to it. Therefore, by Lemma~\ref{lemm.evictionproperties}(d), applied with $S:=\widehat S^\ell_{\mathrm{path}}(\mathcal{C})$ and the allocation $\mu^{\ell-1}$, $\mu^{\ell-1}(j)=o^1\in\Omega(\widehat S^\ell_{\mathrm{path}}(\mathcal{C})\mid\mu^{\ell-1})$. Since $\widehat S^\ell_{\mathrm{path}}(\mathcal{C})\subseteq\widehat S^\ell_{\mathrm{path}}$, by Lemma~\ref{lemm.evictionproperties}(b), $\mu^{\ell-1}(j)\in{\Omega(\widehat{S}^\ell_{\mathrm{path}}\mid\mu^{\ell-1})}$.

Cycles involving the null object displace no outsider, so this argument is not needed for them. Since the ordinary cycles are disjoint, the preceding argument applies independently to every outsider displaced by the simultaneous cycles in $\mathscr{C}^\ell$. Therefore, whenever $\widehat{S}^\ell_{\mathrm{path}}\neq\emptyset$, coalition $\widehat{S}^\ell_{\mathrm{path}}$ exclusion-blocks $\mu^{\ell-1}$ via $\mu^\ell$. Together with the post-blocking formula above, this blocking is self-contained and establishes the myopic transition $\mu^{\ell-1}\mapsto\mu^\ell$.
If $\widehat{S}^\ell_{\mathrm{path}}=\emptyset$, then $\mu^\ell=\mu^{\ell-1}$. When forming the final path, we delete this repeated allocation. This proves the path conclusion at every step; it remains to verify that the last path allocation is $\mu^\rhd$, which we do below.

We now prove by induction that, for every $\ell\in\{1,\ldots,K\}$, $\widehat{S}^\ell$ exclusion-blocks $\mu$ via $\mu^\ell$.
All cycles executed before Step~$t^1$ are $\mu$-cycles, so every agent who left earlier received his $\mu$-allotment. Consequently, the $\mu$-allotment of every agent present at Step~$t^1$ is still available. Indeed, a real object removed earlier was taken by its $\mu$-holder, who has departed, so the $\mu$-allotment of a present agent cannot have been removed. Hence, every cycle agent at Step~$t^1$ weakly prefers his final STTC allotment to his original allotment. Each cycle in $\mathscr{C}^1$ is a non-$\mu$-cycle, so it contains an agent $i$ with $\mu^\rhd(i)\neq\mu(i)$; together with the preceding weak preference and antisymmetry, $\mu^\rhd(i)\mathbin{P_i}\mu(i)$. Since $\mu^0=\mu$, we have $\widehat{S}^1=\widehat{S}^1_{\mathrm{path}}$, and this set is nonempty.
The cyclewise argument, now with $\mu^0=\mu$, shows that every outsider harmed at $\mu^1$ has his original allotment in $\Omega(\widehat{S}^1\mid\mu)$. Thus, $\widehat{S}^1$ exclusion-blocks $\mu$ via $\mu^1$.

Fix $\ell\in\{2,\ldots,K\}$, suppose that the claim holds at $\ell-1$, and consider Step~$t^\ell$. Every member of $\widehat{S}^\ell$ has already participated in one of the cycles in $\mathscr{C}^1,\ldots,\mathscr{C}^\ell$ and therefore receives his final STTC allotment at $\mu^\ell$. By definition, he strictly prefers this allotment to his original allotment. It remains to verify that $\widehat{S}^\ell$ exclusion-enforces $\mu^\ell$ at $\mu$. Let $j\notin\widehat{S}^\ell$ satisfy $\mu(j)\mathbin{P_j}\mu^\ell(j)$. If $\mu(j)\mathbin{P_j}\mu^{\ell-1}(j)$, then the induction hypothesis and Lemma~\ref{lemm.evictionproperties}(b) give $\mu(j)\in\Omega(\widehat{S}^{\ell-1}\mid\mu)\subseteq{\Omega(\widehat{S}^{\ell}\mid\mu)}$.

Suppose instead that $\mu(j)\mathbin{P_j}\mu^{\ell-1}(j)$ does not hold. Since preferences are strict, $\mu^{\ell-1}(j)\mathbin{R_j}\mu(j)$. Together with $\mu(j)\mathbin{P_j}\mu^\ell(j)$, we obtain $\mu^{\ell-1}(j)\mathbin{R_j}\mu(j)\mathbin{P_j}\mu^\ell(j)$.
Thus, $j$ becomes worse off relative to $\mu$ at the current step.
Agent $j$ cannot belong to $S(\mathscr{C}^\ell)$. Indeed, suppose that $j\in S(\mathscr{C}^\ell)$. Then $j$ receives his final STTC allotment at $\mu^\ell$, and, by part~\textup{(iii)} of the history and location property, $\mu^\ell(j)=\mu^\rhd(j)\mathbin{R_j}\mu^{\ell-1}(j)$. But $\mu^{\ell-1}(j)\mathbin{R_j}\mu(j)\mathbin{P_j}\mu^\ell(j)$ implies $\mu^{\ell-1}(j)\mathbin{P_j}\mu^\ell(j)$, which contradicts $\mu^\ell(j)\mathbin{R_j}\mu^{\ell-1}(j)$. Since $\mu^\ell$ differs from $\mu^{\ell-1}$ only for current cycle agents and displaced outsiders, $j$ must be a newly displaced outsider. Hence, $\mu^\ell(j)=\varnothing$, and his real allotment $\mu^{\ell-1}(j)$ is taken by one of the current cycles. Since $j$ is still present at Step~$t^\ell$, part~\textup{(i)} of the history and location property gives $\mu^{\ell-1}(j)\in\{\mu(j),\varnothing\}$.
Moreover, $\mu^{\ell-1}(j)\mathbin{R_j}\mu(j)\mathbin{P_j}\mu^\ell(j)=\varnothing$, which rules out $\mu^{\ell-1}(j)=\varnothing$. Hence, $\mu^{\ell-1}(j)=\mu(j)$.
Let $\mathcal{C}$ be the cycle that displaces $j$. We show that $\mu(j)\in\Omega(\widehat{S}^\ell\mid\mu)$, first establishing that every local improver in $\mathcal{C}$ is either a member of $\widehat{S}^\ell$ or has his endowment contained in $\Omega(\widehat{S}^\ell\mid\mu)$, and then applying Lemmas~\ref{lemm.cyclechain} and~\ref{lemm.evictionproperties}(d) to $\mathcal{C}$ with the chain of dependence read at the original allocation~$\mu$.

For every local improver $i\in \widehat S^\ell_{\mathrm{path}}(\mathcal{C})$, one of the following two cases applies. If $i\in\widehat{S}^\ell$, then $i$ is already a member of the farsighted coalition. If $i\notin\widehat{S}^\ell$, then $\mu(i)\mathbin{R_i}\mu^\rhd(i)\mathbin{P_i}\mu^{\ell-1}(i)$. Since $i$ is still present at Step~$t^\ell$, part~\textup{(i)} of the history and location property gives $\mu^{\ell-1}(i)\in\{\mu(i),\varnothing\}$. The preference chain rules out $\mu^{\ell-1}(i)=\mu(i)$. Hence, $\mu^{\ell-1}(i)=\varnothing$, and therefore $\mu(i)\mathbin{P_i}\mu^{\ell-1}(i)$.

Moreover, since $i$ is still present at Step~$t^\ell$, he has not participated in any earlier cycle. Hence, $i\notin\widehat{S}^{\ell-1}$. By the induction hypothesis, $\widehat{S}^{\ell-1}$ exclusion-blocks $\mu$ via $\mu^{\ell-1}$. Since $i\notin\widehat{S}^{\ell-1}$ and $\mu(i)\mathbin{P_i}\mu^{\ell-1}(i)$, the exclusion-enforcement condition implies $\mu(i)\in\Omega(\widehat{S}^{\ell-1}\mid\mu)$. Because $\widehat{S}^{\ell-1}\subseteq\widehat{S}^{\ell}$, Lemma~\ref{lemm.evictionproperties}(b) gives $\mu(i)\in\Omega(\widehat{S}^{\ell-1}\mid\mu)\subseteq\Omega(\widehat{S}^{\ell}\mid\mu)$. Lemma~\ref{lemm.evictionproperties}(c) then implies that the exclusion rights associated with the objects in $\omega_i$ are available to coalition $\widehat{S}^{\ell}$.

Return to the enumeration of $\mathcal{C}$ from $o^1:=\mu^{\ell-1}(j)$, $o^1\to i^1\to o^2\to\cdots\to o^m\to i^m\to o^1$, supplied by Lemma~\ref{lemm.cyclechain} with $\nu:=\mu^{\ell-1}$, and let $q\in\{1,\ldots,m\}$ be the smallest index such that $i^q$ is a local improver. Since $\mu^{\ell-1}(j)=\mu(j)$, as shown above, $o^1=\mu(j)$. For every $r<q$, agent $i^r$ is not a local improver. Hence, by Lemma~\ref{lemm.cyclechain}(b), $\mu^{\ell-1}(i^r)=o^{r+1}$. Since $o^{r+1}$ is a real object, part~\textup{(i)} of the history and location property then implies $\mu(i^r)=\mu^{\ell-1}(i^r)$. Thus, for every $r<q$, $\mu(i^r)=\mu^{\ell-1}(i^r)=o^{r+1}$.
By Lemma~\ref{lemm.cyclechain}(a), for each $r\leq q$, $o^r\in\omega_{i^r}$; and $i^q$ is either a member of $\widehat{S}^\ell$ or satisfies $\omega_{i^q}\subseteq\Omega(\widehat{S}^\ell\mid\mu)$, as shown above. Therefore, by Lemma~\ref{lemm.evictionproperties}(d), applied with $S:=\widehat{S}^\ell$ and the original allocation $\mu$, $\mu(j)=o^1\in\Omega(\widehat{S}^\ell\mid\mu)$.

This verifies exclusion-enforcement for every outsider harmed at $\mu^\ell$. Therefore, $\widehat{S}^\ell$ exclusion-blocks $\mu$ via $\mu^\ell$, completing the induction.

It remains to show that the final path allocation is the STTC allocation, $\mu^K=\mu^\rhd$. Every agent who leaves in a $\mu$-cycle retains his original allotment, which equals his final STTC allotment. Every agent who leaves in a non-$\mu$-cycle receives his final STTC allotment at the corresponding path step or already holds it when that step leaves the path allocation unchanged. It remains only to consider an agent $i$ who never leaves in a cycle and is assigned $\mu^\rhd(i)=\varnothing$ when the algorithm stops because no real object remains. If $\mu(i)=\varnothing$, then, for every $\ell$, $\mu^\ell(i)=\varnothing$. If $\mu(i)$ is a real object, then it was removed at some step and assigned to an agent $h$ with $\mu^\rhd(h)=\mu(i)$. The cycle of $h$ cannot be a $\mu$-cycle: otherwise $\mu(h)=\mu^\rhd(h)=\mu(i)$, and since $\mu(i)$ is a real object and an allocation assigns a real object to at most one agent, this would imply $h=i$, contradicting that $i$ never leaves in a cycle. Hence, for some $\ell'$, $h\in S(\mathscr{C}^{\ell'})$, and $i$ is displaced at path step $\ell'$, so $\mu^{\ell'}(i)=\varnothing$. In both cases $\mu^K(i)=\varnothing=\mu^\rhd(i)$. Consequently, $\mu^K=\mu^\rhd$.
After deleting repeated allocations and relabeling the remaining sequence, the myopic transitions constructed above form a self-contained exclusion-blocking path from $\mu$ to the exclusion-core-stable allocation $\mu^\rhd$. This proves Theorem~\ref{theorem:pathtoexclusion}.

Finally, applying the farsighted induction at $\ell=K$, $\widehat{S}^K$ exclusion-blocks the original allocation $\mu$ via $\mu^K$. Since we have just shown that $\mu^K=\mu^\rhd$, and $\mu^\rhd$ belongs to the exclusion core, external stability follows. This proves Theorem~\ref{theorem:vNMstability}.
\end{proof}

\begin{proof}[\textbf{Proof of Corollary~\ref{coro:randompath}}]
Let $R$ be a problem and $\mu$ an allocation. Since the sets of agents and objects are finite, the set of allocations is finite. If every allocation is exclusion core stable, the conclusion is immediate. Suppose, therefore, that at least one allocation lies outside the exclusion core. By Theorem~\ref{theorem:pathtoexclusion}, for each such allocation $\nu$, choose one finite self-contained exclusion-blocking path from $\nu$ to an exclusion-core-stable allocation. Let $L_\nu$ denote its length and let $p_\nu>0$ denote the probability, conditional on starting from $\nu$, that the random process follows this path. Since the path is finite and every available transition is selected with positive probability, $p_\nu$ is positive. Since there are only finitely many allocations, $L:=\max_{\nu\notin\EC(R)}L_\nu$ and $p:=\min_{\nu\notin\EC(R)}p_\nu$ are well defined, and $p>0$. Because every exclusion-core-stable allocation is absorbing, from every allocation outside the exclusion core the probability of reaching the exclusion core within the next $L$ steps is at least $p$.

Conditional on the process not having reached the exclusion core after $rL$ steps, its current allocation lies outside the exclusion core. The probability that it also fails to reach the exclusion core during the next $L$ steps is therefore at most $1-p$. Induction on $r$ gives\nopagebreak\vspace{-0.2cm}
\[ \Pr(\text{the exclusion core has not been reached after }rL\text{ steps})
\leq (1-p)^r.\vspace{-0.2cm} \]
Since $(1-p)^r$ converges to zero, with probability one the process reaches an exclusion-core-stable allocation after finitely many steps and remains there.
\end{proof}

\section{An illustration of the STTC algorithm and stability results}\label{appendix:STTCexample}
The following example illustrates the STTC algorithm, the exclusion core characterization in Proposition~\ref{prop.corecharacterization}, the exclusion-core-stability test in Theorem~\ref{theorem:coretest}, the external-stability conclusion of Theorem~\ref{theorem:vNMstability}, and the path construction in Theorem~\ref{theorem:pathtoexclusion}.

\begin{exam}[\textbf{An illustration of the STTC algorithm and stability results}]\label{example:STTC}
Consider a set of agents $N = \{1,2,3,4,5,6,7\}$, a set of objects $O = \{a,b,c,d,e,f\}$, and endowments $\omega_1=\{a,b\}$, $\omega_2=\{c\}$, $\omega_3=\{d\}$, and $\omega_4=\omega_5=\omega_6=\omega_7=\emptyset$. Thus, objects $e$ and $f$ are public. Fix the priority order $1\rhd 2\rhd\cdots\rhd 7$ and the initial allocation $\bm{\mu}:=(\bm{a},\bm{b},\bm{c},\bm{d},\bm{f},\bm{e},\bm{\varnothing})$.

Let $R$ be the preference profile displayed below, and define $\mu^\rhd:=\STTC(R,\mu,\rhd)$. In the preference profile, boldface indicates each agent's allotment at $\mu$, and underlining indicates his allotment at $\mu^\rhd$. An allotment shown in both formats is unchanged by the algorithm.\vspace{-0.2cm}

    \[ \begin{array}{ccccccc}
    R_{1}:& \underline{c} & b & \bm{a} & \varnothing & \ldots\\
    R_{2}:& \underline{\bm{b}} & \varnothing & \ldots\\
    R_{3}:& \bm{c} & \underline{d} & \varnothing & \ldots\\
    R_{4}:& \bm{d} & \underline{a} & \varnothing & \ldots\\
    R_{5}:& \underline{e} & d & \bm{f} & \varnothing & \ldots\\
    R_{6}:& \underline{f} & \bm{e} &\varnothing & \ldots\\
    R_{7}:& a & \underline{\bm{\varnothing}} & \ldots.
    \end{array}\vspace{-0.2cm}
    \]

Two kinds of superscripts are used below. A superscript $t$ on $N$, $O$, $\omega$, $R$, $\rhd$, $\mathcal{G}$, $\mathscr{C}$, or $\tilde{\mu}$ refers to Step~$t$ of the algorithm. Thus, the primed notation $N'$, $O'$, $\omega'$, $R'$, $\rhd'$, and $\mu'$ from Definition~\ref{defi.residualgraph} becomes $N^t$, $O^t$, $\omega^t$, $R^t$, $\rhd^t$, and $\tilde{\mu}^t$ here.

By contrast, a superscript $\ell$ on $\mu^\ell$ indexes the allocations on the blocking path, with $\mu^0=\mu$ and $\mu^K=\mu^\rhd$. In this example, every step produces a non-$\mu$-cycle, so the two indices happen to coincide; in general, they do not.

We now apply the STTC algorithm. At each step, we first describe the reduced STTC graph and the cycles that form. We then identify the corresponding myopic path coalition and cumulative external-stability coalition. At Step~1, we also illustrate the exclusion-core-stability test.

\paragraph{$\bm{(\rhd,\mu)}$-STTC algorithm:}

\paragraph{Step~1.} Let $N^1=N$ and $O^1=O$. Construct the STTC graph $\mathcal{G}^1=\mathcal{G}(N^1,O^1,\omega,R,\mu,\rhd)$ according to Definition~\ref{defi.graph}. All agents point to their most-preferred object, which yields the edges $(1,c)$, $(2,b)$, $(3,c)$, $(4,d)$, $(5,e)$, $(6,f)$, $(7,a)$. Since agent~$1$ owns objects $a$ and $b$, agent~$2$ owns object $c$, and agent $3$ owns object $d$, edges $(a,1)$, $(b,1)$, $(c,2)$, and $(d,3)$ form. In addition, since public object $f$ is assigned to agent~$5$ under $\mu$, edge $(f,5)$ forms, and since public object $e$ is assigned to agent~$6$ under $\mu$, edge $(e,6)$ forms. Since no agent points to the null object $\varnothing$, it has no outgoing edge. The resulting directed STTC graph $\mathcal{G}^1$ is depicted in Figure~\ref{fig.G1}.\vspace{-0.2cm}

    \begin{figure}[!htbp]
    \centering
    \begin{tikzpicture}[
        every node/.style={inner sep=2pt}
    ]


    \foreach \i in {1,...,7} \node (\i) at ({2*(\i-1)},0) {$\i$};


    \node (a) at (-0.6,-1.6) {$a$};
    \node (b) at (0.6,-1.6) {$b$};
    \node (c) at (3.0,-1.6) {$c$};
    \node (d) at (5.0,-1.6) {$d$};
    \node (e) at (8.0,-1.6) {$e$};
    \node (f) at (10.0,-1.6) {$f$};


    \node (null) at ($(1)!0.5!(7)+(0,0.95cm)$) {$\varnothing$};


    \node[left=0.5cm of 1] {$\mathcal{G}^1=$};


    \draw[->, thick] (1) -- (c);
    \draw[->, thick] (2) -- (b);
    \draw[->, thick] (3) -- (c);
    \draw[->, thick] (4) -- (d);
    \draw[->, thick] (5) -- (e);
    \draw[->, thick] (6) -- (f);
    \draw[->, thick] (7) .. controls (12,-3.0) and (-0.6,-3.0) .. (a);


    \draw[->, thick] (a) -- (1);
    \draw[->, thick] (b) -- (1);
    \draw[->, thick] (c) -- (2);
    \draw[->, thick] (d) -- (3);
    \draw[->, thick] (f) -- (5);
    \draw[->, thick] (e) -- (6);

    \end{tikzpicture}
    \caption{The STTC graph $\mathcal{G}^1$ constructed at Step~1.}
    \label{fig.G1}
    \end{figure}
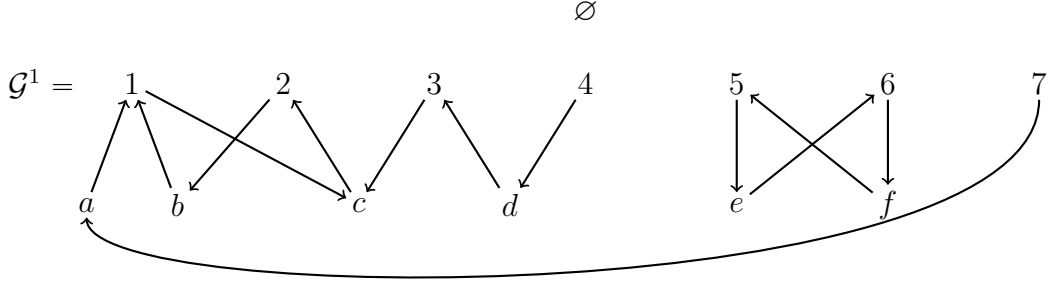\vspace{-0.2cm}

In STTC graph $\mathcal{G}^1$, two cycles form: $1 \rightarrow c \rightarrow 2 \rightarrow b \rightarrow 1$ and $5 \rightarrow e \rightarrow 6 \rightarrow f \rightarrow 5$. These are earliest non-$\mu$-cycles. We denote the set of all non-$\mu$-cycles at this step by $\mathscr{C}^1$ and define\nopagebreak\vspace{-0.2cm}
\[ S(\mathscr{C}^1):= \bigcup_{\mathcal C\in\mathscr C^1}S(\mathcal C)=\{1,2,5,6\}.\vspace{-0.2cm} \]
Execute these cycles by giving agent~$1$ object $c$, agent $2$ object $b$, agent $5$ object $e$, and agent $6$ object $f$.

\subparagraph{Exclusion-core-stability test.} The first non-$\mu$-cycle $1 \rightarrow c \rightarrow 2 \rightarrow b \rightarrow 1$ identifies the following exclusion-blocking. Let $\widehat{S}=\{1\}$, and let the blocking allocation be $\widehat{\sigma}=(c,b,\varnothing,d,f,e,\varnothing)$.

Agent~$1 \in \widehat{S}$ receives $c$, while agent~$3 \notin \widehat{S}$ (with $\mu(3) = c = \widehat{\sigma}(1)$) is left with $\varnothing$. All other agents keep their $\mu$-allotments. To see that coalition $\widehat{S}$ exclusion-blocks $\mu$ via $\widehat{\sigma}$, notice that agent~$1$ is strictly better off receiving object $c$, which is assigned to agent $3$ under $\mu$. Agent~$3$, in turn, is worse off because his most-preferred object, $c$, is now assigned to agent~$1$. However, $b\in\omega_1$ and agent~$2$ holds $b$ at $\mu$; hence agent~$2$ becomes reachable from $\widehat{S}=\{1\}$. Since $c\in\omega_2$, it follows that $c\in\Omega(\widehat{S}\mid\mu)$. Thus, object $c$ belongs to agent $1$'s extended eviction set. By Theorem~\ref{theorem:coretest}, $\mu$ is not exclusion core stable.

We now construct the coalitions and allocations used for the path and external-stability results. Set $\mu^0:=\mu$.

\subparagraph{Path transition.} Define the path coalition\nopagebreak\vspace{-0.2cm}
\[\widehat{S}^1_{\mathrm{path}} = \{i \in S(\mathscr{C}^1) \mid \mu^{\rhd}(i) \mathbin{P_i} \mu^0(i)\} = \{1, 5, 6\}.\vspace{-0.2cm}\]
Note that agent $2 \notin \widehat{S}^1_{\mathrm{path}}$ because $\mu^{\rhd}(2) = b = \mu^0(2)$. Define allocation \(\mu^1=(c,b,\varnothing,d,e,f,\varnothing)\).
Here, agents in $S(\mathscr{C}^1)$ receive their STTC allotments: $\mu^1(1) = \mu^\rhd(1) = c$, $\mu^1(2) = \mu^\rhd(2) = b$, $\mu^1(5) = \mu^\rhd(5) = e$, and $\mu^1(6) = \mu^\rhd(6) = f$. Agent~$3$ receives $\varnothing$ because $\mu^0(3) = c = \mu^\rhd(1)$: his current allotment is taken by agent~$1$ in the cycle. All other agents keep their $\mu^0$-allotments. Every member of $\widehat{S}^1_{\mathrm{path}}$ is strictly better off at $\mu^1$. Agent~$3$ is the only outsider who is worse off, and $\mu^0(3)= c\in\Omega(\widehat{S}^1_{\mathrm{path}} \mid\mu^0)$, because $c$ is owned by agent~$2$, who receives $b\in\omega_1$, and agent~$1$ belongs to $\widehat{S}^1_{\mathrm{path}}$. Hence, $\widehat{S}^1_{\mathrm{path}}$ exclusion-blocks $\mu^0$ via $\mu^1$. Applying Definition~\ref{defi.postblocking} to $\widehat{S}^1_{\mathrm{path}}$ at $\mu^0$ with blocking allocation $\mu^1$ reproduces $\mu^1$. Thus, the blocking is self-contained and establishes the transition $\mu^0 \mapsto \mu^1$.

\subparagraph{External-stability coalition.} Define the external-stability coalition\nopagebreak\vspace{-0.2cm}
\[\widehat{S}^1 = \{i \in S(\mathscr{C}^1) \mid \mu^{\rhd}(i) \mathbin{P_i} \mu(i)\} = \{1, 5, 6\}.\vspace{-0.2cm}\]
At this step, $\widehat{S}^1 = \widehat{S}^1_{\mathrm{path}}$ because both coalitions are defined relative to the original allocation $\mu = \mu^0$. Therefore, the preceding verification also shows that $\widehat{S}^1$ exclusion-blocks $\mu$ via $\mu^1$.

Let $N^2 = N^1\setminus \{1,2,5,6\} = \{3,4,7\}$, $O^2 = O^1\setminus \{b,c,e,f\} = \{a,d\}$.

\paragraph{Step~2.} We derive the reduced problem for $N^2 = \{3,4,7\}$ and $O^2 = \{a, d\}$. The reduced endowment structure is given by $\omega^2_3 = \{d\}$ and $\omega^2_4 = \omega^2_7 = \emptyset$. Since $\omega_1 = \{a,b\}$, agent $1 \notin N^2$ and $a\in O^2$, object $a$ turns into an unassigned public object in the reduced problem (see Definition~\ref{defi.residualgraph}). The reduced preference profile is $R^2_{3}: d, \varnothing$, $R^2_{4}: d, a, \varnothing$, and $R^2_{7}: a, \varnothing$.

The reduced allocation, listed for agents $3$, $4$, and $7$ in that order, is $\tilde{\mu}^2 =(\varnothing,d,\varnothing)$. The reduced priority order is $3\rhd^2 4\rhd^2 7$.

Agents~$3$, $4$, and $7$ point to $d$, $d$, and $a$, respectively. Object~$d$ points to its owner, agent~$3$. Object~$a$ has no remaining owner and is unassigned in the reduced allocation, so it points to the highest-priority remaining agent, also agent~$3$. Since no agent points to the null object, it is isolated. The resulting graph is shown in Figure~\ref{fig.G2}.\vspace{-0.2cm}

    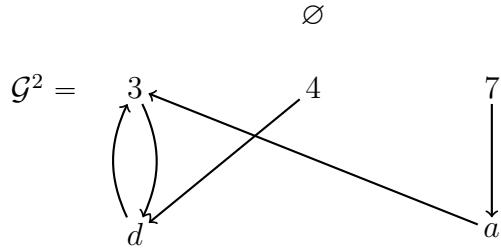
\begin{figure}[!htbp]
    \centering
    \begin{tikzpicture}[
        node distance=2cm,
        every node/.style={inner sep=2pt}
    ]

    \node (3) {$3$};
    \node[right=2cm of 3] (4) {$4$};
    \node[right=2cm of 4] (7) {$7$};

    \node[below=1.5cm of 3] (d) {$d$};
    \node[below=1.5cm of 7] (a) {$a$};

    \node (null) at ($(3)!0.5!(7)+(0,0.95cm)$) {$\varnothing$};

    \node[left=0.5cm of 3] {$\mathcal{G}^2=$};

    \draw[->, thick] (3) to[bend left=25] (d);
    \draw[->, thick] (4) -- (d);
    \draw[->, thick] (7) -- (a);

    \draw[->, thick] (d) to[bend left=25] (3);
    \draw[->, thick] (a) -- (3);

    \end{tikzpicture}
    \caption{The STTC graph $\mathcal{G}^2$ constructed at Step~2.}
    \label{fig.G2}
\end{figure}\vspace{-0.2cm}

In the STTC graph $\mathcal{G}^2$, only one cycle forms, $3 \rightarrow d \rightarrow 3$. Let $\mathscr{C}^2$ be the set containing this non-$\mu$-cycle. Then $S(\mathscr{C}^2)=\{3\}$. We execute the cycle and assign object~$d$ to agent~$3$.

\subparagraph{Path transition.} Define the path coalition\nopagebreak\vspace{-0.2cm}
\[\widehat{S}^2_{\mathrm{path}} = \{i \in S(\mathscr{C}^2) \mid \mu^{\rhd}(i) \mathbin{P_i} \mu^1(i)\} = \{3\}.\vspace{-0.2cm}\]
Agent $3$ strictly improves from $\mu^1(3) = \varnothing$ to $\mu^\rhd(3) = d$. Define allocation $\mu^2=(c,b,d,\varnothing,e,f,\varnothing)$.
Here, agent~$3\in S(\mathscr{C}^2)$ receives $\mu^2(3)=\mu^\rhd(3)=d$. Agent~$4$ receives $\varnothing$ because $\mu^1(4)=d=\mu^\rhd(3)$: his current allotment is taken by agent~$3$. All other agents keep their $\mu^1$-allotments. Agent~$3$ is strictly better off, and agent~$4$ is the only outsider who is worse off. Since $d\in\omega_3\subseteq\Omega(\{3\}\mid\mu^1)$, coalition $\widehat{S}^2_{\mathrm{path}}=\{3\}$ exclusion-blocks $\mu^1$ via $\mu^2$. Applying Definition~\ref{defi.postblocking} to $\widehat{S}^2_{\mathrm{path}}$ at $\mu^1$ with blocking allocation $\mu^2$ reproduces $\mu^2$. Thus, the blocking is self-contained and establishes the transition $\mu^1 \mapsto \mu^2$. Hence, we obtain the self-contained exclusion-blocking path $\mu \mapsto \mu^1 \mapsto \mu^2$, of length two, from $\mu$ to $\mu^2$.

\subparagraph{External-stability coalition.}
Define the external-stability coalition\nopagebreak\vspace{-0.2cm}
\[\widehat{S}^2 = \{i \in \bigcup_{k=1}^{2}S(\mathscr{C}^k) \mid \mu^{\rhd}(i) \mathbin{P_i} \mu(i)\} = \{1, 5, 6\}.\vspace{-0.2cm}\]

Note that $\widehat{S}^2 \neq \widehat{S}^2_{\mathrm{path}}$, illustrating the key difference between path coalitions and external-stability coalitions. Specifically, agent $3 \notin \widehat{S}^2$ because $\mu(3)=c$, $\mu^\rhd(3)=d$, and he strictly prefers $c$ to $d$ (agent $3$ is worse off at the final allocation than at $\mu$). However, $\widehat{S}^2$ still exclusion-blocks $\mu$ via $\mu^2$: all agents in $\widehat{S}^2$ are strictly better off, and agents $3$ and $4$ are worse off but their objects under $\mu$ are in the extended eviction set. Specifically, $\mu(3) = c\in\Omega(\widehat{S}^2\mid\mu)$, as shown at Step~1. Therefore, by Lemma~\ref{lemm.evictionproperties}(c), $\omega_3 = \{d\} \subseteq\Omega(\widehat{S}^2\mid\mu)$, and hence $\mu(4) = d\in\Omega(\widehat{S}^2\mid\mu)$.

Let $N^3 = N^2\setminus\{3\}= \{4,7\}$, and $O^3 = O^2\setminus\{d\}=\{a\}$.

\paragraph{Step~3.} We derive the reduced problem for $N^3 = \{4,7\}$ and $O^3 = \{a\}$. The reduced endowment structure is given by $\omega^3_4 = \omega^3_7 = \emptyset$. The reduced preference profile is $R^3_{4}: a, \varnothing$ and $R^3_{7}: a, \varnothing$.

The reduced allocation, listed for agents $4$ and $7$ in that order, is \(\tilde{\mu}^3=(\varnothing,\varnothing)\). The reduced priority order is $4\rhd^3 7$.

Agents~$4$ and $7$ both point to $a$. Since $a$ is an unassigned public object, it points to the highest-priority remaining agent, agent~$4$. The null object is isolated. The resulting graph is shown in Figure~\ref{fig.G3}.\vspace{-0.2cm}

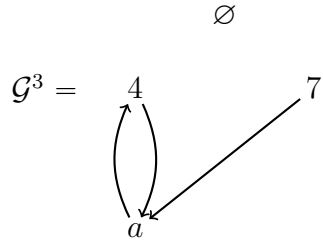
\begin{figure}[!htbp]
        \centering
        \begin{tikzpicture}[
            node distance=2cm,
            every node/.style={inner sep=2pt}
        ]


        \node (4) {$4$};
        \node[right=2cm of 4] (7) {$7$};


        \node[below=1.5cm of 4] (a) {$a$};


        \node (null) at ($(4)!0.5!(7)+(0,0.95cm)$) {$\varnothing$};


        \node[left=0.5cm of 4] {$\mathcal{G}^3=$};


        \draw[->, thick] (4) to[bend left=25] (a);
        \draw[->, thick] (7) -- (a);


        \draw[->, thick] (a) to[bend left=25] (4);

        \end{tikzpicture}
        \caption{The STTC graph $\mathcal{G}^3$ constructed at Step~3.}
        \label{fig.G3}
    \end{figure}\vspace{-0.2cm}

In the STTC graph $\mathcal{G}^3$, only one cycle forms, $4 \rightarrow a \rightarrow 4$. Let $\mathscr{C}^3$ be the set containing this non-$\mu$-cycle. Then $S(\mathscr{C}^3)=\{4\}$. We execute the cycle and assign object~$a$ to agent~$4$.

\subparagraph{Path transition.}
Define the path coalition\nopagebreak\vspace{-0.2cm}
\[\widehat{S}^3_{\mathrm{path}} = \{i \in S(\mathscr{C}^3) \mid \mu^{\rhd}(i) \mathbin{P_i} \mu^2(i)\} = \{4\}.\vspace{-0.2cm}\]
Agent $4$ strictly improves from $\mu^2(4) = \varnothing$ to $\mu^\rhd(4) = a$. Define allocation $\mu^3=(c,b,d,a,e,f,\varnothing)=\mu^\rhd$.
Here, agent~$4\in S(\mathscr{C}^3)$ receives $\mu^3(4)=\mu^\rhd(4)=a$. All other agents keep their $\mu^2$-allotments. Agent~$4$ is strictly better off. Since object~$a$ is unassigned at $\mu^2$, no outsider is worse off at $\mu^3$, so the exclusion-enforcement condition is vacuously satisfied. Hence, $\widehat{S}^3_{\mathrm{path}}=\{4\}$ exclusion-blocks $\mu^2$ via $\mu^3$. Applying Definition~\ref{defi.postblocking} to $\widehat{S}^3_{\mathrm{path}}$ at $\mu^2$ with blocking allocation $\mu^3$ reproduces $\mu^3$. Thus, the blocking is self-contained and establishes the transition $\mu^2 \mapsto \mu^3$. By Proposition~\ref{prop.sttc-exclusion}, $\mu^3=\mu^\rhd$ is exclusion core stable. Therefore, $\mu=\mu^0\mapsto\mu^1\mapsto\mu^2\mapsto\mu^3=\mu^\rhd$ is a self-contained exclusion-blocking path from $\mu$ to the exclusion core. Its path coalitions are $\widehat{S}^1_{\mathrm{path}}=\{1,5,6\}$, $\widehat{S}^2_{\mathrm{path}}=\{3\}$, and $\widehat{S}^3_{\mathrm{path}}=\{4\}$. These are the myopic blocking coalitions associated with the three successive transitions.

\subparagraph{External-stability coalition.}
Define the external-stability coalition\nopagebreak\vspace{-0.2cm}
\[\widehat{S}^3 = \{i \in \bigcup_{k=1}^{3}S(\mathscr{C}^k) \mid \mu^{\rhd}(i) \mathbin{P_i} \mu(i)\} = \{1, 5, 6\} = \widehat{S}^2.\vspace{-0.2cm}\]
The external-stability coalition does not expand at Step~3. Agent~$4$ is a local improver because $\mu^2(4)=\varnothing$, $\mu^\rhd(4)=a$, and he strictly prefers $a$ to $\varnothing$; therefore, $4\in\widehat{S}^3_{\mathrm{path}}$. However, he is not a global improver because $\mu(4)=d$, $\mu^\rhd(4)=a$, and he strictly prefers $d$ to $a$; therefore, $4\notin\widehat{S}^3$. This is precisely the distinction captured by the joint proof. Agent~$4$ was already displaced at Step~2, and $\mu(4)=d$ already belongs to $\Omega(\widehat{S}^2\mid\mu)$. At Step~3, he improves locally by receiving the unassigned public object~$a$, so no additional outsider is displaced and no additional exclusion rights are needed for the global blocking. Hence, the same coalition $\widehat{S}^3=\widehat{S}^2=\{1,5,6\}$ exclusion-blocks $\mu$ via the exclusion-core-stable allocation $\mu^3=\mu^\rhd$.

\paragraph{Termination.} Let $N^4 = N^3\setminus\{4\}= \{7\}$ and $O^4 = O^3\setminus\{a\}=\emptyset$. Since $N^4\neq \emptyset$ and $O^4=\emptyset$, we assign the null object to every agent in $N^4$, and the algorithm terminates after Step~$3$ with the final STTC allocation \(\underline{\mu^\rhd}= (\underline{c},\underline{b}, \underline{d},\underline{a},\underline{e}, \underline{f},\underline{\varnothing})\).

\paragraph{Exclusion-core characterization.}
We finally show that for problem $R$, the exclusion core consists of exactly two allocations: $\mu^{\rhd}$, in which agent~$4$ receives $a$, and \(\mu^\star=(c,b,d,\varnothing,e,f,a)\), in which agent~$7$ receives $a$.
By Proposition~\ref{prop.corecharacterization}, every exclusion-core-stable allocation can be obtained by running the STTC algorithm initialized at the null allocation $\mu^{\varnothing}$ with some priority order. Fix an arbitrary priority order $\rhd'$. We now go through the $(\rhd',\mu^\varnothing)$-STTC algorithm and show that only two outcomes are possible.

At Step~$1$ of the $(\rhd',\mu^{\varnothing})$-STTC algorithm, agent $1$ points to object $c$, object $c$ points to its owner, agent $2$, agent $2$ points to object $b$, and object $b$ points to its owner, agent $1$. Hence, under every priority order, the non-$\mu^{\varnothing}$-cycle $1 \rightarrow c \rightarrow 2 \rightarrow b \rightarrow 1$ forms, assigning object $c$ to agent $1$ and object $b$ to agent $2$.

At Step~$2$, object $c$ has been removed, so agent $3$ points to his own object $d$ and, under every priority order, forms the non-$\mu^{\varnothing}$-cycle $3 \rightarrow d \rightarrow 3$, receiving object $d$.

Next, consider the public objects $e$ and $f$. As long as both objects remain, they point to the same agent, namely the highest-priority remaining agent. More generally, each remaining unassigned public object points to that agent. Only agent~$5$ points to $e$, and only agent~$6$ points to $f$. Hence, a cycle through $e$ forms exactly when agent~$5$ is the highest-priority remaining agent, in which case it is the self-cycle $5 \rightarrow e \rightarrow 5$. Similarly, a cycle through $f$ forms exactly when agent~$6$ is the highest-priority remaining agent, in which case it is the self-cycle $6 \rightarrow f \rightarrow 6$. Since no object points to agent $5$ or agent $6$ before that moment, neither agent can leave the algorithm earlier. Consequently, under every priority order, agent $5$ receives object $e$ and agent $6$ receives object $f$, at steps that depend on the priority order.

Finally, consider agents $4$ and $7$ and object $a$. Once agent $1$ has left at Step~$1$, object $a$ has no remaining owner, becomes public in the reduced problem, and, being unassigned at $\mu^{\varnothing}$, points to the highest-priority remaining agent. Agent $7$ points to object $a$ from Step~$1$ on, and agent $4$ points to object $a$ from Step~$3$ on, once object $d$ has been removed at Step~$2$. No other agent ever points to object $a$. Moreover, agents $4$ and $7$ point only to object $a$ (or, for agent $4$, to object $d$ while it is present) or to the null object, so neither agent can receive any object other than $a$. Hence, object $a$ is assigned in a self-cycle with whichever of agents $4$ and $7$ has the higher priority, once all higher-priority remaining agents have left, and the other agent receives $\varnothing$ when the algorithm terminates. In particular, if $4 \rhd' 7$, then agent~$4$ receives~$a$ and agent~$7$ receives $\varnothing$, yielding $\mu^{\rhd}$. If $7 \rhd' 4$, then agent~$7$ receives~$a$ and agent~$4$ receives~$\varnothing$, yielding $\mu^\star$. Since one of these two cases applies to every priority order, the exclusion core consists of exactly $\mu^{\rhd}$ and $\mu^\star$.\hfill~$\diamond$
\end{exam}

\section{Acyclic endowment structures}\label{appendix:acyclic}

Theorems~\ref{theorem:vNMstability} and~\ref{theorem:pathtoexclusion} are established for mixed-ownership problems. This appendix examines how far they reach. We first characterize the acyclic endowment structures, a class that strictly contains the mixed-ownership structures and splits into two cases: the mixed-ownership structures themselves, and the structures under which a single group of at least two, but not all, agents co-owns every object. The first case is covered by Theorems~\ref{theorem:vNMstability} and~\ref{theorem:pathtoexclusion}. For the second, we establish von Neumann--Morgenstern stability in Theorem~\ref{thm.uniformvNM} and the existence of blocking paths in Theorem~\ref{thm.uniformpath}. Corollaries~\ref{cor.acyclicvNM} and~\ref{cor.acyclicpath} then combine the two cases. Finally, two examples show what can happen when the endowment structure is not acyclic.

\paragraph{Acyclicity.} \citet{balbuzanov2019endowments} obtain positive results for relational economies whose priority structure is acyclic. Transcribing their condition into our setting yields the following restriction on endowment structures. Recall from Section~\ref{subsec.maximaldomain} that $S^o$ denotes the set of owners of object~$o$. The restriction matters only for properly co-owned objects. If an object $o$ has at least two owners and at least one non-owner, then every owner of $o$ also owns every other object, and no non-owner of $o$ owns any other object. Thus, different groups of agents cannot co-own different objects, as they do in Example~\ref{exam:noexstable}.

\begin{defi}[\textbf{Acyclic endowment structures}]\label{defi.acyclic}
An endowment structure is \textbf{acyclic} if, for all distinct objects $o,o'\in O$ and all distinct agents $i,j,k\in N$ with $i,k\in S^o$ and $j\notin S^o$, we have $k\in S^{o'}$ and $j\notin S^{o'}$.\footnote{Condition~(1) of \citet{balbuzanov2019endowments} states that a priority structure $\rhd$ is acyclic if for each $o'\neq o$, $i\rhd_o j$ and not $i\rhd_o k$ together imply $k\rhd_{o'}j$. Their Section~4.4 calls $\rhd$ a representation of an endowment structure if $i\rhd_o j$ exactly when $i\in S^o$ and $j\notin S^o$. Rewriting Condition~(1) in terms of the sets $S^o$ gives: if $i,k\in S^o$ and $j\notin S^o$, then $k\in S^{o'}$ and $j\notin S^{o'}$. \citet[][Footnote~11]{zhang2026consistent} transcribes the same condition into the complex-endowment setting.}
\end{defi}

We thank Jacob Coreno for suggesting the following transparent description of acyclic endowment structures.

\begin{prop}[\textbf{A characterization of acyclic endowment structures}]\label{prop.acyclicchar}
An endowment structure is acyclic if and only if
\begin{enumerate}
\item[\textup{(a)}] no object is properly co-owned, i.e., the structure is a
    mixed-ownership structure; or
\item[\textup{(b)}] there is a coalition $T\subseteq N$ such that, for each object
    $o\in O$, $S^o=T$.
\end{enumerate}
\end{prop}

\begin{proof}[\textbf{Proof of Proposition~\ref{prop.acyclicchar}}]
Fix an endowment structure $\omega$. We first show that~(a) implies acyclicity. Suppose that~(a) holds and fix an object $o$. If $o$ is public, then $S^o=N$, so no agent lies outside $S^o$. If $o$ is individually owned by an agent $h$, then $S^o=\{h\}$, so no two distinct agents both lie in $S^o$. In either case there are no distinct agents $i,j,k$ with $i,k\in S^o$ and $j\notin S^o$. Acyclicity requires something only of such agents, and $o$ was arbitrary, so the structure is acyclic.

We next show that~(b) implies acyclicity. Suppose that~(b) holds with owning coalition $T$. Let $o$ and $o'$ be distinct objects and $i,j,k$ distinct agents with $i,k\in S^o$ and $j\notin S^o$. Since $S^{o'}=T=S^o$, we have $k\in S^{o'}$ and $j\notin S^{o'}$. As the objects and agents were arbitrary, the structure is acyclic.

Conversely, suppose that the structure is acyclic; we show that~(a) or~(b) holds. If $|O|\leq1$, then~(b) holds: take $T:=S^o$ for the unique object $o$ when $|O|=1$, and $T:=N$ when $O=\emptyset$. Hence, suppose that $|O|\geq2$ and that~(a) fails, so that some object $o$ satisfies $2\leq|S^o|<|N|$.

Fix any object $o'\neq o$. To show $S^o\subseteq S^{o'}$, let $k\in S^o$ be arbitrary. Since $|S^o|\geq2$, choose an agent $i\in S^o$ with $i\neq k$, and since $S^o\neq N$, choose an agent $j\notin S^o$. Then $i,j,k$ are distinct, so acyclicity implies $k\in S^{o'}$. Hence, $S^o\subseteq S^{o'}$.
To show the reverse inclusion, let $j\notin S^o$ be arbitrary. Since $|S^o|\geq2$, choose distinct agents $i,k\in S^o$. Then acyclicity implies $j\notin S^{o'}$. Thus, every agent outside $S^o$ also lies outside $S^{o'}$, which is equivalent to $S^{o'}\subseteq S^o$. Hence, $S^{o'}=S^o$. Since $o'\neq o$ was arbitrary, every object has the same set of owners, which is~(b).
\end{proof}

The two cases in Proposition~\ref{prop.acyclicchar} overlap. In case~(b) with $|T|=1$, a single agent owns every object, and with $T=N$, every object is public; both are mixed-ownership structures and are already covered by~(a). The genuinely new case has $2\leq|T|<|N|$ and, for each object $o\in O$, $S^o=T$; in this case one group co-owns everything and the remaining agents own nothing. We refer to it as \textbf{uniform co-ownership}. It is the situation of the ``diarchy'' of \citet{balbuzanov2019endowments}, where two co-monarchs jointly own all houses and the remaining agents do not own any real object.

\paragraph{Uniform co-ownership.} Throughout the remainder of this paragraph the endowment structure is uniformly co-owned with owning coalition $T$, where $2\leq|T|<|N|$; in particular, $T$ is nonempty and may therefore serve as a blocking coalition. The key simplification is that exclusion rights are all-or-nothing.

\begin{lemm}[\textbf{Collapse of the extended eviction sets}]\label{lemm.collapse}
Under uniform co-ownership, for each coalition $S$ and each allocation $\mu$,\nopagebreak\vspace{-0.2cm}
\[ \Omega(S\mid\mu)=
\begin{cases}
O, &\text{ if } T\subseteq S,\\
\emptyset, &\text{ otherwise.}
\end{cases} \vspace{-0.2cm}\]
\end{lemm}

\begin{proof}[\textbf{Proof of Lemma~\ref{lemm.collapse}}]
Let $T$ be the owning coalition, so that $S^o=T$ for each object $o\in O$, and let $S$ be a coalition and $\mu$ an allocation. By definition, $\omega(S)=\{o\in O\mid S^o\subseteq S\}$. If $T\not\subseteq S$, then $\omega(S)=\emptyset$, so $\mu^{-1}(\omega(S))=\emptyset$ and $S^\mu_1=S$; the recursion is already at a fixed point and $\Omega(S\mid\mu)=\omega(S)=\emptyset$. If $T\subseteq S$, then, for each $k$, $S\subseteq S^\mu_k$, so $T\subseteq S^\mu_k$ and hence $\omega(S^\mu_k)=O$; therefore $\Omega(S\mid\mu)=O$.
\end{proof}

For an allocation $\mu$, let $A(\mu):=\{\nu\mid \text{for each }i\in T,\ \nu(i)\mathbin{P_i}\mu(i)\}$ be the set of allocations at which every member of the owning coalition is strictly better off than at $\mu$.

\begin{lemm}[\textbf{The exclusion core under uniform co-ownership}]\label{lemm.uniformEC}
Under uniform co-ownership, an allocation $\mu$ is exclusion core stable if and only if $\mu$ is Pareto efficient and $A(\mu)=\emptyset$.
\end{lemm}

\begin{proof}[\textbf{Proof of Lemma~\ref{lemm.uniformEC}}]
Suppose that $\mu$ is Pareto efficient with $A(\mu)=\emptyset$, and suppose that some coalition $S$ exclusion-blocks $\mu$ via $\sigma$. If $T\subseteq S$, then each member of $T$ is strictly better off at $\sigma$, so $\sigma\in A(\mu)$, a contradiction. If $T\not\subseteq S$, then $\Omega(S\mid\mu)=\emptyset$ by Lemma~\ref{lemm.collapse}, so no agent outside $S$ is worse off at $\sigma$; since $S\neq\emptyset$ and every member of $S$ is strictly better off, $\sigma$ Pareto dominates $\mu$, contradicting Pareto efficiency.

Conversely, suppose that $\mu$ is exclusion core stable. If some allocation $\sigma$ Pareto dominates $\mu$, let $S:=\{i\mid\sigma(i)\mathbin{P_i}\mu(i)\}$, which is nonempty. No agent outside $S$ is worse off at $\sigma$, so the exclusion-enforcement condition is vacuously satisfied and $S$ exclusion-blocks $\mu$ via $\sigma$, a contradiction. Hence, $\mu$ is Pareto efficient. If $\sigma\in A(\mu)$, define $\sigma'$ such that for each $i\in T$, $\sigma'(i)=\sigma(i)$, and for each $j\notin T$, $\sigma'(j)=\varnothing$.
This is an allocation, and each member of $T$ is strictly better off at $\sigma'$. An agent $j\notin T$ is worse off only if $\mu(j)\mathbin{P_j}\varnothing$, which requires $\mu(j)\in O=\Omega(T\mid\mu)$. Hence, $T$ exclusion-blocks $\mu$ via $\sigma'$, a contradiction. Therefore, $A(\mu)=\emptyset$.
\end{proof}

\begin{theorem}[\textbf{Von Neumann--Morgenstern stability under uniform co-ownership}]\label{thm.uniformvNM}
Under uniform co-ownership, the exclusion core is a von Neumann--Morgenstern stable set with respect to exclusion-domination.
\end{theorem}

\begin{proof}[\textbf{Proof of Theorem~\ref{thm.uniformvNM}}]
Let $R$ be a problem. Internal stability is immediate from the definition of the exclusion core: no allocation in $\EC(R)$ is exclusion-blocked, hence none is exclusion-dominated by any allocation.

For external stability, let $\mu\notin\EC(R)$. By Lemma~\ref{lemm.uniformEC}, either $A(\mu)\neq\emptyset$ or $\mu$ is not Pareto efficient. Suppose that $\mu$ is not Pareto efficient, and choose a Pareto efficient allocation $\nu$ that Pareto dominates $\mu$; such a $\nu$ exists because there are finitely many allocations. If $\nu\in\EC(R)$, then no agent outside $S:=\{i\mid\nu(i)\mathbin{P_i}\mu(i)\}$ is worse off at $\nu$, so this nonempty coalition exclusion-blocks $\mu$ via $\nu$; thus $\mu$ is exclusion-dominated by an allocation in $\EC(R)$, as required. If $\nu\notin\EC(R)$, then, since $\nu$ is Pareto efficient, Lemma~\ref{lemm.uniformEC} gives $A(\nu)\neq\emptyset$; choose $\nu'\in A(\nu)$. For each $i\in T$, $\nu'(i)\mathbin{P_i}\nu(i)\mathbin{R_i}\mu(i)$, so $\nu'\in A(\mu)$. In every case that remains, $A(\mu)\neq\emptyset$.

We claim that $\EC(R)\cap A(\mu)\neq\emptyset$. Choose $\nu^0\in A(\mu)$ and construct a sequence as follows. Given $\nu^k$, stop if $\nu^k\in\EC(R)$. Otherwise, by Lemma~\ref{lemm.uniformEC}, either $\nu^k$ is not Pareto efficient, in which case let $\nu^{k+1}$ Pareto dominate $\nu^k$; we call this a \emph{step of the first kind}. Or $A(\nu^k)\neq\emptyset$, in which case let $\nu^{k+1}\in A(\nu^k)$; we call this a \emph{step of the second kind}. In both cases, $\nu^{k+1}\in A(\mu)$: by induction, each $i\in T$ satisfies $\nu^k(i)\mathbin{P_i}\mu(i)$; a step of the first kind gives $\nu^{k+1}(i)\mathbin{R_i}\nu^k(i)$ and a step of the second kind gives $\nu^{k+1}(i)\mathbin{P_i}\nu^k(i)$, so in either case $\nu^{k+1}(i)\mathbin{P_i}\mu(i)$ by transitivity.

The construction terminates. Fix an agent $i_0\in T$. Every step leaves $i_0$'s allotment weakly better and every step of the second kind makes it strictly better, so there are at most $|O|$ steps of the second kind. These steps divide the sequence into at most $|O|+1$ blocks, and every step within a block is of the first kind. At such a step the new allocation Pareto dominates the previous one, so within a block each allocation is Pareto dominated by the next. Since there are finitely many allocations and Pareto domination is transitive and irreflexive, each block is finite. As there are finitely many blocks, the construction terminates. Let $\nu^*$ be the terminal allocation. The construction stops only at an allocation in $\EC(R)$, so $\nu^*\in\EC(R)$; and every allocation in the sequence lies in $A(\mu)$, so $\nu^*\in A(\mu)$. Hence, $\nu^*\in\EC(R)\cap A(\mu)$.

Finally, $T$ exclusion-blocks $\mu$ via $\nu^*$. Since $\nu^*\in A(\mu)$, each member of $T$ is strictly better off. Let $j\notin T$ be worse off at $\nu^*$. As $\nu^*$ is Pareto efficient, no agent is assigned an object ranked below the null object, for otherwise reassigning that agent to $\varnothing$ would be a Pareto improvement; hence $\nu^*(j)\mathbin{R_j}\varnothing$ and therefore $\mu(j)\neq\varnothing$. Thus, $\mu(j)\in O=\Omega(T\mid\mu)$ by Lemma~\ref{lemm.collapse}. Hence, $\mu$ is exclusion-dominated by $\nu^*\in\EC(R)$.
\end{proof}

\begin{theorem}[\textbf{Paths to exclusion core stability under uniform co-ownership}]\label{thm.uniformpath}
Under uniform co-ownership, for each allocation $\mu$ there is a self-contained exclusion-blocking path from $\mu$ to the exclusion core.
\end{theorem}

\begin{proof}[\textbf{Proof of Theorem~\ref{thm.uniformpath}}]
Let $R$ be a problem and let $\mu$ be an allocation. If $\mu\in\EC(R)$, the constant sequence consisting only of $\mu$ is a self-contained exclusion-blocking path of length zero. Hence, suppose that $\mu\notin\EC(R)$ and set $\mu^0:=\mu$. We construct a sequence recursively. Given $\mu^k$, stop if $\mu^k\in\EC(R)$. Otherwise, proceed as follows.

First, suppose that $\mu^k$ is not Pareto efficient. Let $\mu^{k+1}$ Pareto dominate $\mu^k$ and set $S:=\{i\mid\mu^{k+1}(i)\mathbin{P_i}\mu^k(i)\}$, which is nonempty. No agent is worse off at $\mu^{k+1}$, so the exclusion-enforcement condition is vacuously satisfied and $S$ exclusion-blocks $\mu^k$ via $\mu^{k+1}$. Let $j\notin S$. Since $\mu^{k+1}$ Pareto dominates $\mu^k$, we have $\mu^{k+1}(j)\mathbin{R_j}\mu^k(j)$, and $j\notin S$ rules out strict preference, so $\mu^{k+1}(j)=\mu^k(j)$. Since $\mu^{k+1}(j)=\mu^k(j)$ and $\mu^{k+1}$ assigns each real object to at most one agent, $\mu^k(j)\notin\mu^{k+1}(S)\setminus\{\varnothing\}$. Therefore, $\mu^{k+1}$ equals the post-blocking allocation induced by $(S,\mu^{k+1})$ at $\mu^k$, so the exclusion-blocking is self-contained. We call this a \emph{step of the first kind}.

Second, suppose that $\mu^k$ is Pareto efficient. Since $\mu^k\notin\EC(R)$, Lemma~\ref{lemm.uniformEC} implies that $A(\mu^k)\neq\emptyset$. Fix $\nu\in A(\mu^k)$ and define\nopagebreak\vspace{-0.2cm}
\[ \mu^{k+1}(i):=
\begin{cases}
\nu(i) &\text{ if } i\in T,\\
\varnothing &\text{ if } i\notin T \text{ and } \mu^k(i)\in\nu(T)\setminus\{\varnothing\},\\
\mu^k(i) &\text{ otherwise.}
\end{cases} \vspace{-0.2cm}\]

This is an allocation: $\nu$ assigns each real object to at most one member of $T$, and an agent outside $T$ retains $\mu^k(i)$ only if no member of $T$ receives it. Each $i\in T$ satisfies $\mu^{k+1}(i)=\nu(i)\mathbin{P_i}\mu^k(i)$. An agent $j\notin T$ is worse off only if $\mu^{k+1}(j)=\varnothing\neq\mu^k(j)$, which by Lemma~\ref{lemm.collapse} occurs only if $\mu^k(j)\in\nu(T)\setminus\{\varnothing\}\subseteq O=\Omega(T\mid\mu^k)$. Hence, $T$ exclusion-blocks $\mu^k$ via $\mu^{k+1}$.
Since $\mu^{k+1}(T)\setminus\{\varnothing\}=\nu(T)\setminus\{\varnothing\}$, the displayed map is precisely the post-blocking allocation induced by $(T,\mu^{k+1})$ at $\mu^k$, so the exclusion-blocking is self-contained. Note also that $\mu^{k+1}\in A(\mu^k)$. We call this a \emph{step of the second kind}.

\paragraph{Termination.} Fix an agent $i_0\in T$. Every step leaves $i_0$'s allotment weakly better and every step of the second kind makes it strictly better, so there are at most $|O|$ steps of the second kind. These steps divide the sequence into at most $|O|+1$ blocks, and every step within a block is of the first kind. At such a step the new allocation Pareto dominates the previous one, so within a block each allocation is Pareto dominated by the next. Since there are finitely many allocations and Pareto domination is transitive and irreflexive, each block is finite. As there are finitely many blocks, the construction terminates. By Lemma~\ref{lemm.uniformEC}, the terminal allocation is exclusion core stable. Since every transition in the constructed sequence is a self-contained exclusion-blocking, the sequence is a self-contained exclusion-blocking path from $\mu$ to the exclusion core.
\end{proof}

Combining Proposition~\ref{prop.acyclicchar} with Theorems~\ref{theorem:vNMstability} and~\ref{thm.uniformvNM} yields the following.

\begin{corollary}[\textbf{Von Neumann--Morgenstern stability under acyclic endowment structures}]\label{cor.acyclicvNM} Suppose that the endowment structure is acyclic. Then the exclusion core is a von Neumann--Morgenstern stable set with respect to exclusion-domination.
\end{corollary}

Similarly, combining Proposition~\ref{prop.acyclicchar} with Theorems~\ref{theorem:pathtoexclusion} and~\ref{thm.uniformpath} yields the following.

\begin{corollary}[\textbf{Paths to exclusion core stability under acyclic endowment structures}]\label{cor.acyclicpath}
Suppose that the endowment structure is acyclic. Then, from each allocation, there is a self-contained exclusion-blocking path to the exclusion core.
\end{corollary}

Corollary~\ref{cor.acyclicvNM} for the exclusion core of an object allocation problem with an acyclic endowment structure is the counterpart of Proposition~4 of \citet{balbuzanov2019endowments}, which establishes von Neumann--Morgenstern stability of the strong and weak exclusion cores of a relational economy with an acyclic priority structure. The same acyclicity condition supports both results, although the two solution concepts differ in general and coincide only under mixed ownership.

\paragraph{Acyclicity and the exclusion core characterization.} Corollaries~\ref{cor.acyclicvNM} and~\ref{cor.acyclicpath} do not weaken Corollary~\ref{cor:maxdomain}. Acyclicity is a domain for both conclusions but not for the characterization of the exclusion core by the STTC algorithms.
Under uniform co-ownership with $2\leq|T|<|N|$ every object is properly co-owned, so the construction proving Theorem~\ref{thm.strictsuperset} applies: it requires only a properly co-owned object and an agent who is not one of its co-owners, both of which are present.
Mixed ownership therefore remains the maximal domain on which the STTC algorithms characterize the exclusion core, while it is not maximal for von Neumann--Morgenstern stability or for the existence of blocking paths.

\paragraph{Beyond acyclicity.} Both examples below have non-acyclic endowment structures: in the first, external stability and the existence of self-contained exclusion-blocking paths to the exclusion core both fail; in the second, both hold at every preference profile. Although Lemma~\ref{lemm.paretoexclusion} is stated for mixed-ownership problems, its argument does not use the ownership structure: if $\sigma$ Pareto dominates $\mu$, no agent outside the coalition of agents who are strictly better off at $\sigma$ is worse off, so that nonempty coalition exclusion-blocks $\mu$ via $\sigma$. Hence, in both examples below, every exclusion-core-stable allocation is Pareto efficient.

\begin{exam}[\textbf{A cycle of exclusion-domination outside the exclusion core}]\label{exam:noexstable}
This example is due to Jacob Coreno. Let $N=\{1,2,3,4,5\}$ and $O=\{a,b,c\}$. Agent~$1$ owns object $a$, agents~$2$ and~$3$ co-own object $b$, and agents~$4$ and~$5$ co-own object $c$; that is, $S^a=\{1\}$, $S^b=\{2,3\}$, and $S^c=\{4,5\}$. Let $\omega$ denote this endowment structure. Objects $b$ and $c$ are properly co-owned, and $S^a\neq S^b$; by Proposition~\ref{prop.acyclicchar}, $\omega$ is therefore not acyclic. Preferences are as follows, where ``$\ldots$'' denotes an arbitrary ranking of the remaining objects below the null object; none of the conclusions below depends on these rankings.\nopagebreak\vspace{-0.2cm}
\[\begin{array}{ccccc}
R_1: & c & \varnothing & \ldots\\
R_2: & c & a & b & \varnothing\\
R_3: & c & a & b & \varnothing\\
R_4: & b & \varnothing & \ldots\\
R_5: & a & \varnothing & \ldots
\end{array}\vspace{-0.2cm}\]
Consider the allocations\nopagebreak\vspace{-0.2cm}
\[ \mu^1:=(\varnothing,b,c,\varnothing,a),\quad \mu^2:=(c,a,\varnothing,b,\varnothing),\quad \mu^3:=(\varnothing,c,b,\varnothing,a),\quad \mu^4:=(c,\varnothing,a,b,\varnothing),\vspace{-0.2cm}\]
let $\mu^5:=\mu^1$, and let $X:=\{\mu^1,\mu^2,\mu^3,\mu^4\}$. Moreover, let $T_1:=\{1,2,4\}$, $T_2:=\{2,3,5\}$, $T_3:=\{1,3,4\}$, and $T_4:=\{2,3,5\}$. We show the following.\par\smallskip

\noindent\emph{For each $k\in\{1,2,3,4\}$, $\mu^{k+1}$ is the unique allocation that exclusion-dominates $\mu^k$; coalition $T_k$ exclusion-blocks $\mu^k$ via $\mu^{k+1}$, and the post-blocking allocation induced by $(T_k,\mu^{k+1})$ at $\mu^k$ is $\mu^{k+1}$.}\par\smallskip

\noindent This has three consequences. First, no allocation in $X$ is exclusion core stable, and no allocation in $X$ is exclusion-dominated by an exclusion-core-stable allocation; hence, $\EC(\omega,R)$ is not externally stable. Second, suppose that a coalition $S$ exclusion-blocks $\mu^k$ via $\sigma$ in a self-contained way. Then $S$ also exclusion-blocks $\mu^k$ via the induced post-blocking allocation. This post-blocking allocation therefore exclusion-dominates $\mu^k$, so it equals $\mu^{k+1}$. Hence, every self-contained exclusion-blocking path that starts in $X$ follows the cycle $\mu^1\mapsto\mu^2\mapsto\mu^3\mapsto\mu^4\mapsto\mu^1\mapsto\cdots$ and never reaches the exclusion core. Thus, the conclusion of Corollary~\ref{cor.acyclicpath} does not extend beyond acyclic endowment structures. Third, for each initial allocation in $X$, the random process of Section~\ref{subsec.paths} moves deterministically along this cycle, whatever the selection probabilities, so the conclusion of Corollary~\ref{coro:randompath} does not extend either.

\paragraph{Three properties of the recursion.} Fix an allocation $\mu$ and a coalition $S$, and let $S^\mu_\ell$ denote the terminal set of the recursion defining $\Omega(S\mid\mu)$.
\begin{enumerate}
\item[(P1)] \emph{If $\mu(i)=\varnothing$, then $i\in S^\mu_\ell$ if and only if $i\in S$.} The recursion adds only agents whose allotments belong to a direct eviction set, which contains only real objects.
\item[(P2)] \emph{If $i\notin S$, $\mu(i)$ is a real object, and $i\in S^{\mu(i)}$, then $\mu(i)\notin\Omega(S\mid\mu)$.} Agent $i$ could enter the recursion only once $\mu(i)$ belongs to the current direct eviction set, which requires all owners of $\mu(i)$, including $i$ himself, to have entered already. Hence, $i\notin S^\mu_\ell$, so $S^{\mu(i)}\not\subseteq S^\mu_\ell$ and $\mu(i)\notin\omega(S^\mu_\ell)$.
\item[(P3)] \emph{If $\omega(S)=\emptyset$, then $\Omega(S\mid\mu)=\emptyset$}, so $S$ exclusion-blocks $\mu$ via $\sigma$ only if no agent outside $S$ is worse off at $\sigma$.
\end{enumerate}

\paragraph{The allocation $\mu^1=(\varnothing,b,c,\varnothing,a)$.} Suppose that $S$ exclusion-blocks $\mu^1$ via $\sigma$. Agents~$3$ and~$5$ receive their most-preferred objects, so $S\subseteq\{1,2,4\}$.
Suppose that $1\notin S$. Then $S\subseteq\{2,4\}$ and $\omega(S)=\emptyset$, so, by (P3), agents~$3$ and~$5$ keep $c$ and $a$ at $\sigma$. Agent~$2$ prefers only $c$ and $a$ to $b$, so $2\notin S$, and, not being worse off, he keeps $b$. Then $S=\{4\}$, but agent~$4$ prefers only $b$ to $\varnothing$, a contradiction. Hence, $1\in S$ and $\sigma(1)=c$.
Agent~$3$ is therefore worse off, so $c\in\Omega(S\mid\mu^1)$ and $S^c=\{4,5\}\subseteq S^{\mu^1}_\ell$. Since $\mu^1(4)=\varnothing$, (P1) gives $4\in S$, and hence $\sigma(4)=b$.
If $2\in S$, then $\sigma(2)\in\{c,a\}$. If $2\notin S$, then, by (P2), $b\notin\Omega(S\mid\mu^1)$, so agent~$2$ is not worse off and $\sigma(2)\in\{c,a,b\}$. Since $c$ and $b$ are assigned to agents~$1$ and~$4$, in either case $\sigma(2)=a$, and therefore $\sigma(3)=\sigma(5)=\varnothing$. Thus, $\sigma=\mu^2$.
Conversely, the recursion for $T_1=\{1,2,4\}$ at $\mu^1$ adds agent~$5$, who holds $a\in\omega(T_1)$, and then agent~$3$, who holds $c\in\omega(\{1,2,4,5\})$; hence, $\Omega(T_1\mid\mu^1)=O$. Every member of $T_1$ is strictly better off at $\mu^2$ ($c\mathbin{P_1}\varnothing$, $a\mathbin{P_2}b$, $b\mathbin{P_4}\varnothing$), so $T_1$ exclusion-blocks $\mu^1$ via $\mu^2$. The outsiders, agents~$3$ and~$5$, lose $c$ and $a$ to members of $T_1$ and receive $\varnothing$, so the induced post-blocking allocation is $\mu^2$.

\paragraph{The allocation $\mu^2=(c,a,\varnothing,b,\varnothing)$.} Suppose that $S$ exclusion-blocks $\mu^2$ via $\sigma$. Agents~$1$ and~$4$ receive their most-preferred objects, so $S\subseteq\{2,3,5\}$.
Suppose that $\{2,3\}\not\subseteq S$. Then $\omega(S)=\emptyset$, so, by (P3), agents~$1$ and~$4$ keep $c$ and $b$. Agent~$2$ prefers only $c$ to $a$, so $2\notin S$, and he keeps $a$. Then every member of $S\subseteq\{3,5\}$ receives $\varnothing$ at $\sigma$, a contradiction. Hence, $\{2,3\}\subseteq S$ and $\sigma(2)=c$.
Agent~$1$ is therefore worse off, so $\{4,5\}\subseteq S^{\mu^2}_\ell$, and, since $\mu^2(5)=\varnothing$, (P1) gives $5\in S$ and $\sigma(5)=a$. Agent~$3$ must then receive $b$, and $\sigma(1)=\sigma(4)=\varnothing$. Thus, $\sigma=\mu^3$.
Conversely, the recursion for $T_2=\{2,3,5\}$ at $\mu^2$ adds agent~$4$, who holds $b$, and then agent~$1$, who holds $c$; hence, $\Omega(T_2\mid\mu^2)=O$. Every member of $T_2$ is strictly better off at $\mu^3$, and the outsiders, agents~$1$ and~$4$, lose $c$ and $b$ to members of $T_2$ and receive $\varnothing$. Hence, $T_2$ exclusion-blocks $\mu^2$ via $\mu^3$, and the induced post-blocking allocation is $\mu^3$.

\paragraph{The allocations $\mu^3$ and $\mu^4$.} Agents~$2$ and~$3$ have the same preferences and co-own the same object. Interchanging their allotments therefore preserves exclusion-blocking, exclusion-domination, and post-blocking allocations, provided that the two agents are also interchanged in the coalitions. This interchange maps $\mu^1$, $\mu^2$, $T_1$, and $T_2$ to $\mu^3$, $\mu^4$, $T_3$, and $T_4$, respectively. Hence, the two cases above yield the claims for $\mu^3$ and $\mu^4$.\par\smallskip

For completeness, we describe the exclusion core. It consists of the seven allocations
\[ \begin{array}{llll}
(c,b,\varnothing,\varnothing,a), & (c,\varnothing,b,\varnothing,a), & (\varnothing,c,\varnothing,b,a), & (\varnothing,\varnothing,c,b,a),\\
(\varnothing,c,a,b,\varnothing), & (\varnothing,a,c,b,\varnothing), & \text{and} & \mu^\star:=(c,\varnothing,\varnothing,b,a).
\end{array}\]
The first six are the outcomes of the extended $(\rhd,\mu^\varnothing)$-STTC algorithms. Allocation $\mu^\star$ assigns the properly co-owned object $b$ to a non-co-owner while both of its co-owners receive $\varnothing$; as in the proof of Theorem~\ref{thm.strictsuperset}, it is exclusion core stable but generated by no owner-respecting flexible TTC algorithm. From every allocation outside $X\cup\EC(\omega,R)$, a self-contained exclusion-blocking path to the exclusion core exists; thus, the four allocations in $X$ are exactly those from which the exclusion core cannot be reached.\footnote{We verified these statements, as well as the claims proved above, by exhaustive enumeration of all $136$ allocations and all $31$ coalitions, for each of the eight completions of the rankings below~$\varnothing$. By contrast, a similar enumeration shows that, for every problem with three agents, two objects, and a properly co-owned object, a self-contained exclusion-blocking path to the exclusion core exists from every allocation, although the exclusion core need not be externally stable. The Python script is available as an ancillary file with the arXiv version of this paper (arXiv:2610.04015).}\hfill~$\diamond$
\end{exam}

The next example shows that, outside the acyclic class, external stability and the existence of paths to the exclusion core may nevertheless both hold at every preference profile.

\begin{exam}[\textbf{External stability and paths beyond acyclicity}]\label{exam:beyondacyclic}
Consider three agents $N=\{1,2,3\}$ and two objects $O=\{a,b\}$. Agents 1 and 2 co-own object $a$, while object $b$ is public; let $\omega$ denote this endowment structure. Since $2\leq|S^a|<|N|$, object $a$ is properly co-owned, so $\omega$ is not a mixed-ownership structure, and $S^a\neq S^b$, so the two objects do not share a common set of owners. By Proposition~\ref{prop.acyclicchar}, the endowment structure $\omega$ is therefore not acyclic.
We show that the exclusion core is nevertheless externally stable for \emph{every} preference profile and that, for every preference profile, a self-contained exclusion-blocking path to the exclusion core exists from every allocation.

Because $\omega(\{1,2\})=\{a\}$, $\omega(N)=O$, and, for every coalition $S\notin\{\{1,2\},N\}$, $\omega(S)=\emptyset$, the extended eviction sets are\nopagebreak\vspace{-0.2cm}
\[ \Omega(\{1,2\}\mid\mu)=
\begin{cases}
\{a,b\} &\text{ if } \mu(3)=a,\\
\{a\} &\text{ otherwise,}
\end{cases}
\qquad
\Omega(N\mid\mu)=\{a,b\},\vspace{-0.2cm} \]
and, for every other coalition $S$, $\Omega(S\mid\mu)=\emptyset$. Indeed, if $\mu(3)=a$, then agent~3 occupies an object in $\omega(\{1,2\})$, so the recursion adds him and $\omega(\{1,2,3\})=O$.

For a preference profile $R$ and an allocation $\mu$, let $A(\mu)$ denote the set of allocations at which agents~1 and~2 are both strictly better off than at $\mu$. Then, \vspace{-0.2cm}
\begin{equation}\label{eq.beyondchar}
\mu\in\EC(\omega,R) \text{ if and only if }
\mu \text{ is Pareto efficient, and if } \mu(3)=a, \text{ then } A(\mu)=\emptyset.\vspace{-0.2cm}
\end{equation}

To verify~\eqref{eq.beyondchar}, suppose first that the right-hand side holds and that some coalition $S$ exclusion-blocks $\mu$ via $\sigma$. If $S=N$, then every agent is strictly better off at $\sigma$, so $\sigma$ Pareto dominates $\mu$, contradicting Pareto efficiency.
If $\{1,2\}\not\subseteq S$, then $\Omega(S\mid\mu)=\emptyset$, so no agent outside $S$ is worse off at $\sigma$; as every member of $S$ is strictly better off, $\sigma$ again Pareto dominates $\mu$. There remains $S=\{1,2\}$. If $\mu(3)\neq a$, then $\Omega(\{1,2\}\mid\mu)=\{a\}$, so agent~3 cannot be worse off and $\sigma$ Pareto dominates $\mu$ once more. If $\mu(3)=a$, then agents~1 and~2 are both strictly better off at $\sigma$, so $\sigma\in A(\mu)$, contradicting $A(\mu)=\emptyset$.

Conversely, suppose that $\mu\in\EC(\omega,R)$. If some allocation $\sigma$ Pareto dominates $\mu$, then no agent outside $S:=\{i\mid\sigma(i)\mathbin{P_i}\mu(i)\}$ is worse off, so this nonempty coalition exclusion-blocks $\mu$ via $\sigma$; hence $\mu$ is Pareto efficient. Suppose that $\mu(3)=a$ and that $\sigma\in A(\mu)$. Define the allocation $\widehat\sigma$ by\nopagebreak\vspace{-0.2cm}
\[ \widehat\sigma(i) :=
\begin{cases}
\sigma(i) &\text{ if } i\in\{1,2\},\\
\varnothing &\text{ if } i=3 \text{ and } a\in\sigma(\{1,2\}),\\
a &\text{ if } i=3 \text{ and } a\notin\sigma(\{1,2\}).
\end{cases} \vspace{-0.2cm}\]

Agents~1 and~2 are strictly better off at $\widehat\sigma$, and agent~3 is worse off only when $\widehat\sigma(3)=\varnothing\neq\mu(3)$, in which case $\mu(3)=a\in\Omega(\{1,2\}\mid\mu)=\{a,b\}$. Thus, $\{1,2\}$ exclusion-blocks $\mu$ via $\widehat{\sigma}$, contradicting $\mu\in\EC(\omega,R)$. Therefore, no $\sigma\in A(\mu)$ exists, so $A(\mu)=\emptyset$.

Given~\eqref{eq.beyondchar}, external stability can be verified by exhaustive enumeration. There are thirteen allocations, and each agent has six strict preference orderings over $O\cup\{\varnothing\}$, yielding $6^{3}=216$ preference profiles. For each profile, we compute the exclusion core directly from Definition~\ref{defi.exclusioncore}, confirm that it agrees with \eqref{eq.beyondchar}, and check that every allocation outside the exclusion core is exclusion-dominated by an allocation inside it. The verification succeeds for all 216 profiles. The same enumeration also confirms that, for each of the 216 profiles, a self-contained exclusion-blocking path to the exclusion core exists from every allocation.\footnote{The Python script used for this exhaustive verification is available as an ancillary file with the arXiv version of this paper (arXiv:2610.04015).}

The following profile illustrates the conclusion. Let $R$ be given by $R_{1}: a, b, \varnothing$, $R_{2}: a, b, \varnothing$, and $R_{3}: a, b, \varnothing$. Then, an allocation is Pareto efficient exactly when both objects are assigned, and among these it is exclusion core stable exactly when agent~3 does not hold $a$: if $\mu(3)=a$, then agents~1 and~2 hold $b$ and $\varnothing$ in some order, and giving $a$ to the one holding $b$ and $b$ to the one holding $\varnothing$ places an allocation in $A(\mu)$. Hence, by \eqref{eq.beyondchar}, $\EC(\omega,R)=\bigl\{(a,b,\varnothing),\ (b,a,\varnothing),\ (a,\varnothing,b),\ (\varnothing,a,b)\bigr\}$, and each of the nine remaining allocations is exclusion-dominated by an exclusion-core-stable allocation, as the following table records.

\begin{center}
\begin{tabular}{lllll}
$\mu$ & $S$ & $\nu$ & agents worse off & $\Omega(S\mid\mu)$\\[2pt]
$(\varnothing,\varnothing,\varnothing)$ & $\{1,2\}$ & $(a,b,\varnothing)$ & none & $\{a\}$\\
$(a,\varnothing,\varnothing)$ & $\{2\}$ & $(a,b,\varnothing)$ & none & $\emptyset$\\
$(\varnothing,a,\varnothing)$ & $\{1\}$ & $(b,a,\varnothing)$ & none & $\emptyset$\\
$(b,\varnothing,\varnothing)$ & $\{2\}$ & $(b,a,\varnothing)$ & none & $\emptyset$\\
$(\varnothing,b,\varnothing)$ & $\{1\}$ & $(a,b,\varnothing)$ & none & $\emptyset$\\
$(\varnothing,\varnothing,b)$ & $\{1\}$ & $(a,\varnothing,b)$ & none & $\emptyset$\\
$(\varnothing,\varnothing,a)$ & $\{1,2\}$ & $(a,b,\varnothing)$ & agent~3 & $\{a,b\}$\\
$(b,\varnothing,a)$ & $\{1,2\}$ & $(a,b,\varnothing)$ & agent~3 & $\{a,b\}$\\
$(\varnothing,b,a)$ & $\{1,2\}$ & $(b,a,\varnothing)$ & agent~3 & $\{a,b\}$\\
\end{tabular}\vspace{-0.2cm}
\end{center}

In each of the first six rows every member of $S$ is strictly better off at $\nu$ and no agent is worse off, so the exclusion-enforcement condition is vacuously satisfied. In each of the last three rows agents~1 and~2 are both strictly better off, and the only agent harmed is agent~3, whose allotment $\mu(3)=a$ lies in $\Omega(\{1,2\}\mid\mu)=\{a,b\}$. Moreover, in every row the blocking allocation $\nu$ coincides with the post-blocking allocation induced by $(S,\nu)$ at $\mu$, so each exclusion-blocking is self-contained and each of the nine allocations reaches the exclusion core by a self-contained exclusion-blocking path of length one.\hfill~$\diamond$
\end{exam}

Corollaries~\ref{cor.acyclicvNM} and~\ref{cor.acyclicpath} show that acyclicity is sufficient for external stability and for the existence of self-contained exclusion-blocking paths to the exclusion core. Example~\ref{exam:beyondacyclic} shows that acyclicity is necessary for neither property. Example~\ref{exam:noexstable} shows that, without acyclicity, both properties may fail. Hence, non-acyclicity alone determines neither property.

\newpage
\setlength{\bibsep}{6pt}
\begingroup\linespread{1.185}\selectfont

\endgroup

\end{document}